\documentclass[12pt]{article}%
\usepackage{amsmath}%
\usepackage{tikz}
\usepackage{amsfonts}%
\usepackage{amssymb}%
\usepackage{amsmath}
\usepackage{marginnote}
\usepackage{bm}
\usepackage{bbm}
\usepackage{latexsym}
\usepackage{graphicx}
\usepackage{sgame}
\usepackage{pstricks}
\usepackage{pdfpages}
\usepackage[utf8]{inputenc}
\usepackage[T1]{fontenc}
\usepackage{natbib}
\usepackage{pgf,tikz}
\usepackage{comment}
\usepackage[title]{appendix}

\newtheorem{theorem}{Theorem}

\newtheorem{assumption}{Assumption}

\newtheorem{claim}[theorem]{Claim}

\newtheorem{corollary}{Corollary}

\newtheorem{lemma}{Lemma}

\newtheorem{proposition}{Proposition}

\newcommand{\calP}{{\cal P}}

\newcommand{\ep}{\varepsilon}
\newcommand{\dN}{{{\bf N}}}
\newcommand{\dR}{{{\bf R}}}
\newcommand{\E}{{{\bf E}}}

\newcommand{\prob}{{{\bf P}}}

\newcommand{\calC}{{\cal C}}

\newcommand{\calE}{\mathcal{E}}
\newcommand{\calF}{\mathcal{F}}

\newcommand{\calH}{\mathcal{H}}

\newcommand{\calU}{\mathcal{U}}

\usepackage{setspace}
\newenvironment{proof}[1][Proof]{\textbf{#1.} }{\ \rule{0.5em}{0.5em}}

\newcounter{figurecounter}
\begin{document}


\title{Stationary social learning in a changing environment}
\author{Rapha\"el L\'evy\thanks{HEC Paris.}, Marcin P\k{e}ski\thanks{University of Toronto}, Nicolas Vieille\thanks{HEC Paris.}}

\maketitle

\begin{abstract}
We consider social learning in a changing world. With changing states, societies can remain responsive only if agents regularly act upon fresh information, which drastically limits the value of observational learning. When the state is close to persistent, a consensus whereby most agents choose the same action typically emerges. However, the consensus action is not perfectly correlated with the state, because societies exhibit inertia following state changes. Phases of inertia may be longer when signals are more precise, even if agents draw large samples of past actions, as actions then become too correlated within samples, thereby reducing informativeness and welfare.
\end{abstract}

\section{Introduction}
The literature on social learning has extensively studied how agents can learn relevant information from observing others' actions, in a variety of situations. In particular, the literature has focused on two important related questions: first, when do informational cascades arise, i.e., agents simply imitate whatever their predecessor does, thereby ignoring their private information and  
freezing the learning dynamics; 
second, under what conditions does observational learning allow to aggregate information so that agents can safely herd to play the right action. The literature has been extremely fruitful in answering these questions. However, 
little attention has been drawn to the possibility that the underlying state of nature might change over time.\footnote{Notable exceptions are \cite{OMS, acemoglu, tamuz, Golub}. See below for a detailed account on how our work relates to these papers.} Nevertheless, in most applications where social learning matters, e.g., technology adoption, restaurants, investment decisions, the optimal course of action is likely to change over time, and it is important to assess the aggregative properties of social learning in such a context. 

In addition, the possibility of state changes raises new interesting questions both from applied and theoretical perspectives. 
For instance, the dynamics of learning sheds light on how societies react to changes in the environment and on the process by which a technology or dominant consensus may be replaced by a new one.\footnote{Examples of  change in the dominant technology abound, ranging from the ``war of the currents" in the late 19th century, the ``quartz crisis" in watchmaking in the 1970s to Facebook overtaking MySpace as the dominant social network.} In addition to technology adoption, a noteworthy application is social contagion in collective action.\footnote{It has been documented both theoretically and empirically that collective action works like a cascade \citep{granovetter,lohmann}. \cite{aidt}, using data from the Swing riots in 19th century England, empirically establish the key role played by local contagion in a context where access to large-scale news was essentially limited.} 
Indeed, mass political and social movements often feature a \textit{domino effect} whereby contagion rapidly accelerates, as was exemplified by the Arab Spring in the early 2010s. An interesting question that we explore in this paper is whether such domino effects can reflect the dynamics of social learning following unobservable changes in the environment. 
%
From a theoretical perspective, 
%
%
the possibility of state changes creates a tension between information aggregation and the need to be responsive to changes in the environment.
Suppose, for instance, that all agents at some point play the correct action. Knowing this, 
new agents should imitate the action they observe as long as the probability that the state changes is not too high, thereby acting as in an informational cascade. But such a cascade cannot last for too long for 
sooner or later, the state will change, and the information value of past actions will deplete. Therefore, efficiency requires society to be able not only to reach a consensus in which agents all take the same (correct) action, but also to collectively react to a potential change in the environment as swiftly as possible. Such reactivity imposes that public history be regularly replenished by fresh information coming from agents who act on their private signals. This essentially limits the extent of social learning, by capping how much agents can rely on others. The goal of this paper is to evaluate how this tension between information aggregation and the responsiveness of society to novelty shapes equilibrium welfare. 

\medskip

To do this, we develop a model with a continuum of agents where (a) the state of nature follows a Markov chain, (b) agents observe a finite sample of past actions, and (c) agents have access to a (possibly costly) signal that is informative about the current state. Though unnecessary for our results, the presence of costly information acquisition makes social learning even more desirable. Indeed, the potential welfare gains from social learning can then result not only from better-informed decisions but also from savings on information acquisition costs. Aggregation of information 
in the presence of costly information acquisition raises the classical Grossman-Stiglitz paradox: if information was perfectly (or sufficiently well) aggregated, then no one would have incentives to acquire information, and hence information could not be aggregated. As seen above, this logic is even reinforced in the presence of changing states because information must be acquired all along to ensure that society remains reactive to potential state changes. 
%
We investigate the interplay between information aggregation and information acquisition in steady-state equilibria where behavior 
does not depend on calendar time, and 
we show that introducing a positive probability that the state changes, no matter how small, has a dramatic impact on the extent of social learning. 

\medskip
We first analyze the simpler cases where agents sample one or two past actions for which we can fully characterize equilibrium behavior. At a steady-state equilibrium, agents acquire information with positive probability irrespective of their sample since otherwise public information would at some point become obsolete. More specifically, if agents 
were not to acquire information (or to ignore their private signals in case they come for free), we show that the dynamics of the population would be inexorably attracted to a consensus where all agents play the same action regardless of the state, so such a consensus could not be informative.
 As such, our model exhibits social learning in that past actions convey information, but they are not sufficiently informative to generically allow information aggregation, even when the state is arbitrarily persistent. 
This comes in strong contrast with \cite{BF}, who show that in a fixed-state model social learning eventually aggregates information as soon as agents sample two actions, even when signals have bounded strength. Not only is this result not robust to changing states, but in our case observing two actions does no better than observing only one action, and possibly does strictly worse. 
 %
%
\medskip

For larger sample sizes, we focus on the case where the state is almost fully persistent. We first characterize the aggregate behavior of the population and establish that a consensus typically emerges. That is, in any period, most agents are most likely to play the same action, but this consensus action changes over time as a result of the state dynamics. 
Transitions to a new consensus thus exhibit a \textit{snowball} effect: once the fraction of the population playing the correct action in the new state takes off, convergence to the new consensus is extremely swift. This does not necessarily imply that the equilibrium is efficient, however, because there is a social reaction lag: society initially remains stuck for some time in the old consensus before it swiftly transitions to the new one. The length of this phase of inertia, or alternatively, the correlation between the consensus action and the state, determines the efficiency of the equilibrium. In other words, social learning makes it possible to reach a consensus, but the key question of whether this consensus is sufficiently informative to allow future players to herd 
remains. 

\medskip

In a third stream of results, we investigate this question 
and outline a simple sufficient condition under which 
the consensus is not informative enough, so that the equilibrium welfare is no higher than if agents sampled only one single past action. 
%
%
For the consensus action to be sufficiently correlated with the current state,  convergence to the consensus should be sufficiently slow to allow the population to quickly swing to the opposite consensus, should the state change. This requires that actions within samples be not too correlated and that observing at least a dissenting action within one's sample is sufficiently likely. One noteworthy case where this systematically fails to occur is perfect signals. In such a case, as soon as a majority of the population plays one action, the only possible source of dissent comes from people who acquire information, but in the case of perfect signals, these agents observe the same (correct) signal and all play the same (consensual) action.
Thus, the resulting consensus cannot be sufficiently informative. 
When signals are perfect, or binary and close to perfect, the equilibrium welfare is then equal to the welfare of a single decision maker with no access to observational learning.  

However, when signals are sufficiently imprecise, it is possible to find an equilibrium in which agents who observe a unanimous sample herd and thus save on the cost of acquiring information. 
This implies that, holding the value of acquiring information fixed, an increase in the precision of private information may depress welfare. 
This result highlights the dual role played by information acquisition in counterbalancing the forces of imitation. On the one hand, once a consensus is in place, the arrival of fresh and accurate information helps remain reactive to potential changes; on the other hand, as the consensus is being built, the arrival of fresh and contrarian information is instrumental in maintaining enough diversity of actions, which prevents the population from being later stuck in some irreversible 
consensus. Consequently, private signals must be positively correlated with state
, but not too much. 
In a related vein, we show that welfare may be improved when agents observe actions from the further past. This tends to decrease the correlation of actions within samples, thereby fostering learning and hence opportunities for herding.  

\medskip

Overall, except in the special case in which signals are free and have unbounded strength, equilibrium welfare falls short of efficiency, even in the persistent limit, i.e., information is not fully aggregated. We show that this is an equilibrium feature. Indeed, a social planner with the power to dictate individual strategies could achieve a close to optimal steady-state welfare  where a negligible fraction of agents acquire information and the consensus action is (asymptotically) perfectly correlated with the state. Therefore, the possibility that the state changes exacerbates the wedge between the optimal 
and the equilibrium welfare. 

\medskip

Our paper relates to the literature on observational learning pioneered by \cite{banerjee} and \cite{BHW}, and later by \cite{SS}.\footnote{See also \cite{SSsurvey} for a brief survey.} We assume that agents draw a finite sample of past actions, as in \cite{SSsampling} and \cite{BF}, who identify conditions on signals and sampling that lead to asymptotic learning. Compared to these papers, we allow the state to change, which drastically reduces the efficiency of social learning. As an illustration, \cite{BF} show that there is asymptotic convergence to the correct action when agents sample two or more actions. 
By contrast, we show that welfare is typically bounded away from efficiency when the state may change, even with an arbitrarily small probability. As is the case in our paper, \cite{BV} and \cite{ali} consider a setup in which private signals are costly. Again, our focus on nonpersistent states shifts the focus away from the question of asymptotic learning that is central in their papers. 

In addition, our paper relates to a stream of papers that consider social learning in a changing world. \cite{OMS} show that cascades must then end in finite time, but arise for sure when the state is persistent enough. 
In a different line of research, \cite{acemoglu} and \cite{tamuz} consider non-Bayesian models in which agents use a linear updating rule. 
Closest to our approach, \cite{Golub} and \cite{Meg} also develop stationary analyses of social learning. \cite{Golub} consider stationary equilibria in a Gaussian environment where agents in a network learn from their neighbors. They show that learning is improved when agents have heterogeneous neighbors who have access to signals of different precision. While such heterogeneity 
is ruled out in our model where all agents are symmetric, we derive (counterintuitive) comparative statics results on the precision of signals that 
reflect the adverse welfare impact of an excessive correlation of actions, in line with their result. \cite{Meg} consider a Markovian environment where past actions may be misrecorded. 
Their main focus is on whether agents put too much or too little weight on their private information compared to what a planner would do, while we investigate how the level of equilibrium welfare varies with sample size and the precision of signals. 

\medskip

This paper is organized as follows. We outline the essential ingredients of the model in Section \ref{sec prelim}. In  Section \ref{small}, we characterize equilibria for small samples, and then turn to the equilibrium analysis for arbitrary samples in the persistent limit case in Section \ref{general}. Section \ref{planner} addresses the planner's problem and Section \ref{concl} concludes.

\section{The model}\label{sec prelim}

 \subsection{States and actions}
We consider a social learning model in discrete time with binary actions and binary states. In each period, a continuum of short-lived agents choose an action $a$ from the set $A=\{0,1\}.$ The optimal action depends on the state of nature $\theta\in\Theta=\{0,1\}$: an agent obtains a payoff of 1 when his action matches the current state ($a=\theta),$ and a payoff of 0 otherwise.

Successive states over time $(\theta_t)$ follow a Markov chain. For simplicity, we assume this chain to be symmetric: for all $\theta\in \Theta$, $\prob(\theta_{t+1}\neq \theta\mid \theta_t= \theta)= \lambda,$ where $\lambda$ captures persistence in the state. States are \emph{iid} if $\lambda= \frac12$, and fully persistent if $\lambda =0$. We assume $\lambda\in \left( 0,\frac12\right):$  the state is persistent, 
but not fully
.

\medskip

%



 \subsection{Timing}

The sequence of events unfolds as follows. In period $t,$ each new-born agent
\begin{enumerate}
\item observes (for free) a sample of $n$ actions drawn from the past,
\item decides whether to acquire additional information about $\theta_t$ at cost $c\geq 0,$
\item if information was acquired, observes a signal correlated with $\theta_t,$
\item picks an action $a\in\{0,1\}.$
\end{enumerate}

 \subsection{Sampling}

We assume that each action within one's sample is randomly picked from some random earlier period $t-\tau$ ($\tau\geq 1$).\footnote{For the purpose of the steady-state analysis, we adopt the convenient abstraction of a doubly infinite timeline $t\in \mathbf{Z}$.} 
The lag $\tau$ follows a geometric distribution with parameter $\rho>0,$ i.e., the probability that an action is picked from generation $t-\tau$ is 
$\rho(1-\rho)^{\tau-1}$ for each $\tau\geq 1$.\footnote{The only stationary recursive sampling process involves geometric weighting of the past, as argued by \cite{SSsampling}, so the focus on a geometric rate of decay is basically imposed in our steady-state analysis.} In the special case $\rho=1$, the sample is composed of actions played in the previous period only. 
We assume that  $\tau$ is not observed, i.e., the vintages of the sampled actions are unknown. The sample composition can then be simply summarized by the count of each action within the sample. Actions are sampled independently within a sample, and across agents. Finally, we assume proportional sampling in that the probability of sampling one action in a given pool corresponds to the prevalence of this action in the pool.\footnote{It would be possible to consider alternative sampling procedures, as \cite{BF} and \cite{Meg} for instance do. However, we stick to proportional sampling for simplicity.} 

\medskip

It is worth mentioning another narrative that leads to the same model. Under this alternative interpretation, agents are instead long-lived, and a fraction $\rho$ of the population is replaced by new agents in each period. Upon entering, each new agent observes a sample composed of $n$ actions chosen by agents in the current population, possibly acquires information on the current state, and then finally chooses an action. 

 \subsection{Private signals and information acquisition}\label{sec SE}

We turn to the description of the information that agents can acquire, and derive their optimal information acquisition strategies. 
Following common practice (see, for instance, \cite{SS}), instead of describing signals using their conditional distributions given $\theta,$ we take an isomorphic view and identify signals with the posterior beliefs that they induce under the 
uniform prior. 
Accordingly, we denote by $H_{\theta}$ the cdf of posterior beliefs when the state is $\theta.$ 
We assume the distributions of signals to be symmetric across states, that is, $H_0(q)=1-H_1(1-q).$ In addition, since signals are informative, $H_1(q)<H_0(q)$ for some $q.$ Finally, we denote by $\underline q$ and $\overline q=1-\underline q$ the infimum and supremum of the  support of the unconditional distribution $H:= \frac12 H_0+\frac12 H_1$.


Since the Markov chain $(\theta_t)_t$ is symmetric, the invariant distribution of $\theta$ puts equal probabilities on each state of the world. Therefore, an agent who observes no past action ($n=0$) and observes a signal $q$ has a posterior belief $q$. Instead, an agent with access to observational learning makes inferences based on his sample and holds an \emph{interim} belief, which we denote $p$. To be more specific, $p$ is the probability that the agent assigns to the current state being $\theta=1$ after observing his sample.



If an agent with interim beliefs $p$ draws a 'signal' $q$, his posterior belief becomes $$\frac{pq}{pq+\left(1-p\right)\left(1-q\right)}.$$ Such an agent picks action 1 if and only if
\begin{equation*}
\frac{pq}{pq+\left(1-p\right)\left(1-q\right)}\geq\frac12\Leftrightarrow q\geq1-p, 
\end{equation*}
with indifference when $q=1-p$.
Accordingly, the  probability of playing the right action $a=\theta$ when holding an interim belief $p$ and  acquiring information reads 
%
\begin{equation}
v(p)=p\left(1-H_1(1-p)\right)+(1-p)H_0(1-p)
\end{equation}
The function $v$ is convex, increasing on $[\frac12,1]$ and symmetric: $v(p)= v(1-p)$ for all $p.$
%
%
%
The expected value from acquiring information for an agent with an interim belief $p$ is thus $v(p)-c.$

If instead the agent chooses to not acquire information, he will choose the action that matches the most likely state, i.e., play $a=1$ if and only if
$\displaystyle p\geq\frac12,
$
and then obtain a payoff equal to $u(p):=\max(p,1-p).$ Of course, since more information cannot hurt, $v(p)\geq u(p)$ for all $p.$

\begin{assumption}\label{Hyp1} $v(\frac12)-c>\frac12.$   \end{assumption}
This assumption ensures that  agents facing  maximal uncertainty about the state (i.e., with beliefs $p=\frac12$) have an incentive to acquire information. If this was not the case, then no one would ever acquire information. Past actions would then be fully uninformative, thereby precluding social learning. Assumption \ref{Hyp1} is thus minimal to make the problem interesting, and we thus maintain it throughout the paper.

\smallskip

Because $v$ is convex, $u$ is piecewise linear, and $v(p)=u(p)$ for $p\in\left\{0,1\right\}$, Assumption \ref{Hyp1} implies that there exists a unique $\hat p\in[\frac12,1]$ such that $v(p)-c>u(p)\Leftrightarrow p\in(1-\hat p,\hat p).$ 



\paragraph{Remark} 

As soon as $c>0,$ one has $\hat p<1$ for any signal distribution. However, even for $c=0,$ it is not guaranteed that $\hat p=1.$ Indeed, if signals have bounded strength, i.e., if $0<\underline q<\overline q<1,$ then $v(p)=u(p)$ for any $p\notin \left[\underline{q},\bar q\right].$ Acquiring information at those beliefs has no value because the agent plays the same action regardless of the signal $q$ that he observes. This implies that $\hat p<1.$ However, if $[\underline q,\overline q]=[0,1],$ then signals have unbounded strength and one has $ v(p)>u(p)$ for all $p\in(0,1):$ regardless of the interim belief $p,$ there is always a positive probability to draw signals that lead the agent to overturn his action, so that information is strictly valuable. One thus has $\hat p=1$ if and only if signals are free and have unbounded strength.



\begin{figure}[htp]
   \begin{minipage}[c]{.46\linewidth}
      \includegraphics[scale=0.3]{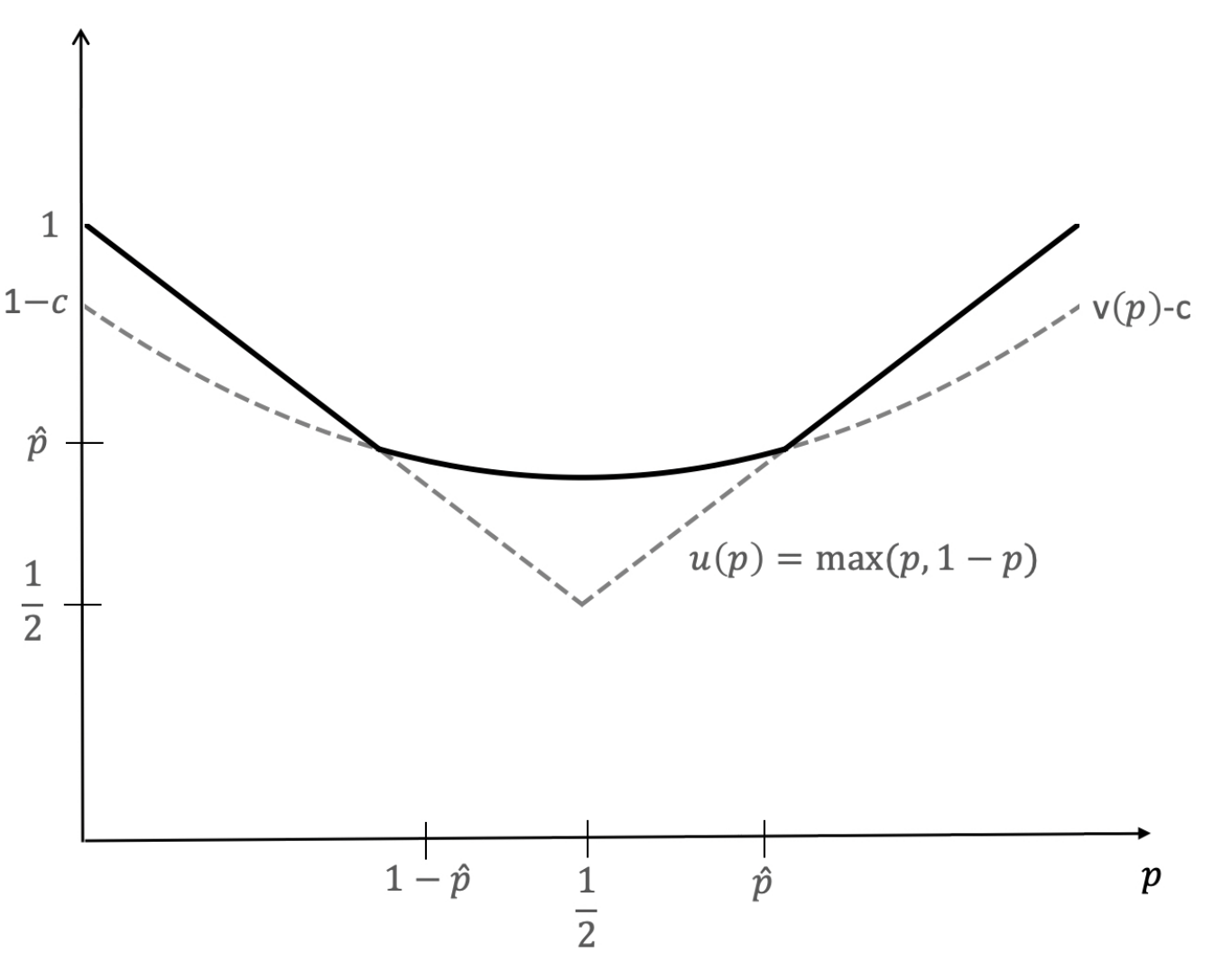}\caption{Signals with unbounded strength}
   \end{minipage} \hfill
   \begin{minipage}[c]{.46\linewidth}
      \includegraphics[scale=0.3]{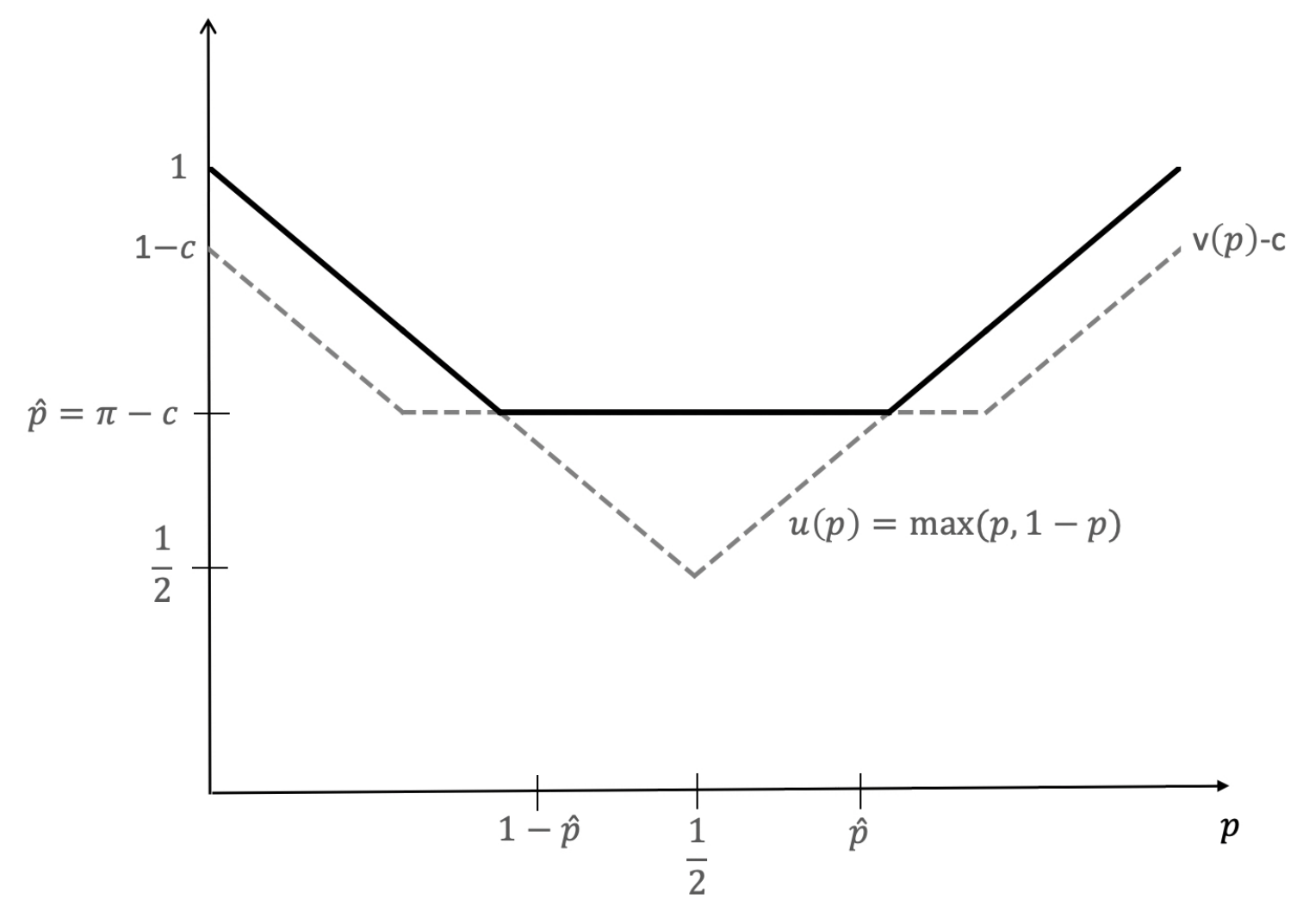}\caption{Binary signals with precision $\pi$}
   \end{minipage}
\end{figure}

\medskip

We depict in Figures 1 and 2 the typical shape of $u(p)$ and $v(p)-c$ 
when signals have unbounded strength (Figure 1) and when signals are binary with precision $\pi>\frac12$ (Figure 2). In the latter case, 
there are only two posterior beliefs $q\in\left\{1-\pi,\pi\right\}$ 
and $v(p)=\max(p,1-p,\pi).$ This implies that $\hat p=\pi=c.$ Notice that in the case of perfect signals $(\pi=1),$ one thus has $\hat p=1-c.$

\medskip

\subsection{Equilibrium concept} 
We focus on symmetric steady-state equilibria. Informally, in a steady-state equilibrium, agents in different periods make the same probabilistic inferences and follow the same strategy. Let $p_k$ denote the (time-independent) interim belief of an agent who draws a sample composed of $k$ ones and $n-k$ zeros. 
Upon observing their samples, agents play a strategy (i.e., an information acquisition strategy and an action plan) that is optimal given $p_k,$ and $p_k$ is derived using Bayes' rule from the invariant joint distribution of states and sample compositions. In addition, in a symmetric steady-state equilibrium, the strategy of an agent who observes a sample composed of $k$ ones is the mirror image of that of an agent who observes $n-k$ ones, that is, the equilibrium is unchanged when one relabels actions and states. 
%
While strategies and equilibrium steady states are formally defined and shown to exist in Section 
\ref{subsec:Strategies-and-Equilibrium}, such an informal definition is sufficient for the next section, where we derive equilibria for small samples ($n\leq2).$ 
\medskip

Before turning to this, let us remark that Assumption \ref{Hyp1} implies that, for any sample size, at least some agents must acquire information with positive probability at an equilibrium steady-state. This is formally stated in Lemma \ref{lemm1} below.

\begin{lemma}\label{lemm1}
In any steady-state equilibrium, there exists $k\leq n$ such that $p_k\in[1-\hat p,\hat p].$
\end{lemma}

\begin{proof} Suppose that $p_k\notin[1-\hat p,\hat p]$ for all $k.$ Since information is then never acquired, 
samples become uninformative at any steady-state, and beliefs 'drift back' to $p_k=\frac12$ for all $k$. A contradiction.\end{proof}
\medskip

Lemma \ref{lemm1} implies that information must at least be sometimes acquired. If signals are free and information is thus always acquired, then Lemma \ref{lemm1} implies that there is at least one sample composition $k$ such that information is valuable, i.e., ignoring one's signal is not strictly dominant. 

\section{Small samples}\label{small}

For small samples $(n\leq2)$, we are able to fully characterize equilibrium behavior. 
The case $n=0$ where agents do not sample any past action serves as a natural benchmark to assess the extent of social learning. When $n=0$,  agents have no feedback on previous actions, and hold an interim belief $\frac12$ at any steady-state. Assumption \ref{Hyp1} implies that they acquire information and thus obtain an expected payoff of $ v(\frac12)-c$. 

\medskip


\subsection{The case $n=1$}\label{sec n=1}
 
We assume here that agents sample only one past action. Proposition \ref{n1} below summarizes the main properties of the unique equilibrium. For expositional simplicity, the result is stated for the case $\rho=1,$ i.e., the action sampled in period $t$ was played in period $t-1$. 

\begin{proposition}\label{n1}
There is a unique symmetric equilibrium:
\begin{itemize}
\item If $\lambda\leq\frac{c}{2(\hat p+c)-1}\equiv \lambda_*,$  agents acquire information with probability  $\beta= \displaystyle \lambda \frac{2\hat p -1}{c(1-2\lambda)}\in(0,1)$.
\item If $\lambda>\lambda_*,$  agents acquire information with probability $1$.
\end{itemize}

\end{proposition}

\medskip

The intuition is clear. When states are close to \emph{iid}, past actions cannot possibly convey much information about the current state, and information must be acquired with probability 1. As the state becomes more persistent, past actions may potentially become more informative. Still, to provide an incentive for later generations to acquire information, the sampled action should not be too informative, and the probability of buying information decreases as $\lambda$ decreases. In the persistent limit $\lambda \to 0$, agents acquire information with vanishing probability, and hence most likely rely on observational learning to make their decisions. 

\medskip

\begin{proof} 
By symmetry, one has $p_1=1-p_0,$ with $p_1\in \left[1-\hat p,\hat p\right]$ by Lemma \ref{lemm1}.

Consider an agent acting in period $t$, who samples the action $a_{t-1}$ played by some agent $B$, and assume $a_{t-1}=1$. Note that $p_1$ obeys the following equation:\footnote{In a steady state equilibrium, calendar time plays no role ($p_1$ is the same across periods): we only use the subscripts $t$ and $t-1$ to distinguish between current and past actions or states.} 
\begin{equation}\label{eq1}
 p_1=\prob(\theta_t= 1\mid a_{t-1}=1) = (1-\lambda)\prob(\theta_{t-1}=1\mid a_{t-1}=1)+\lambda\prob(\theta_{t-1}=0\mid a_{t-1}=1),
 \end{equation}
where $\prob(\theta_{t-1}=1\mid a_{t-1}=1)$ is the probability that agent $B$ has played the correct action.

Since the equilibrium welfare  exceeds $\frac12$, the sampled action  $a_{t-1}$ is positively correlated with the previous state $\theta_{t-1}$. In addition, because states are persistent,  sampling an action $a_{t-1}=1$ provides positive evidence that the current state is $\theta_t=1$, hence $p_1\in \left[\frac12,\hat p\right]$.

Assume first that $p_1=\hat p$. Agents are then indifferent between acquiring information or not. Let $\beta$ denote the probability of acquiring information. By the steady-state property, any agent from the previous generation $t-1$ also either held an interim belief $p_0$ or $p_1,$ and hence acquired information with probability $\beta.$  Since the action $a_{t-1}$ of agent $B$ matched $\theta_{t-1}$ with probability $v(p_k)$ if $B$ did acquire information, and with probability $u(p_k)$ if he did not, the overall probability that $B$ chose the correct action when holding an interim belief $p_k$ is
$\beta v(p_k)+(1-\beta)u(p_k).$ Since $v(p_k)-c=u(p_k)=\hat p$ for all $k,$ one can rewrite
 \[ \prob(\theta_{t-1}= 1\mid a_{t-1} =1)= \hat p+\beta c\]
 and derive from (\ref{eq1}) that
 \begin{equation*}
 \hat p=(1-\lambda)\left(\hat p+\beta c\right)+\lambda\left(1-\hat p-\beta c\right) \end{equation*}
 We conclude
 \begin{equation}
\beta= \displaystyle \lambda \frac{2\hat p -1}{c(1-2\lambda)}
\end{equation}
Since $\beta\in [0,1]$, it must  be that $\lambda\leq\frac{c}{2(\hat p+c)-1}:=\lambda_*.$
 

\medskip

Assume now that $p_1<\hat p$, or equivalently, $ v(p_1)-c>u(p_1)=p_1.$ In that case, all agents acquire information with probability 1. 
Then, 
 \[ \prob(\theta_{t-1}= 1\mid a_{t-1} =1)=  v(p_1)\]
 
 and (\ref{eq1}) now reads

 \begin{equation}\label{eq1bis}
 p_1= (1-\lambda)v(p_1)+\lambda(1- v(p_1)).
 \end{equation}
 
Since $v'(p_1)<1$ for $p_1\leq\hat p$ and since $ v(\frac12)-c>\frac12,$  (\ref{eq1bis}) has a (unique) root in $[\frac12,\hat p)$ if and only if $\lambda> \lambda_*=\frac{c}{2(\hat p+c)-1}.$ \end{proof}
  
\paragraph{Welfare} 
As long as $\lambda\leq  \lambda_*,$ each agent is indifferent between acquiring information or not, that is, holds beliefs $\hat p$ or $1-\hat p.$ Thus, the equilibrium welfare is equal to $\hat p.$ For $\lambda> \lambda_*,$ the equilibrium welfare is given by $ v(p_1)-c,$ where $p_1$ is the solution of (\ref{eq1bis}). It is easy to see that welfare increases as the state becomes more persistent, i.e., as $\lambda$ decreases, and is equal to $\frac12$ in the \textit{iid} limit $(\lambda=\frac12).$ In the special case where signals are binary   with precision $\pi,$ since $ v(p)=\pi$ for all $p\in[1-\hat p,\hat p],$ welfare is equal to $\pi-c=\hat p$ for all values of $\lambda.$

 \paragraph{Remark 1}
If $\rho<1,$ $\lambda$ should be replaced  by the probability $\Lambda$ that the action sampled has been taken at a date when the state was different from that of today. Noting that the probability that the state that prevailed $m$ periods ago is different from the current state is equal to $\frac{1}{2}-\frac{1}{2}\left(1-2\lambda\right)^{m},$ $\Lambda$ is given by  
\begin{equation*}
\Lambda=\dot{\frac12\sum_{m\geq1}\rho\left(1-\rho\right)^{m-1}\left(1-\left(1-2\lambda\right)^{m}\right)}=\frac{\lambda}{1-\left(1-\rho\right)\left(1-2\lambda\right)}
\end{equation*}
%
$\Lambda$ is therefore the counterpart of $\lambda$ in the situation where actions may be sampled from the further past. Plainly, $\Lambda= \lambda$ if  $\rho=1$. 
In the limit  $\rho\to0,$ actions are drawn uniformly from the past, and $\Lambda\to \frac12.$ Finally, holding $\rho$ fixed, $\Lambda\to0$ as $\lambda\to0.$
Since $\Lambda$ decreases in $\rho,$ welfare is (weakly)  increasing in $\rho:$ observing more recent actions facilitates social learning, thereby improving welfare. However, for $\Lambda\leq\lambda_*,$ welfare is independent of $\rho.$\footnote{This is both reminiscent and in contrast with \cite{SSsampling}, who show that, among the sampling procedures in which each agent samples a single past action, the welfare-maximizing one consists of sampling the most recent action. In our case, welfare is maximized when $\rho=1,$ but not strictly so: because the state is changing, agents at equilibrium are indifferent between acquiring information or not. Accordingly, the equilibrium informational content of past actions does not strictly increase when sampled actions are more recent.}

 \medskip
 

 
\paragraph{Remark 2} While the equilibrium analysis is straightforward both in terms of welfare and equilibrium strategies, 
describing the aggregate behavior of the population proves extremely complex. Let $\chi_t$ denote the fraction of agents playing action 1 at date $t.$  For the sake of this discussion, consider the case of \emph{perfect} signals. 
A $t$-agent  chooses action 1 if either she acquires information and $\theta_t= 1$, or does not, and observes $a_{t-1}= 1$. 
Since there is a continuum of such agents, one has 
\[\chi_{t}= (1-\beta)\chi_{t-1} +\beta \mathbf{1}_{\theta_t=1}.\]


 
 There is extensive literature analyzing the dynamics of $(\chi_t)$ when successive states $(\theta_t)$ are \emph{iid}, with applications to stochastic growth models -- see \emph{e.g.} \cite{MMP} or \cite{BM}. 
 Under the \emph{iid} assumption, $(\chi_t)$ follows a Markov chain over $[0,1],$ and the (unique) invariant distribution of $(\chi_t)$ coincides with the distribution of the so-called random Erd\"os series whose properties are highly sensitive to $\beta$ \citep{erdos,solo,PS}. 
While the results from this literature do not apply to our setup since they assume \emph{iid} states, they strongly suggest that the random sequence $(\chi_t)$ is a highly complex object. Even worse, $(\chi_t) $ is no longer a Markov chain as soon as $\lambda<\frac12$, but a hidden Markov chain.  
 
\bigskip

 
 

\paragraph{Remark 3} Our focus on equilibrium steady states creates a general equilibrium effect. Since agents in every period should behave the same, actions cannot be too informative because otherwise future generations will stop acquiring information. When $n=1$, this effect is maximal because there is only one possible interim belief for future generations (up to symmetry), which, from Lemma \ref{lemm1}, lies in  $[1-\hat p,\hat p].$ However, when $n>1$, there is a wider scope of possible sample compositions. Welfare gains would be achieved if agents who acquire information with some sample $k$ generate enough information so that agents with a sample composition $k'$ can herd and  save on the cost of acquiring information -- that is, if there exists $(k,k')$ such that $p_k\leq\hat p<p_{k'}.$ 
Ultimately, the efficiency of social learning reflects the magnitude of such information externalities across samples.

 \subsection{The case $n=2$}

In the case $n=2,$ agents who sample two different actions are confused. Their interim belief is $p_1=\frac12$ and they accordingly acquire information. In line with the previous remark, the efficiency of social learning then depends on whether the information produced by these agents is sufficient to ensure that 
agents who instead observe unanimous samples can herd. In other words, the key question is whether $p_2>\hat p.$ Casual intuition suggests that this should be the case when the state is sufficiently persistent.
 %
However, this intuition is incorrect, as we now show.  

 
 
  \begin{proposition}\label{th n=2}
 In any equilibrium, one has  $p_k\in[1-\hat p,\hat p]$ for all $k=0,1,2.$
 \end{proposition}

For $c>0$, this implies that agents are always willing to acquire information, irrespective of their sample. 
If signals are free ($c=0),$ Proposition \ref{th n=2} implies that it can never be strictly optimal to ignore one's signal. The main intuition is as follows. Assume agents with a unanimous sample herd for sure at equilibrium. If the population is imbalanced in some period $t$, in that most agents play the same action, then an agent  acting in period $t+1$ will most likely observe a unanimous sample and herd. It turns out that 
the information generated by the agents who draw balanced samples at $t+1$ is then insufficient to allow the population to adjust to a potential state change, and  the imbalance is reinforced on average. In other words, the forces of imitation are too strong, and the population is inexorably attracted into a consensus where it can no longer be responsive to changes in the environment. 

The detailed proof of Proposition \ref{th n=2} is in the Appendix. We provide a sketch in the simpler case where $\rho=1$ below.
 
 \medskip
\begin{proof}[Proof Sketch] 
By symmetry, agents with a balanced sample have an interim belief $p_1=\frac12$ and thus acquire information, from Assumption 1. Let $\phi_\theta:=1-H_\theta(\frac12)$ denote the probability that such agents  end up playing action 1 in state $\theta.$ 
For the sake of contradiction, we assume that agents drawing a unanimous sample herd for sure (that is, $p_2>\hat p).$ Then the probability $\chi_t$ that a generic $t$-agent plays action 1 is 
\[\chi_{t-1}^2+2\chi_{t-1}(1-\chi_{t-1})\phi_{\theta_t} .\]
Since there is a continuum of agents, 
$\chi_t$ 
is given by $\chi_t= g_{\theta_t}(\chi_{t-1})$, where
\[g_\theta(\chi)= \chi^2 + 2\chi(1-\chi) \phi_\theta.\]
Because $\phi_1>\frac12>\phi_0$, $g_1(\chi) >\chi > g_0(\chi)$ for each $\chi\in (0,1)$: the popularity of action 1, $\chi_t,$ increases over time as long as $\theta_t= 1$, and decreases otherwise, as shown in the following graph.

\begin{center}%
\begin{tikzpicture}[scale=0.6]%
%
%
\draw [->] (0,0) -- (11,0) node[at end, right] {$\chi_{t-1}$}; \draw
(-0.2,-0.2) node{$0$};
\draw [->] (0,0) -- (0,11) node [at end, left] {$\chi_t$};%
\draw [-] (-0.1,10)--(0.1,10) node [at end, left] {$1$};%

\draw [-] (10,0.1)--(10,-0.1) node [at end, below] {$1$};%

 \draw plot [smooth] coordinates {(0,0) (1,1.45)(2,2.8) (3,4.05) (4,5.2) (5,6.25) (6,7.2) (6.25,7.421) (7,8.05) (7.421,8.37794) (8,8.8) (9,9.45) (10,10)};
%
\draw (5,7) node [above]{{\footnotesize $g_1(\chi_{t-1})$}};
\draw (5,3) node [below]{{\footnotesize $g_0(\chi_{t-1})$}};

 \draw plot [smooth] coordinates {(0,0)(1,0.55) (2,1.2) (3,1.95) (3.5,2.362)(4,2.8) (4.5,3.2625) (5,3.75) (5.72681,4.5) (6.81254,5.72681)(7.69846,6.81254) (8,7.2)  (8.37794,7.69846) (9,8.55)(10,10)}; 

\draw [smooth] (0,0) -- (10,10);
\draw [->,dotted,red] (5, 5) -- (5,6.25) ;
\draw [->,dotted] (5,6.25) -- (6.25,6.25) ;

\draw [->,dotted,red] (6.25,6.25)-- (6.25,7.421) ;

\draw [->,dotted] (6.25,7.421) --(7.421,7.421) ;

\draw [->,dotted,red] (7.421,7.421) -- (7.421,8.37794);

\draw [->,dotted] (7.421,8.37794)-- (8.37794,8.37794);

\draw [->,dotted,cyan] (8.37794,8.37794)-- (8.37794,7.69846);

\draw [->,dotted] (8.37794,7.69846)-- (7.69846,7.69846);

\draw [->,dotted,cyan] (7.69846,7.69846)--(7.69846,6.81254);

\draw [->,dotted] (7.69846,6.81254)--(6.81254,6.81254);

\draw [->,dotted,cyan] (6.81254,6.81254)--(6.81254,5.72681);

\draw [->,dotted] (6.81254,5.72681)--(5.72681,5.72681);

\draw [->,dotted,cyan] (5.72681,5.72681)--(5.72681,4.5);

\draw [->,dotted] (5.72681,4.5)--(4.5,4.5);

\draw [->,dotted,cyan] (4.5,4.5)--(4.5,3.2625);

\draw [->,dotted] (4.5,3.2625)--(3.2625,3.2625);

\draw [->,thick] (8.37794,6)--(8.37794,7);

\draw (8.37794,6) node [below] {{\tiny State change}};

\end{tikzpicture}%
\end{center}%


For $\chi$ close to  $0$, the ratio $g_\theta(\chi)/\chi$ is approximately equal to $2\phi_\theta$, hence $\ln \chi$ increases or decreases by a fixed amount, $\ln 2\phi_\theta$ in every period. Thus, as soon as $\chi_t$ is close to zero, the stochastic process $(\ln \chi_t)_t$ approximately follows a random walk, except that the increments $\ln 2\phi_{\theta_t}$ of the walk are not \emph{iid} over time, but follow a symmetric Markov chain.  

Using $4\phi_0\phi_1=4\phi_1(1-\phi_1)<1$, we derive that $\ln 2\phi_1 <-\ln 2\phi_0$, that is, step sizes are \emph{higher} (in absolute value) when $\ln \chi_t$ \emph{decreases}. This means that the random walk has a downward drift. Consequently, if $\chi_t\leq \ep$ for some $t$, there is a positive probability, bounded away from zero, that the popularity $\chi$ will never exceed $\ep$ after $t$. If instead $\chi_t\in (\ep,1-\ep)$ for some $t$,  $\chi$ eventually reaches $[0,\ep]\cup [1-\ep,1]$ with probability 1. This implies that $\chi_t\to \{0,1\}$, almost surely: at the steady state, all agents play the same action ($\chi\in \{0,1\}$). Since the state of nature keeps changing over time, this implies in turn that  samples are uninformative at the steady state. However, in that case, all agents would rather acquire information, even with a unanimous sample. A contradiction. 
\end{proof}

\medskip

At equilibrium, agents with a unanimous sample thus acquire information with positive probability $\beta>0$. The value of $\beta$ is given by an equilibrium condition. Intuitively, $\beta$ should be high enough so that the population does not become trapped in some irreversible consensus, and low enough so that future generations still have an incentive to acquire information. 
However, unlike for $n=1$, there is no explicit formula for $\beta$.

\medskip
  
\paragraph{Welfare} 
From Proposition \ref{th n=2}, $p_2\leq\hat p.$ At equilibrium, interim beliefs are either $p_1=\frac12$, $p_2\in \left[\frac12, \hat p\right]$, or $p_0=1-p_2$. Hence, the equilibrium welfare is a convex combination of $v(\frac12)-c$ and of $v(p_2)-c.$ If signals are binary with precision $\pi,$ then $v(p)=\pi$ for all $p\in(1-\hat p,\hat p),$ and welfare is equal to $\pi-c=\hat p.$ If, however, signals are nonbinary, then the expected welfare is strictly smaller than $\hat p.$ 

\smallskip

The result below summarizes our findings on the equilibrium welfare for $n=0,1,2.$


\begin{corollary}\label{comparison}
With binary signals, the equilibrium welfare is equal to $\hat p$ for all $n\leq2$ and $\lambda >0$. 
With nonbinary signals, the equilibrium welfare if $n=2$ is strictly larger than if $n=0.$ However, it is strictly lower than if $n=1$ as soon as $\lambda\leq\lambda_*.$
\end{corollary}

 

 
Corollary \ref{comparison} highlights two striking results on welfare. If signals are binary, the welfare is the same when observing a sample of size 1 or 2 as when there is no observational learning ($n=0).$ In other words, while there is observational learning at equilibrium (some agents do herd), it does not generate any welfare gain compared to the situation with no opportunities for social learning. Instead, when signals are nonbinary, equilibrium welfare is nonmonotonic in the sample size, at least for small $\lambda$ and $n.$ 
This is all the more surprising as such a result holds when the state is arbitrarily close to persistent, that is, when the 
logic of observational learning acts most forcefully. 

\subsection{Comparison to fixed state models}



To assess the consequences on social learning of having an evolving state, it is instructive to confront our results to 
the literature on social learning with a fixed state. The most relevant benchmark here is \cite{BF}, who study the long-run dynamics of behavior of a population with a fixed state and costless signals. Essentially their model coincides with ours when $\lambda=0$ and $c=0.$ Since Propositions \ref{n1} and \ref{th n=2} hold whether information acquisition is costly or not, the key difference between our analysis and \citeauthor{BF}'s lies in the possibility that the state changes. To make the comparison as transparent as possible, let us then consider our results in the case where signals are free ($c=0$). 

In the case $n=1,$ $\lambda_*=0$ when $c=0,$ so $p_1$ is given by the solution of \eqref{eq1bis}. One remarks that $p_1\to\hat p$ as $\lambda$ goes to $0.$ Accordingly, if signals have unbounded strength, that is, $\hat p=1,$ the unique steady-state corresponds to an efficient cascade where all agents play the correct action. This result is the same as in \cite{BF}, and is fully in line with  \cite{SS}. In the case where $\hat p<1$ (signals have bounded strength), our result slightly contrasts with \cite{BF}, who show that there is a continuum of equilibrium steady-states -- namely, any $p_1\geq\hat p$ is an equilibrium. In any such equilibrium, there is a fraction $p_1$ of agents that play the correct action in each period, and newborn agents mimic the action that they sample, so that the fraction of agents playing the correct action remains constant. For $p_1\geq\hat p,$ signals are uninformative at these beliefs, so herding is fully rational -- as then $v(p_1)=u(p_1).$ In the parlance of the canonical model, such an equilibrium corresponds to a cascade in which each agent essentially ignores his private information and mimics whatever the predecessor he samples does.\footnote{Notice that since there is a continuum of agents, not all agents play the same action, even in a cascade.} Such equilibrium multiplicity is impossible as soon as the state changes ($\lambda>0)$ because it cannot be that agents systematically ignore their signals. In other words, Lemma \ref{lemm1} then applies, imposing 
$p_1\leq\hat p,$ which reduces the set of equilibria to a unique one where $p_1=\hat p.$ 

%

\smallskip
 
The main contrast with \cite{BF} arises when $n=2.$ While \citeauthor{BF} establish that, under a minimal assumption on the strength of signals, there is asymptotically complete learning, we establish that information aggregation with a state close to persistent only arises when $\hat p=1,$ that is, when signals have unbounded strength.\footnote{Thus, with unbounded signals, a changing state impairs information aggregation neither when $n=1$ nor $n=2.$ As we will see later, this result will also hold for all $n.$} However, as soon as signals have bounded strength ($\hat p<1),$ welfare is bounded away from 1, even when $\lambda$ is arbitrarily close to 0. 
%
Accordingly, introducing even a small probability that the state changes has a strong (adverse) impact on the ability to learn from others. 
Another implication is that observing more actions compensates for limited signal quality with a fixed state, but exacerbates the inefficiency with a changing state.

\section{The general case: equilibrium analysis}\label{general}
The result of Corollary \ref{comparison} that welfare is generically larger for $n=1$ than for $n=2$ illustrates a key difference between the case $n=1$ and the general case $n\geq2$ that we analyze in this Section. 
%
When $n=1$, the interim belief at date $t$ of an agent who observes $a_{t-1}$ is derived from the joint distribution of $\theta_{t-1}$ and $\chi_{t-1}$ but involves only the \emph{expected} value of $\chi_{t-1}$ in each state. That is, interim beliefs reflect how often \textit{on average} a random agent from the past plays the right action. 
However, as soon as $n\geq 2$, 
the conditional distribution of $\theta_{t-1}$ given a sample $k$ involves the correlation of actions within a typical sample. 
In particular, when past actions become more correlated with the state, they also become more correlated among each other, which reduces the informational content of the samples. 

Formally, the likelihood ratio of the belief assigned to $\theta_{t-1}=1$ is given by (in the case $\rho=1)$
\begin{equation}\label{form LR}
 \frac{\int_0^1 \chi^k(1-\chi)^{n-k}d\mu_1(\chi)}{\int_0^1 \chi^k(1-\chi)^{n-k}d\mu_0(\chi)}
\end{equation}
where $\mu_\theta$ is the distribution of $\chi_{t-1}$ if $\theta_{t-1}= \theta.$ 
Such a belief involves all $k$-th moments of $\chi_{t-1}$ for $k\leq n;$ in addition, the evolution of $\chi^k$ over time involves even higher powers of $\chi$, as can be checked.

Consequently, the steady-state equilibrium equations involve the entire joint distribution of $(\theta_t,\chi_t)$. 
As we have seen in Section \ref{sec n=1} (Remark 2), this distribution is a highly complex object, which leaves little hope to be able to fully describe equilibrium steady-states in general. 
This is why we focus on the persistent limit case ($\lambda \to 0)$, which we view as the most interesting case, and for which we are still able to derive significant results.

\medskip

This section is organized as follows. We start with a formal definition of strategies and equilibrium steady-states, and establish equilibrium existence. We then state our main results in the persistent limit, along two lines. We first establish a 
general result on the distribution of sample compositions, hence on the aggregate behavior of the population. We next turn to results on the correlation between states and samples, i.e., on the equilibrium welfare. 


\subsection{Strategies and equilibrium\label{subsec:Strategies-and-Equilibrium}}

A strategy specifies whether to buy information as a
function of one's sample, and a lottery over actions as a function of 
one's sample and, if relevant, the acquired signal. For conciseness, we identify the composition of a sample  with the count $k$ of ones within the sample. A strategy is thus a pair $\sigma=(\beta,\alpha)$
of (measurable) maps, with $\beta:\{0,\ldots,n\}\to [0,1]$ and $\alpha:\{0,\ldots,n\}\times[0,1]\to\Delta\left(\{0,1\}\right)$, with the understanding that $\beta(k)$ is the probability of acquiring information upon observing sample $k,$ and $\alpha(k,q)$ is the probability of playing action $1$ upon observing sample $k$ and drawing a signal $q\in [0,1]$. Note that not acquiring information is equivalent to drawing a signal $q=\frac12$ with probability 1, so an agent with sample $k$ who does not acquire information plays 1 with probability $\alpha(k,\frac12).$ 
%
%
Therefore, an agent who observes a sample $k$ 
plays action 1 in state $\theta$ with probability 
\begin{equation}
\phi_{\theta}\left(k\right):=\beta(k)\int_0^1\alpha(k,q)dH_{\theta}(q)+(1-\beta(k))\alpha(k,\frac12).\label{eq:elpha}
\end{equation}

To express interim beliefs at equilibrium, we need to determine the prevalence of ones in the pool of actions from which one samples. If $\rho=1,$ this pool only consists of actions from the previous period, so this prevalence coincides with the fraction of agents who played 1 in the past period, namely $\chi_{t-1}.$ However, in the general case $\rho\leq 1,$  sampled actions may be older, and the prevalence corresponds to a weighted average of the popularities of action 1 in all past periods (see Eq. \eqref{motionX} below). Therefore, the pair $(\theta_t,\chi_t)$ does not have a Markovian structure if $\rho<1$, and hence is no suitable state variable for our analysis. The relevant state variable is actually the pair $(\theta_t,x_t),$ where $x_t$ is the probability that a given action sampled at date $t+1$ is one (that is, $x_t$ is the prevalence of ones in the pool from which one samples at date $t+1$), which instead does have a Markovian structure.\footnote{We use the convention that $x_t$ corresponds to the prevalence at $t+1$ to make sure that $x_t=\chi_t$ when $\rho=1.$ This ensures consistency with the discussion of the case $n=2,$ where we sketched the proof in the case $\rho=1$.}

Notice that 
\begin{equation}\label{motionX}
x_{t}=\displaystyle \sum_{m\geq 1}\rho(1-\rho)^{m-1}\chi_{t+1-m}= (1-\rho) x_{t-1} +\rho \chi_t
\end{equation}

Since the sample composition at date $t$ follows a Binomial distribution with parameters $n$ and $x_{t-1}$ ($k\sim B(n,x_{t-1})$), it follows, taking expectations, that the fraction of agents choosing action 1 in period $t$ reads
\begin{equation}\label{zx}
 \chi_{t}=  \sum_{k=0}^{n}{n \choose k}x_{t-1}^{k}(1-x_{t-1})^{n-k}\phi_{\theta_{t}}(k).
\end{equation}

Combining \eqref{motionX} and \eqref{zx}, we derive
\begin{equation}\label{motionX1}
x_{t}=g_{\theta_{t}}\left(x_{t-1}\right),
\end{equation}
\begin{comment}
 \begin{eqnarray}
x_{t} & =(1-\rho)x_{t-1}+\rho \chi_{t} \nonumber\\
& =:g_{\theta_{t}}\left(x_{t-1}\right).\label{motionX1}
\end{eqnarray}
\end{comment}
where 
\begin{equation}\label{defg}
g_\theta(x)= (1-\rho) x +\rho\sum_{k=0}^{n}{n \choose k}x^{k}(1-x)^{n-k}\phi_{\theta}(k).
\end{equation}

Thus, given a strategy $\sigma$, $x_{t}$ is a deterministic function of $x_{t-1}$ and of $\theta_{t}$.


\begin{comment}
More generally, however, when $\rho<1$, we need to distinguish the two variables. 
Specifically, we have then

\begin{equation}
x_{t}=\displaystyle \sum_{m\geq 0}\rho(1-\rho)^{m}\chi_{t-m}= (1-\rho) x_{t-1} +\rho \chi_t\end{equation}

Since the sample at date $t+1$ $k\sim B(n,x_t)$, it follows taking expectations that
 the fraction of agents choosing action 1 in period $t+1$ is 
\[ \chi_{t+1}=  \sum_{k=0}^{n}{n \choose k}x_t^{k}(1-x_t)^{n-k}\phi_{\theta_{t+1}}(k), \]
and therefore, 
\begin{align*}
x_{t+1} & =(1-\rho)x_{t}+\rho \chi_{t+1} \\
& =:g_{\theta_{t+1}}\left(x_{t}\right).
\end{align*}
where 
\[g_\theta(x)= (1-\rho) x +\rho\sum_{k=0}^{n}{n \choose k}x^{k}(1-x)^{n-k}\phi_{\theta}(k) .\]

That is, $x_{t+1}$ is a deterministic function of $x_t$ and of $\theta_{t+1}$. 
\end{comment}
\medskip 




We can now formally define an equilibrium. An \emph{equilibrium steady-state} (ESS) is a pair $(\mu,\sigma)$ where
$\mu\in\Delta(\Theta \times[0,1])$ is an invariant measure for the
Markov chain $(\theta_t,x_t)$ induced by $\sigma$,
and $\sigma$ is optimal given $\mu$. 

The optimality condition on $\sigma$ reads 
\begin{description}
\item[C1] $\beta(k)=1$ if $p_k\in\left(1-\hat p,\hat p\right)$
and $\beta(k)=0$ if $p_k\notin\left[1-\hat p,\hat p\right]$.
\item[C2] $\alpha(k,q)=1$ if $q>1-p_k$ and $\alpha(k,q)=0$ if $q<1-p_k$ 
\end{description}
where 
$p_k=\prob(\theta_t=1\mid k)$ is the belief over the current state when drawing sample $k.$ 
%
Formally, one has $p_k=(1- \lambda) \prob(\theta_{t-1}=1\mid k) +\lambda \prob(\theta_{t-1}=0\mid k)$, where
\begin{equation}
\frac{ \prob(\theta_{t-1}=1\mid k)}{ \prob(\theta_{t-1}=0\mid k)}= \frac{ \int_0^1 x^k(1-x)^{n-k}d\mu(1,x)}{\int_0^1 x^k(1-x)^{n-k}d\mu(0,x)}.
\end{equation}

The invariance equation for $\mu$  reads
 \begin{description}
 \item[C3] $\mu(\theta,X)=\left(1-\lambda\right)\mu\left(\theta,g_{\theta}^{-1}\left(X\right)\right)+\lambda\mu\left(1-\theta,g_{\theta}^{-1}\left(X\right)\right)$ for any measurable $X\subset[0,1].$
 \end{description}

We focus on symmetric equilibria and  require in addition that $\mu$ and $\sigma$ treat the two states and actions symmetrically. Formally:
 \begin{description}
 \item[C4] $\beta(k)=\beta(n-k)$ and $\alpha(k,q)= 1-\alpha(n-k,1-q)$ for each $k$ and $q$.
  \item[C5] $\mu$ is invariant under the transformation $(\theta,x)\mapsto (1-\theta,1-x)$. 
 \end{description}

Let $G(\lambda,\rho)$ denote the game. The equilibrium welfare coincides with the expected payoff of a typical agent, obtained by subtracting the cost of information acquisition from the expected probability of matching the state. 
The former is given by
\[c \int_{\left[0,1\right]}\sum_{k}{n\choose k} x_{t}^{k}\left(1-x_{t}\right)^{n-k}\beta(k) d\mu(\Theta,x_t).\]
The latter is given by 

\begin{align*}
\intop_{\Theta\times\left[0,1\right]}\sum_{k}{n\choose k} x_{t}^{k}\left(1-x_{t}\right)^{n-k}\left[\mathbf{1}_{\theta_{t+1}=1}\phi_{\theta_{t+1}}\left(k\right)+\mathbf{1}_{\theta_{t+1}=0}\left(1-\phi_{\theta_{t+1}}\left(k\right)\right)\right]d\tilde\mu\left(\theta_{t+1},x_{t}\right),
\end{align*}
where $\tilde\mu\left(\theta_{t+1},x_{t}\right)=(1-\lambda)\mu\left(\theta_{t},x_{t}\right)+\lambda \mu\left(1-\theta_{t},x_{t}\right)$
is the invariant measure for $x_t$ and the subsequent state $\theta_{t+1}.$





\bigskip

Before exploring equilibrium behavior, we first establish equilibrium existence. 
 \begin{theorem}\label{thm existence}
There exists a symmetric equilibrium steady state. 
\end{theorem}

The proof (in the Appendix) uses a standard fixed-point argument. 

\subsection{Aggregate behavior: a consensus result}

We first derive a general result regarding the aggregate behavior in the population. It asserts that, at any ESS, the distribution of $x$ becomes arbitrarily concentrated around 0 and 1 in the persistent limit $(\lambda\to0).$

\begin{theorem}\label{th consensus}
Let $n\geq 2$ 
and let $(\mu_\lambda, \sigma_\lambda)$ be any ESS of $G(\lambda,\rho)$. Then, as $\lambda\to 0$, the marginal of $\mu_\lambda$ over  $x\in [0,1]$  converges to the uniform distribution over the two-point set $\{0,1\}$.\footnote{Limits are understood in the sense of weak convergence of probability measures over $[0,1].$}
\end{theorem}

When the state is close to fully persistent, $x_t$ is then close to 0 or 1. That is, in a typical period, an agent most likely observes a unanimous sample composed only of ones or zeros.  Actions within one's sample thus tend to be highly correlated. This directly implies that 
$\chi_t$ is also close to 0 or 1.
In other words, choices in a given period are highly correlated within the population, and the correlation becomes perfect in the limit $\lambda\to0$: society asymptotically achieves a consensus. 
However, because the state changes, this approximate consensus is not permanent. That is, the consensus evolves as the result of the state dynamics.


\medskip
Theorem \ref{th consensus} follows immediately from the more general Theorem \ref{th consensus general} below.

\begin{theorem}
\label{th consensus general} Let $n\geq2.$ There exists a constant $K<\infty$ such that 
\[
\int_{\Theta\times [0,1]} x\left(1-x\right)d\mu\left(\theta,x\right)\leq K\Lambda.
\]
for every $\lambda,\rho>0$ and every ESS $(\mu,\sigma)$ of $G(\lambda,\rho).$
\end{theorem}


Recall that $\Lambda=\frac{\lambda}{1-\left(1-\rho\right)\left(1-2\lambda\right)}$ 
is the probability that a typical sampled action has been taken under a different state than that of today. 
Theorem \ref{th consensus general} thus generically provides an upper bound on the probability that two actions within a sample differ.




\medskip\noindent

\begin{proof}[Intuition and Proof outline] The intuition for Theorems \ref{th consensus} and \ref{th consensus general} is  simple.  Fix an equilibrium $(\mu,\sigma)$ and denote $\kappa^*$ the (steady-state) fraction of agents whose action matches the state.  $\kappa^*$ differs from the equilibrium payoff $w^*$ in that information costs are not accounted for.

Consider a generic agent in period $t$, and assume that she observes her sample in sequence,  $a ^{(1)},\ldots,a ^{(n)}$. One strategy $\sigma_1$ available to her is to simply replicate the first action in the sample, $a ^{(1)}$. For concreteness, assume that $a ^{(1)}=1$. This strategy would yield a payoff of $\kappa^*$ if the state were to be invariant. With a changing state, it yields \[w(\sigma_1):= \prob(\theta_t=1\mid a^{(1)} =1)\geq \kappa^*-\Lambda.\]

Since $\kappa^* \geq w^*$, and since no strategy yields more than the equilibrium payoff $w^*$, this implies that the marginal gain in observing the second action $a^{(2)}$ is at most $\Lambda$. 

In turn, this implies that $a ^{(1)}$ and $a ^{(2)}$ coincide with high probability when $\Lambda$ is small, as we now argue.
Indeed, consider an alternative strategy $\sigma_2$ consisting of copying $a ^{(1)}$ if the second sampled action confirms the first action  ($a^{(1)} = a^{(2)} $) and acquiring information otherwise ($a^{(1)} \neq a^{(2)} $). Since the  belief of the agent is $1/2$ when observing $a ^{(1)}\neq a ^{(2)}$, her conditional payoff is $v(1/2)-c$ in that case. Therefore, the agent's payoff $w(\sigma_2)$ under the alternative strategy $\sigma_2$ is a convex combination of $v(1/2)-c$ and of $\prob(\theta_t=1\mid a ^{(1)}= a ^{(1)}=1)$, where the weights are the  conditional probabilities of $a ^{(2)}=0$ and of $a ^{(2)}=1$ given $a ^{(1)}=1$.

On the other hand, the martingale property of beliefs ensures that $w(\sigma_1)= \prob(\theta_t=1\mid a^{(1)} =1)$ is a convex combination of $1/2$ and of $\prob(\theta_t=1\mid a ^{(1)}= a ^{(1)}=1)$, with the \emph{same} weights. Since $v(1/2)-c > 1/2$, and since $w(\sigma_2)\leq\kappa^*\leq w(\sigma_1)+\Lambda$, it follows that the probability $\prob(a ^{(2)}=0\mid a ^{(1)}=1)$ that the second action will contradict the first action is at most of the order of $\Lambda$.

To conclude, recall that $a ^{(1)}$ and $a ^{(2)}$ are independent draws from a Bernoulli distribution with parameter $x $, where $x $ is first drawn according to $\mu$. 
Since $a ^{(1)}$ and $a ^{(2)}$ coincide with high probability, it must be that $x $ is quite close to 0 or to 1, with high $\mu$-probability.
\end{proof}

\bigskip

Theorems \ref{th consensus} and \ref{th consensus general} describe not only the aggregate behavior of the population in the persistent limit -- society almost always achieves a consensus -- but also its dynamics following a change in the environment. For instance, suppose that the state is $\theta=0$ and almost all agents play $a=0.$ Consider what happens when $\theta$ changes to 1. The transition to the new consensus is characterized by different phases. In a first phase, the old consensus persists despite the state change. Indeed, most agents observe a unanimous sample and most likely herd. This inertia makes society little reactive. At some point, however, sufficiently many agents will realize that the state is likely to have changed, and the consensus breaches. 
Once the fraction of agents playing 1 takes off, the transition to the new consensus where almost everyone plays 1 is 
extremely swift. In other words, there is a 
\textit{domino} effect whereby the popularity of action 1 snowballs.\footnote{Of course, it possible that the state reverts to 0 before the new consensus is reached, but this typically happens with negligible probability when $\lambda$ is small.} Notice, however, that the overall transition can be long, because society can be stuck for some time in the phase of inertia. The length of the phase of inertia determines the correlation between the consensus action and the state and, ultimately, the equilibrium welfare. 

\subsection{Equilibrium welfare}\label{sec welfare}

Theorem \ref{th consensus} establishes that, when $\lambda$ is small, agents most likely observe a unanimous sample. In addition to the insights that it provides on the aggregate behavior of the population, this result is useful to characterize welfare in the persistent limit. Indeed, given that almost all agents observe a unanimous sample, the ex ante welfare coincides with the expected payoff of an agent observing such a (unanimous) sample. Thus, there are two possibilities. Either such an agent does not acquire information and thus obtains a welfare given by $p_0=p_n>\hat p,$ or he does, and his welfare is $\hat p.$\footnote{Implicit here is that an agent who acquires information does it with a probability less than 1. As seen formally in the case $n=1,$ it is straightforward to see that there cannot be an equilibrium in which almost all agents acquire information with probability 1 when $\lambda$ is close to 0.}  
However, as illustrated by the case $n=2,$ whenever agents observing a unanimous sample do not acquire information and simply herd on the unanimous action, there is a risk that the society collectively gets trapped in some irreversible  -- and hence uninformative -- consensus, thereby precluding social learning. In Theorem \ref{th learning sym}, we 
provide a sufficient condition under which 
such dynamics also arise 
in the general case. 
%
While the logic of the proof is close to that of Proposition \ref{th n=2}, the proof of Theorem \ref{th learning sym} is technically significantly more involved. In particular, it is only valid under a restriction on the ESS, namely that it is \emph{regular}. We define an ESS $(\mu,\sigma)$ to be \emph{regular} if $p_{n-1}\geq 1-\hat p.$\footnote{Note that if $n=2$, then at a symmetric equilibrium $p_1=\frac{1}{2}$, so any such equilibrium must be regular.} Though we cannot formally rule out non-regular ESS, we view such equilibria as highly pathological because in such a tentative equilibrium, an agent would strictly prefer not to acquire information and to choose the minority action 0 upon observing $n-1$ ones and 1 single zero.

\begin{theorem}\label{th learning sym}
Assume $\hat p<1.$ Let $n\geq 2$ and let $(\mu_\lambda,\sigma_\lambda)$ be a regular ESS of $G(\lambda,\rho)$ with associated welfare $w^*_{\lambda}.$ 
If $\left(1-\rho+ n \rho H_0(1-\hat p)\right)\left(1-\rho + n \rho H_1(1-\hat p)\right)<1,$ then $p_0=1- \hat p$ and $p_n= \hat p$ for $\lambda$ small enough. 
\end{theorem}

Theorem \ref{th learning sym} allows us to pin down the limit distribution $\mu_{\mathrm{lim}}:=  \lim_{\lambda \to 0} \mu_\lambda$.  Indeed, we know from Theorem \ref{th consensus} that $\mu_{\mathrm{lim}}$ is concentrated on $(\theta,x)\in \Theta\times \{0,1\}$ and from the symmetry requirement, that the two marginals of $\mu_{\mathrm{lim}}$ are uniform. Theorem \ref{th learning sym} implies in addition that $\mu_{\mathrm{lim}}(\theta,x)= \frac12 \times \hat p$ when $x=\theta$. Thus, on average, society spends a fraction $1-\hat p$ of time in the phases of inertia, i.e., in an incorrect consensus. 


For small values of $\lambda$, agents who draw a unanimous sample are thus indifferent between acquiring information or not. 
Since, by Theorem \ref{th consensus}, almost all agents observe a unanimous sample in the limit $\lambda \to 0,$ Theorem \ref{th learning sym} immediately implies the following corollary.

\begin{corollary}\label{corol}
Under the assumptions of Theorem \ref{th learning sym}, one has $\underset{\lambda\to0}\lim \,w_{\lambda}^*=\hat p.$
\end{corollary}

In the persistent limit, equilibrium welfare thus exactly equals $\hat p,$ i.e., welfare is the same as in the case where $n=1.$ Remarkably, when signals are binary, welfare is even the same as in the case $n=0,$ that is, when there is no observational learning of any sort. Observing others' actions is informative, but not sufficiently so to allow agents to herd for sure, and social learning does not generate any extra value. 

For the sake of comparative statics, let us look at the case in which signals are binary with precision $\pi.$ The condition in Theorem \ref{th learning sym} then reads
%
\begin{equation}\label{SCbinary}
\left(1-\rho+n\rho \pi\right)\times \left(1-\rho+n\rho \left(1-\pi\right)\right)<1.
\end{equation}



This inequality is easier to satisfy when $n$ is lower, $\rho $ is higher and $\pi$ is higher. The intuition is as follows. 
For observing a unanimous sample to be sufficiently informative to induce agents to herd for sure, it must be that the convergence to a consensus is sufficiently slow. Otherwise, the chances of drawing a non-unanimous sample vanish too quickly and the population cannot adjust to state changes.\footnote{This concept of slow convergence is all relative. Indeed, as explained above, transitions towards the consensus are extremely rapid. What we mean to stress here is that if agents drawing unanimous samples do not acquire information, then transitions are the only instances when information can be aggregated; hence, they should not be \textit{too} rapid.} Slow convergence arises when actions within one's sample are not too correlated, that is, it is sufficiently likely that some \textit{contrarian} action (opposite to the future consensus) is observed. Clearly, this is less likely to happen as $\pi$ increases. Indeed, if signals are more precise, the signals observed by those agents who acquire information are more correlated with the true state, and hence are more correlated among each other.  For expositional simplicity, consider the extreme case of perfect signals and suppose that the state is 0. In such an instance, a consensus in which all agents play 0 will emerge. As this consensus is being built, the only possible source of contrarian actions (agents playing 1) is agents who observe a majority of ones and follow the crowd. However, as soon as a significant fraction of agents play 0, which will occur sooner or later, the fraction of agents playing one will vanish extremely quickly, inducing an excessively fast convergence to the consensus, thereby reducing its informativeness. 
%

In a similar vein, when $\rho$ decreases, the time periods from which actions are sampled become less correlated.\footnote{In the limit case $\rho=1,$ all actions are drawn from the same (previous) period.} Consequently, it becomes more likely to draw actions from time periods at which the states (and hence the prevailing consensus) were different, which in turn tends to decrease the correlation of actions within samples. 




Finally, the probability of observing a contrarian action decreases as the sample size $n$ decreases, so that \eqref{SCbinary} is more likely to be satisfied with smaller samples. \eqref{SCbinary} notably always holds when $n=2,$ consistent with Proposition \ref{th n=2}. When $\rho=1$ and $\pi=1$, \eqref{SCbinary} also holds irrespective of the sample size $n.$ Welfare is then equal to $\hat p=1-c.$ Strikingly, this is as much welfare as when all agents systematically acquire information, i.e., when there are no opportunities for observational learning $(n=0).$ If $c=0,$ this is immaterial: there is perfect learning, and welfare is maximized (equal to 1).\footnote{More generally, while Theorem \ref{th learning sym} only holds for $\hat p<1,$ Corollary \ref{corol} actually does also hold for $\hat p=1.$ In this case where signals are free and have unbounded strength, 
information is perfectly aggregated when $\lambda$ is close to 0 -- indeed, as in the cases where $n=1$ and $n=2.$} 
However, as soon as $c>0,$ although there is social learning at equilibrium (actually, most agents do herd), information aggregation is limited because the few agents acquiring information do not generate enough information to improve welfare compared to the case where $n=0.$ 
%
%
 A natural question arising from Theorem \ref{th learning sym} is whether the level of welfare of $\hat p$ is an absolute upper bound, or whether it is possible to obtain strictly higher welfare at equilibrium. For this to happen, it is necessary that information is acquired only upon observing (some) non-unanimous samples. 
 We exhibit in Theorem \ref{thm blurry} an example showing that this is possible even when the sample size is as small as $n=3$. 

\begin{theorem}\label{thm blurry}
Fix $\rho$. Assume $n=3$ and binary signals with precision $\pi\leq \frac23$.

There exists $\lambda_0>0$ such that the following holds for each $\lambda<\lambda_0$. For $\ep>0$ small enough, the game $G(\lambda\ep,\rho \ep)$ has an  ESS $(\mu,\sigma)$ in which the equilibrium payoff is bounded away from $\hat p$ as $\lambda \to 0$. 
\end{theorem}

The proof of Theorem \ref{thm blurry} is based on a simple idea, namely to view the game $G(\lambda,\rho)$ as a discretized version of a continuous time game when $\lambda$ is small. The reason why the continuous-time analysis is simpler is that invariant distributions can be identified explicitly, unlike in discrete time. However, the proof involves several technical complications. We sketch here the central mechanics of the proof and some of the technical issues. All details are relegated to the appendix. 


\bigskip\noindent
\begin{proof}[Proof Sketch]
Fix $\lambda>0$ and $\rho=1$.  Between any two consecutive periods of $G(\lambda \ep,\rho \ep)$, the state changes with probability $\lambda\ep$, a fraction $\rho\ep=\ep$ of agents is renewed, and the state population evolves according to the equation 
$x_{\mathrm{new}}= (1- \ep) x_{\mathrm{old}} + \ep g_\theta(x_{\mathrm{old}}).$\footnote{As earlier, $g_\theta(x)$ is the probability that a generic agent chooses action 1 in state $\theta$, given $x$.} 
Equivalently, 
\[\frac{x_{\mathrm{new}}-x_{\mathrm{old}}}{\ep}=  h_\theta(x_{\mathrm{old}}),\]
with $h_\theta(x):=g_\theta(x)-x$.

The main idea is to view $G(\lambda\ep,\rho\ep)$  as a discretized version of a continuous-time game $\Gamma(\lambda,\rho)$ where each period has duration $\ep$. In the continuous-time game $\Gamma(\lambda,\rho)$, the state changes at \emph{rate} $\lambda$ per unit of time, and new agents arrive at rate $\rho$. Incoming agents first observe a sample of size $n$ from the current population, before making choices.  For a given strategy, the continuous-time process $(\theta_t,x_t)$ follows a piecewise deterministic Markov process:  between two consecutive jumps of $(\theta_t)$,  say at dates $\tau<\tau'$, the $x$-component moves continuously over $[0,1]$, according to the differential equation 
\[x'(t)= h_{\theta_\tau}(x(t)).\]
We prove in the  appendix  the following technical result: if $\sigma_\ep$ is a strategy and $\mu_\ep \in \Delta(\Theta\times[0,1])$ is time-invariant  for $\sigma_\ep$ in $G(\lambda\ep,\rho\ep)$ and if $\sigma=\lim_{\ep\to 0} \sigma_\ep$, then any (weak) limit point of $(\mu_\ep)$ is  time-invariant  for the process $(\theta_t,x_t)$ induced by $\sigma$ in $\Gamma(\lambda,\rho)$. 
 
 \smallskip
 
The bulk of the proof has to do with the analysis of $\Gamma(\lambda,\rho)$. The reason why the analysis is simpler in continuous time is that invariant distributions can be identified explicitly, unlike in discrete time. Indeed, fix a strategy and assume that the invariant distribution in state $\theta$ has a $C^1$ density $f_\theta$. Consider the following argument on  the mass of those agents 'located'  in  $A:=\{1\}\times [\bar x,\bar x +dx]$, where $dx$ is small.  Over a short time interval of duration $dt$, agents 'located' in $\{1\}\times [\bar x- h_1(\bar x)dt,\bar x]$ enter $A$ from the left (or exit  if $h_1(\bar x)<0$). Hence the mass of agents entering  from the left is (approximately) $f_1(\bar x )h_1(\bar x) dt$, while the mass of agents exiting to the right is $f_1(\bar x+dx )h_1(\bar x+dx) dt$. 

In addition, the state changes with probability  $\lambda dt$, hence the mass of agents 'departing' to state 0 is $\lambda dt\times f_1(\bar x)dx$, and  the mass of agents entering $A $ from state 0 is 
$\lambda dt\times f_0(\bar x)dx$. 

Time invariance implies that 
\[\left(f_1\left(\bar x+dx \right)h_1\left(\bar x+dx\right) -f_1\left(\bar x\right)h_1\left(\bar x\right) \right)dt = \lambda\left(f_0(\bar x)-f_1(\bar x)\right)dx dt,\]
which yields 
\begin{equation}\label{mass}\left(f_1h_1\right)'=\lambda (f_0-f_1)\end{equation}
 and, for similar reasons, $\left(f_0h_0\right)'=\lambda (f_1-f_0)$.\footnote{These equations may be viewed as a one-dimensional version of the standard mass conservation equation for compressible fluids, allowing for 'phase transitions'. In physics textbooks, this equation is usually  written $\mbox{div} (f \vec{v})=0$.}


Consider the case $n=3$. We look for strategies $\sigma_b$ such that an agent who observes a unanimous sample herds (i.e., does not acquire information and follows the crowd), while an agent with a balanced sample acquires information with probability $\displaystyle b\in \left[\frac{1}{3(1-\pi)},\frac{1}{3\pi}\right].$\footnote{For $b\notin  \left[\frac{1}{3(1-\pi)},\frac{1}{3\pi}\right]$ it can be checked that the process $x(t)$ either converges to 0 or to 1, or remains bounded away from 0 and 1.} A critical observation is that for such values of $b$, there is a unique invariant distribution $\mu\in \Delta\left(\Theta\times (0,1)\right)$ that puts no mass on $x=0$ and $x=1.$ Existence is shown by solving (\ref{mass}) and using a verification argument. Uniqueness is established by showing that the process $(\theta_t,x_t)$, when sampled at the successive times where the state changes, is an irreducible process in $\Theta\times (0,1)$ and therefore has at most one invariant measure. 

In addition, $h_\theta$ has a simple form as a polynomial of degree 3, which allows to solve for $f_\theta$. An analysis of $f_\theta$ for the limit values  $b=\frac{1}{3(1-\pi)}$ and $b=\frac{1}{3\pi}$ shows that for small $\lambda$, there is one intermediate value of $b\in  \left[\frac{1}{3(1-\pi)},\frac{1}{3\pi}\right]$ such that the interim belief when sampling $k=2$ is equal to $\hat p$. In addition, $\frac{f_1}{f_0}= -\frac{h_0}{h_1}$ is increasing. In turn, this implies that $p_0<p_1=1-\hat p<p_2=\hat p< p_3$. This shows that (i) the strategy $\sigma_b$ is an ESS, and (ii) welfare exceeds $\hat p$. Proving that the game $G(\lambda \ep,\rho\ep)$ has an ESS that converges to $\sigma_b$ raises yet other technical issues that are relegated to the appendix. 
\end{proof}

\bigskip

The results of this section have noteworthy implications regarding the substitutability between the quality of private signals and the efficiency of information aggregation. To illustrate this, consider the case of binary signals with precision $\pi.$ 
An implication of Theorems \ref{th learning sym} and \ref{thm blurry} is that, holding the net value of acquiring information $\pi-c$ fixed, increasing $\pi$ may depress welfare. Indeed,  welfare cannot be larger than $\hat p=\pi-c$ 
when signals are too precise but may exceed $\hat p$ when $\pi$ is sufficiently small. This illustrates the dual role played by information acquisition (private signals) in counterbalancing the forces of imitation. On the one hand, the arrival of fresh and accurate information is necessary to remain reactive to potential changes; on the other hand, the arrival of fresh and imperfect information cultivates diversity of actions, which prevents the population from being later stuck in some absorbing consensus. By reducing the correlation of actions within samples, imperfect signals are instrumental in maintaining enough dissent within the population.\footnote{In this respect, \cite{Golub} also underline the importance of having agents with sufficiently diverse signal distributions for information aggregation.} 
In the same spirit, welfare may be larger when actions are sampled further in the past ($\rho$ is small). This comparative statics result contrasts with the result we derived in the case where $n=1,$ and thus suggests that the finding of \cite{SSsampling} that welfare is larger when samples are drawn from more recent periods does not hold for larger samples when the state can change. 
This reversal in comparative statics between the cases $n=1$ and $n\geq2$ illustrates once again the key role played by the correlation of actions within samples. A higher $\rho$ and $\pi$ increase the correlation of actions with the true state. When $n=1,$ the impact on the efficiency of social learning is unambiguously positive; but as soon as $n\geq2,$ 
actions within samples also become more correlated, which reduces informativeness, and hence welfare. 





\section{The planner's problem}\label{planner}

We have shown in the previous sections that, at least under a wide range of circumstances, the equilibrium welfare remains bounded away from 1 even when the state is arbitrarily  persistent. One key question is whether this is an equilibrium feature or an unescapable feature of the environment. Indeed, it could be that the combination of costly information acquisition and changing states intrinsically reduces the maximal welfare attainable: a high welfare requires that few agents 
acquire information, i.e., herding is prominent. Meanwhile, a high welfare requires that the population remains highly responsive whenever the state changes, hence that a sufficient weight is put on new information, so that the tension between these two objectives is consubstantial to the environment. 

\smallskip

To address this issue, we analyze the problem of a social planner who faces the constraints of the environment, but is otherwise free to choose any strategy. Let $\sigma$ be the strategy dictated by the planner, $\mu$ be an invariant distribution of states and samples for $\sigma,$ and $W_\lambda(\sigma,\mu)$ denote the corresponding welfare.
 
\begin{theorem}\label{thm:social planner}There exists $(\sigma_{\lambda},\mu_\lambda)$
such that $\underset{\lambda\to0}\lim W_\lambda\left(\sigma_{\lambda},\mu_{\lambda}\right)=1$.
\end{theorem}

A welfare of 1 is the highest possible level of welfare that is obtained only when the actions of all agents match the state and, yet, no one pays for information. Under $(\sigma_\lambda,\mu_\lambda)$, both components of the welfare are thus asymptotically optimized: there is a vanishing fraction of agents acquiring information, and most agents most often choose the correct action.\footnote{The proof (in the appendix) shows that $W\left(\sigma_{\lambda},\mu_{\lambda}\right)\geq 1-A\lambda\left(1-\ln\lambda\right)$, for some  $A<+\infty$ that does not depend on $\lambda$. This provides a lower bound for the first best welfare for any $\lambda >0.$}

The (symmetric) strategy $\sigma_\lambda$ that we construct has the following features. With a sample 
$k=0,$ 
agents acquire information with a probability of the order of $\lambda$. The probability of acquiring information increases linearly with $k$ when $k\leq n/2$, and agents acquire information  with probability 1 with a fully balanced sample $k=n/2$ (assuming $n$ even for convenience). 
Agents who do not acquire information replicate the most common action in their sample, while agents who do acquire information play action 1 if and only if their 'signal' $q$ exceeds $\frac12:$ 
agents acquiring information thus play action 1 with probability $1-H_\theta(\frac12)$ regardless of their sample $k.$

Assume for the sake of discussion that the population is in a near-consensus $x\simeq 0$, and that the state switches to $\theta=1$. Since the fraction of agents acquiring information is always at least $\lambda$, the number of agents choosing 1 starts increasing. As we show, after roughly $\lambda\ln \left(1/\lambda\right)$ stages, a significant fraction of agents observe samples $k\geq 1$. At this point, and because the probability of acquiring information increases with $k$, the dynamics of $x_t$ accelerate, and the switch to a near-consensus 
$x\simeq 1$ takes only finitely many stages. 
This implies that for small values of $\lambda$, any invariant distribution $\mu_\lambda$ assigns most weight to near-consensus states $x\simeq 0$ or $x\simeq 1$ that are (almost) perfectly correlated with $\theta$: with high $\mu_\lambda$-probability, most agents choose the correct action. This implies in turn that on average most agents receive a unanimous sample, in which case they buy information with probability $\lambda$. This ensures that the fraction of agents who acquire information is of the order of $\lambda$. 

Therefore, the fact that it is impossible to efficiently aggregate information through social learning is an equilibrium phenomenon. Intuitively, agents observing unanimous samples acquire too little information at equilibrium, because if they were to acquire more information future generations would have no incentive whatsoever to acquire information, thereby fully freezing learning. A direct consequence of this shortfall of information acquisition is a lower reactivity to state changes which materializes in the form of a significant amount of time spent in obsolete consensus.

\section{Conclusion}\label{concl}

We consider a general model of social learning with binary actions and states in which states change over time, information is possibly costly, and agents draw finite samples of past actions. We show that, under a wide range of situations, the possibility that the state changes drastically limits the value of social learning. 
This possible inefficiency is related to the tension between information aggregation and the need for society to react to potential state changes. Responsiveness to state changes requires that agents regularly acquire fresh information, while efficiency imposes that information acquisition is minimal. The problem is that, while the equilibrium fraction of agents acquiring information is vanishingly small in the persistent limit, the dynamics of behavior induced is such that actions tend to be too correlated among each other, but not enough correlated with the state, hence less social learning and a lower welfare. In these circumstances, welfare is no larger than in the case where only one action is sampled in which case the correlation of actions within the sample is by definition immaterial. 

While our paper clearly suggests that the planner could improve welfare by dictating a policy that improves the reactivity of agents while keeping information acquisition at a small rate, an interesting and natural follow-up question has to do with implementation. How could welfare improvements be obtained by the planner by designing the environment in which agents operate? In particular, a natural question is related to the feedback mechanism that maximizes the efficiency of social learning in such a context. We leave this question for future research.

\bibliography{Biblio-LPV}

\newpage

\appendix

\begin{center}
{\huge \textbf{FOR ONLINE  PUBLICATION}}
\end{center}
 \begin{appendices}
 \section{The case $n=2$: Proof of Proposition \ref{th n=2}}
\setcounter{equation}{0} \renewcommand{\theequation}{A.\arabic{equation}}
Suppose that agents with a unanimous sample herd for sure, that is, $p_2>\hat p.$ 
%
%
Since agents who observe a balanced sample $k=1$ hold beliefs $p_1=\frac12$ and thus acquire information, the fraction of agents choosing action 1 in period $t+1$ reads
\begin{equation*}
 \chi_{t+1}= x_t^2 +2x_t(1-x_t) \phi_{\theta_{t+1}},
 \end{equation*}

where $x_t$ is the probability that any action sampled at date $t+1$ is 1, and $\phi_\theta$ is the probability of playing action 1 in state $\theta$ when holding an interim belief $\frac12.$  Note that $\phi_\theta= 1-H_\theta(\frac12)$ if $H_\theta$ is continuous at $\frac12$, and $\phi_\theta= 1-\frac12\left(H_\theta(\frac12)+H_\theta(\frac12)_-\right)$ otherwise. In any case, $\phi_1=1-\phi_0>\frac12$.

The prevalence $x_t$  of ones  in the pool from which one samples at date $t+1$ is given by
\begin{equation*}
x_{t}=\displaystyle \sum_{m\geq 1}\rho(1-\rho)^{m-1}\chi_{t+1-m}= (1-\rho) x_{t-1} +\rho \chi_t.
\end{equation*}


Putting things together, the sequence $(x_t)$ follows the recursive equation 
\[x_{t+1}= (1-\rho) x_t +\rho\left\{ x_t^2 +2x_t(1-x_t) \phi_{\theta_{t+1}}\right\}.\]

We follow the proof sketch from the main body, and prove that the sequence $(x_t)$ is convergent.

\begin{lemma}\label{lemm cemetery} The sequence $(x_t)$ converges a.s., with $\lim_{t\to +\infty} x_t \in \{0,1\}$.
\end{lemma}

\begin{proof}
As a preparation, set  $\phi^\rho_\theta := (1-\rho)+\rho \times 2\phi_\theta$ and observe that $\phi^\rho_0\phi^\rho_1<1$. Indeed, $(\phi^\rho_0,\phi^\rho_1)\in \dR^2$ is a convex combination of $(1,1)$ and of $(2\phi_0,2\phi_1)$, therefore lies on the straight line with equation  $y_0+y_1= 2$. This line is tangent to the (half-)hyperbola $\calC$ of equation $y_0y_1=1$ ($y_0,y_1>0$) at the point $(1,1)$, and strictly 'below' $\calC$, except at the point of tangency. Since $\rho>0$, it follows that $\phi^\rho_0\phi^\rho_1<1$.  

Choose $\tilde \phi_\theta >\phi^\rho_\theta$ such that $\tilde \phi_0\tilde \phi_1<1$, and $\ep_0>0$ such that 
$\phi_\theta^\rho +\rho x<\tilde \phi_\theta$ for each $\theta$ and $x<\ep_0$. Note that $g_\theta(x)= x\left(1-\rho +\rho x + 2(1-x)\rho \phi_\theta\right) \leq \tilde \phi_\theta x$ for $x<\ep_0$.
\medskip

Let $\ep<\ep_0$ be arbitrary. We define two increasing and interlacing sequences $(\tau_m^{\mathrm{in}})_m$ and $(\tau_m^{\mathrm{out}})_m$ of possibly infinite stopping times. We first set
\[\tau_1^{\mathrm{out}}:=\inf\{t\geq 0: x_t<\ep \mbox{ and } \theta_t=0,\mbox{ or } x_t>1-\ep \mbox{ and } \theta_t= 1\},\]
and 
\[\tau_1^{\mathrm{in}}:=\inf\{t\geq \tau_1^{\mathrm{out}}: x_t\in [\ep,1-\ep]\},\]
with $\inf\emptyset =+\infty$. The stopping times $\tau_1^{\mathrm{out}}$ and $\tau_1^{\mathrm{in}}$ are essentially the first exit and entry times in $[\ep,1-\ep]$, except for the extra condition on the exit state in the definition of $\tau_1^{\mathrm{out}}$.  For $m\geq 1$, we set
\[\tau_{m+1}^{\mathrm{out}}:=\inf\{t\geq \tau_m^{\mathrm{in}}: x_t<\ep \mbox{ and } \theta_t=0,\mbox{ or } x_t>1-\ep \mbox{ and } \theta_t= 1\},\]
and 
\[\tau_{m+1}^{\mathrm{in}}:=\inf\{t\geq \tau_{m+1}^{\mathrm{out}}: x_t\in [\ep,1-\ep]\}.\]

\medskip

Below, we show that whenever the sequence $(x_t)$ enters the interval $[\ep,1-\ep]$, it almost surely leaves it in finite time.

\begin{claim}\label{claim1.1}
One has $\prob(\tau^{\mathrm{out}}_{m+1}<+\infty \mid \tau_m^{\mathrm{in}}<+\infty)=1$ for each $m$.
\end{claim}

\begin{proof}[Proof of the claim]
For $x\in [\ep,1-\ep]$ and $\theta\in \Theta$, one has 
\begin{equation}\label{eq10}|g_\theta(x)-x|= \rho(2\phi_1-1)(x-x^2)\geq \rho\left(2\phi_1-1\right) \ep(1-\ep).\end{equation}
Thus, if $x_t \in [\ep,1-\ep]$, the difference $x_{t+1}-x_t$ is bounded away from zero, positive if $\theta_{t+1}=1$, negative otherwise.  Set $N:=\displaystyle \lceil\frac{1}{\rho\ep((1-\ep)2\phi_1-1)}\rceil$. From (\ref{eq10}) and  the choice of $N$, it follows that if $x_t\in [\ep,1-\ep]$, one has 
$x_{t+N}\geq 1-\ep$ if $\theta_{t+1}=\cdots = \theta_{t+N}= 1$, and $x_{t+N}\leq \ep$ if $\theta_{t+1}=\cdots = \theta_{t+N}= 0$.

Since the probability that $\theta_{t+1}=\cdots = \theta_{t+N}$ is $(1-\lambda)^N$, it follows that 
\[\prob\left(\tau_{m+1}^{\mathrm{out}}\leq t+N\mid \tau_m^{\mathrm{in}} \leq t <\tau_{m+1}^{\mathrm{out}}\right) \geq (1-\lambda)^N,\]
which implies 
\[\prob\left(\tau_{m+1}^{\mathrm{out}}\geq t+jN\mid \tau_m^{\mathrm{in}} \leq t <\tau_{m+1}^{\mathrm{out}}\right) \leq \left(1-(1-\lambda)^N\right)^j\]
for each $j$, and the result follows when $j\to +\infty$.
\end{proof}
\medskip

In the next statement, $(\calH_t)_t$ is the filtration induced by $(\theta_t, x_t)_t$ and $\calH_{\tau_m^{\mathrm{out}}}$ is the stopped filtration at time $\tau_m^{\mathrm{out}}$. We show that the probability that $x_t$ ever re-enters the interval $[\ep,1-\ep]$ once it leaves it, is bounded away from 1.

\begin{claim}\label{claim2.1}
There exists $a>0$ such that $\prob\left(\tau_m^{\mathrm{in}}= +\infty\mid \calH_{\tau_m^{\mathrm{out}}}\right)\geq a$, w.p. 1 on the event $\tau_m^{\mathrm{out}}<+\infty$.
\end{claim}

\begin{proof}[Proof of the claim]
Fix a finite history $h$ of length $t$ such that $\tau_m^{\mathrm{out}} =t$.  W.l.o.g., assume $x_t<\ep$ and $\theta_t= 0$. By the Markov property, we may assume w.l.o.g. that $t=0$ and $m=1$. We define an auxiliary sequence $(W_t)$ of random variables by 
$W_0= 0$ and $W_{t+1}= W_t + \ln \tilde \phi_{\theta_{t+1}}$ for $t\geq 1$. Since $x_{t+1}\leq x_{t}  \tilde \phi_{\theta_{t+1}}$ for each $t$, one has 
\[W_t \geq \ln x_{t}-\ln x_0\]
for each $t\leq \tau_1^{\mathrm{in}}$. This implies that $\tau_1^{\mathrm{in}} \geq \inf\{t\geq 1: W_t\geq 0\}$, and $\prob(\tau_1^{\mathrm{in}}<+\infty) \leq \prob\left(\sup_{t\geq 1} W_t \geq 0\right)$.
To show that the latter probability is bounded  away from 1, we introduce the successive dates at which the state changes: we set first $\psi_0=0$, $\psi_1:=\inf\{t>0:\theta_t= 1\}$ and, for $m\geq 1$, 
$\psi_{2m}=\inf\{t>\psi_{2m-1}: \theta_t=0\}$ and $\psi_{2m+1}=\inf\{t>\psi_{2m}: \theta_t=1\}$. Finally we denote by 
\[X_j:= W_{\psi_{2j+2}}-W_{\psi_{2j}}\]
the (algebraic) increase in $W$ between $\psi_{2j}$ and $\psi_{2j+2}$. 

Observe that $(W_t)$ decreases between $\psi_{2j}$ and $\psi_{2j+1}$ and  increases  between $\psi_{2j+1}$ and $\psi_{2j}$ hence 
\[\sup_t W_t \geq 0 \Leftrightarrow \sup_j \left(X_0+\cdots +X_j\right) \geq 0.\]
By construction, the r.v.'s $(X_j)$ are \emph{iid} with $\E[X_1]= \displaystyle \frac{1}{\lambda} \left(\ln \tilde \psi_1 +\ln \tilde \psi_0\right) <0$.
The sequence $(X_0+\cdots +X_j)_j$ is therefore a simple random walk with negative drift, which implies 
\[\prob\left(\sup_j\left(X_1+\cdots +X_j\right)\geq 0\right)\leq 1-a \mbox{ for some }a>0.\]
\end{proof}

\medskip

Claims \ref{claim1.1} and \ref{claim2.1} imply that 
\[\prob\left(\tau_{m}^{\mathrm{out}} <+\infty \mbox{ and }  \tau_{m+1}^{\mathrm{in}}=+\infty\mbox{ for some }m\right)=1.\]
Thus, for each $\ep>0$, there is a random time $T_0$ such that either $x_t<\ep$ for all $t\geq T_0$, or $x_t>1-\ep$ for all $t\geq T_0$. This concludes the proof of Lemma \ref{lemm cemetery}.
\end{proof}

\begin{lemma}\label{lemm10}
The only invariant measure for $(\theta_t,x_t)$ is the uniform distribution over $\Theta\times \{0,1\}$. 
\end{lemma}

\begin{proof}
By Lemma \ref{lemm cemetery}, any invariant measure is concentrated on $\Theta\times \{0,1\}$.\footnote{Indeed, by the invariance property, $\mu (\Theta\times [\ep,1-\ep]) =\prob\left(x_t\in [\ep,1-\ep]\right)$ for each $t$. Since for each $\ep>0$  the RHS converges to zero as $t\to +\infty$, one has  $\mu(\Theta\times (0,1))=0$.}
Since the sets $\{x=0\}$ and $\{x=1 \}$ are absorbing for  $(\theta_t, x_t)$, one has for $a\in \{0,1\}$
\begin{eqnarray*}
\mu(0,a) &=& \prob\left((\theta_{t+1},x_{t+1})= (0,a)\right) \\
&=& (1-\lambda)  \prob\left((\theta_{t},x_{t})= (0,a)\right) + \lambda \prob\left((\theta_{t},x_{t})= (1,a)\right) \\
&=& (1-\lambda) \mu(0,a)+ \lambda \mu(1,a).
\end{eqnarray*}
Hence $\mu(0,a)= \mu(1,a)$: $\mu$ is a product distribution. Since $\mu$ is symmetric, it must be the uniform distribution.
\end{proof}
\medskip

We can now conclude.  By Lemma \ref{lemm10}, the sample  is non-informative: $p_k= \frac12$ for each $k=0,1,2$. But agents would then rather buy information for sure: this is the desired contradiction. 


\section{Equilibrium existence: Proof of Theorem \ref{thm existence}}
\setcounter{equation}{0} \renewcommand{\theequation}{B.\arabic{equation}}

We prove the existence of an  ESS $(\mu,\sigma)$ by means of a fixed-point
argument on an auxiliary space $\Sigma$ of (symmetric) triples
$(\mu,b,\phi)$. Define:
\begin{itemize}
\item $M$ to be the set of distributions $\mu\in\Delta(\Theta \times [0,1])$
that are invariant under the transformation $\left(\theta,x\right)\rightarrow\left(1-\theta,1-x\right)$. 
\item $B$ to be the set of \emph{belief systems} $p=(p_{k})\in[0,1]^{\{0,\ldots,n\}}$
such that $p_{k}=1-p_{n-k}$ for each $k$, with the interpretation
that $p_{k}$ is the interim belief with a sample  $k$. 
\item $\Phi$ to be the set of $\phi=(\phi_{\theta,k})\in[0,1]^{\Theta\times\{0,\ldots,n\}}$
such that $\phi_{\theta,k}=1-\phi_{1-\theta,n-k}$ for each $k$, with the interpretation
that $\phi_{\theta,k}$ is the  probability of playing action
1, when in state $\theta$ and sampling $k$. 
\end{itemize}
$M$ is compact metric when endowed with the topology of weak convergence,
and $\Sigma:=M\times B\times \Phi$ is convex compact with the product
topology. We define a set-valued map $\Psi:\Sigma\to\Sigma$ by $\Psi(\mu,p,\phi):=\Psi_{1}(\phi)\times\Psi_{2}(\mu)\times\Psi_{3}(p)$,
where $\Psi_{1},\Psi_{2},\Psi_{3}$ are defined next.

\subsubsection*{Definition and properties of $\Psi_{1}$}

Let $\phi\in \Phi$ be given. It induces a (symmetric) Markov chain $(\theta_{t},x_{t})$
in the usual way, with $(x_{t})$ obeying the recursive equation
\[
x_{t+1}=(1-\rho)x_{t}+\rho\sum_{k=0}^{n}{n \choose k}x^{k}(1-x)^{n-k}\phi_{\theta_{t+1},k}.
\]
We set $\Psi_{1}(\phi):=\{\mu\in M:\ \mu\mbox{ is an invariant measure for }(\theta_{t},x_{t})\}$.

\begin{lemma} The map $\Psi_{1}$ is uhc, with non-empty convex values.
\end{lemma}

\begin{proof} The proof is standard, and only sketched. For given
$\phi$, denote $P_\phi(z,dz')$ the one-step transition probability
of $(\theta_{t},x_{t})$: for each $z=(\theta,x)\in\Theta\times [0,1]$,
$P_\phi(z,dz')$ has a two-point support.

For fixed $\phi$ and $f\in C(\Theta\times [0,1])$, the map ${\displaystyle Tf(z):=\int_{\Theta\times [0,1]}f(z')P_\phi(z,dz')}$
is continuous in $z$. This implies that the map $\mu\in M\mapsto\mu P_\phi$
is continuous in the weak-{*} topology, and thus has a fixed point
by Tychonov Theorem. Thus, $\Psi_{1}(\phi)\neq\emptyset$.

The same argument shows that $(\phi,\mu)\in \Phi\times M\mapsto\mu P_\phi$
is continuous as well (and linear in $\mu$). This completes the proof.

\end{proof}

\subsubsection*{Definition and properties of $\Psi_{2}$}

Let $\mu\in M$. The probability of  sample $k\in\{0,1,\ldots,n\}$
in state $\theta$ is $P_{\theta}(k\mid\mu):={\displaystyle \int_{[0,1]}{n\choose k} x^k (1-x)^{n-k}\mu_{\theta}(dx)}$,
where $\mu_{\theta}$ is the  distribution of $x$ in state  
$\theta$.\footnote{Symmetry of $\mu$ implies that the marginal probability of $\theta$
is $\frac{1}{2}$, hence this conditional distribution is well-defined.} Denote by $\Psi_{2}^{k}(\mu)$ the set of all beliefs $p_{k}\in[0,1]$
that are consistent with Bayesian updating, 
\[
\Psi_{2}^{k}(\mu):=\left\{ p_{k}\in[0,1]:P_{1}(k\mid\mu)(1-p_{k})=P_{0}(k\mid\mu)p_{k}\right\} ,
\]
and 
\[
\Psi_{2}(\mu):=\left\{ p\in\prod_{k}\Psi_{2}^{k}(\mu):\ p_{k}=1-p_{n-k}\mbox{ for each }k\right\} .
\]

Lemma \ref{lemm psi2} below is immediate.

\begin{lemma}\label{lemm psi2} The map $\Psi_{2}:M\to B$ is uhc,
with non-empty convex values. \end{lemma} 

\subsubsection*{Definition and properties of $\Psi_{3}$}

Fix a system $p\in B$  of beliefs. Given a sample $k$, and a state $\theta$, we
define $\Psi_{3}^{k,\theta}(p_{k})$ as the set of probabilities  of playing
action 1 in state $\theta$ that may arise  when holding the belief $p_{k}$, then acquiring and using information in an optimal way. Formally, 
\begin{itemize}
\item $\Psi_{3}^{k,\theta}(p_{k})=\{1\}$ if $p_{k}>\hat p$. 
\item $\Psi_{3}^{k,\theta}(\hat p):=[1-H_{\theta}(1-\hat p),1]$. 
\item $\Psi_{3}^{k,\theta}(p_{k})=[1-H_{\theta}(1-p_{k}),1-H_{\theta}(1-p_{k})_-]$
if $p_{k}\in(1-\hat p,\hat p)$. 
\item $\Psi_{3}^{k,\theta}(1-\hat p):=[0,1-H_{\theta}(1-\hat p)_-]$. 
\item $\Psi_{3}^{k,\theta}(p_{k})=\{0\}$ if $p_{k}<1-\hat p$. 
\end{itemize}
We next set 
\[
\Psi_{3}(p):=\{\phi\in \Phi:\ \phi_{\theta,k}\in\Psi_{3}^{\theta,k}(p_{k})\mbox{ for each }(\theta,k)\}.
\]

Lemma \ref{lemm psi3} below is immediate.

\begin{lemma}\label{lemm psi3} The map $\Psi_{3}:B\to \Phi$ is uhc,
with non-empty convex values. \end{lemma}

\paragraph{Conclusion:}

By Kakutani Theorem, $\Psi$ has a fixed point $(\mu,p,\phi)$. The definition of $\Psi_{3}$ implies the existence of a strategy $\sigma=\left(\beta,\alpha\right)$
which induces probabilities $\phi$ such that, thanks to the definition of $\Psi_{2}$, $\sigma$ is a best response to interim beliefs induced
by the distribution $\mu$. In turn, the definition of $\Psi_{1}$ ensures that $\mu$ is an invariant distribution for $\sigma$. Symmetry is
guaranteed by definition. Hence, $(\mu,\sigma)$ is an ESS.

\section{The consensus result: Proof of Theorem \ref{th consensus general}}
\setcounter{equation}{0} \renewcommand{\theequation}{C.\arabic{equation}}

Let $(\lambda,\rho)$ and an ESS $(\mu,\sigma)$ of $G(\lambda,\rho)$ be given, with equilibrium payoff $w^*$. Consider a thought experiment in which an agent in period $t$  observes the actions $a^{(1)},\ldots, a^{(n)}$ in her sample in some random order, and denote by $\calF_m$ the information structure induced by the observation of  the first $m$ actions. We view $\calF_m$ as a distribution over (interim) beliefs, hence $\calF_m\in \Delta([0,1])$. With such notation, $w^*= w(\calF_n)$, where $w:=\max(u,v-c)$. Because more information cannot hurt,  we have
\begin{equation}
w^{*}=w\left(\mathcal{F}_{n}\right)\geq\cdots\geq w\left(\mathcal{F}_{1}\right)\text{ and }w\left(\mathcal{F}_{m}\right)\geq u\left(\mathcal{F}_{m}\right).\label{eq:VI basic inequalities}
\end{equation}

We look more closely at the information structures $\calF_0,\calF_1$ and $\calF_2$, and introduce some notation. We denote by $\nu:=\int_{\Theta\times [0,1]} (1- x) d\mu(\theta,x)$ the probability of $a^{(1)}=0$, and by $q_a=\prob(\theta_{t}=1\mid a^{(1)}=a)$ the interim belief when first sampling $a\in A$. We also denote by $\nu_{a}$  the conditional probability of next sampling $a^{(2)}=0$ given $a^{(1)}=a$, and by $q_{aa'}:=\prob(\theta_{t}=1\mid (a^{(1)}, a^{(2)})= (a,a'))$ the interim belief given the first two actions in the sample. Using the notation  $\left(a_{1}\right)^{p_{1}}....\left(a_{k}\right)^{p_{k}}$ to describe the finite-support probability distribution that assigns probability $p_m$ to $a_m$, we thus have
\begin{align}
\mathcal{F}_{0} & =\left(\frac12\right)^{1},\label{eq:information 0}\\
\mathcal{F}_{1} & =\left(q_{0}\right)^{\nu}\left(q_{1}\right)^{\left(1-\nu\right)},\label{eq:information 1}\\
\mathcal{F}_{2} & =\left(q_{00}\right)^{\nu\nu_{0}}\left(q_{01}\right)^{\nu\left(1-\nu_{0}\right)}\left(q_{10}\right)^{\left(1-\nu\right)\nu_{1}}\left(q_{11}\right)^{\left(1-\nu\right)\left(1-\nu_{1}\right)}.\label{eq:information 2}
\end{align}

 Because
$a^{(1)}$ and $a^{(2)}$  are exchangeable, we
have
\[
q_{01}=q_{10}\text{ and }\nu\left(1-\nu_{0}\right)=\left(1-\nu\right)\nu_{1}.
\]
By the martingale property of beliefs, the expected belief in all three information structures is the same
and equal to 
\begin{equation}
\frac12=\nu q_{0}+\left(1-\nu\right)q_{1}=\nu\nu_{0}q_{00}+2\nu\left(1-\nu_{0}\right)q_{01}+\left(1-\nu\right)\left(1-\nu_{1}\right)q_{11}.\label{eq:average belief}
\end{equation}

In addition, by symmetry, one has $\nu=\frac12,$ which yields $q_0=1-q_1,$ $q_{00}=1-q_{11},$ $\nu_0=1-\nu_1$ and $q_{01}=q_{10}=\frac12.$

The result relies on the fundamental observation below.
\begin{lemma}
\label{lem:.Key informational}
One has $w^{*}-\Lambda\leq u\left(\mathcal{F}_{1}\right)$.
\end{lemma}
\begin{proof}
Consider the strategy consisting of replicating $a^{(1)}$. The payoff from such a strategy $u\left(\mathcal{F}_{1}\right)=(1-\Lambda)\kappa^*+\Lambda(1-\kappa^*)$ where $\kappa^*$ is the equilibrium probability that a random agent plays the correct action. It follows that
\begin{eqnarray*}
u\left(\mathcal{F}_{1}\right)&=&\kappa^*-\Lambda+2\Lambda(1-\kappa^*)\\
&\geq&w^*-\Lambda,
\end{eqnarray*}
using $\kappa^*\geq w^*$ ($w^*$ differs from $\kappa^*$ by the expected cost of information acquisition).

\end{proof}

\bigskip
Together with (\ref{eq:VI basic inequalities}), Lemma \ref{lem:.Key informational}  implies
\begin{equation}\label{information}
w(\calF_1)-u(\calF_1)\leq \Lambda \mbox{ and } w(\calF_2)-w(\calF_1)\leq \Lambda.
\end{equation}

\bigskip


Observe that 
\begin{eqnarray*}
\Lambda&\geq & w\left(\mathcal{F}_{1}\right)-u\left(\mathcal{F}_{1}\right)\\
&=& \nu\left[w\left(q_{0}\right)-u\left(q_{0}\right)\right]+\left(1-\nu\right)\left[w\left(q_{1}\right)-u\left(q_{1}\right)\right].\\
&= & w\left(q_{0}\right)-u\left(q_{0}\right),
\end{eqnarray*}
using $\nu=\frac12$ and $u(q_1)=u(q_0)$ and $w(q_1)=w(q_0)$ by symmetry.
\bigskip

Notice that, since $\theta$ and $a$ are positively correlated in any ESS, and states are persistent, one has $q_{1}>q_{0}.$ This implies $q_0<\frac12<q_1.$ We  now  conclude. Starting from  (\ref{information}), one has the following inequalities:
\begin{eqnarray*}
\Lambda &\geq &  w\left(\mathcal{F}_{2}\right)-w\left(\mathcal{F}_{1}\right)\nonumber \\
&= &  \E_\mu\left[ w(q_{a^{(1)}a^{(2)}})-w(q_{(a^1)}) \right] \nonumber  \\
&= & \prob\left(a^{(1)}= 0\right) \E\left[ w(q_{a^{(1)}a^{(2)}})-w(q_{(a^1)})\mid a^{(1)}= 0\right] +\prob\left(a^{(1)}= 1\right) \E\left[ w(q_{a^{(1)}a^{(2)}})-w(q_{(a^1)})\mid a^{(1)}= 1\right]\nonumber \\
&\geq& \frac12\left[\nu_{0}w\left(q_{00}\right)+\left(1-\nu_{0}\right)w\left(q_{01}\right)-w\left(q_{0}\right)\right] \\
&= &  \frac12\left[\nu_{0}w\left(q_{00}\right) +\left(1-\nu_{0}\right)u \left(q_{01}\right)- w\left(q_0\right)\right] +\frac12(1-\nu_0)d\label{eq:VI ineq 1}, \\
\end{eqnarray*}
where $d:=w(q_{01})-u(q_{01})=w(\frac12)-u(\frac12)>0$ by Assumption 1. 

Using $w(q_0)\leq u(q_0)+\Lambda$ and $u(q_{00})\leq w(q_{00}),$ this implies
\begin{equation*}
\Lambda\geq\frac12\left[\nu_{0}u\left(q_{00}\right) +\left(1-\nu_{0}\right)u \left(q_{01}\right)- u\left(q_0\right)\right]-\frac12\Lambda+\frac12(1-\nu_0)d
\end{equation*}

Finally, using $q_{00}\leq q_{01}=\frac12,$ the fact that $u$ is affine on $(0,\frac12)$ implies
 \[\nu_{0}u\left(q_{00}\right)+\left(1-\nu_{0}\right)u \left(q_{01}\right)=
u\left( \nu_0 q_{00}+\left(1-\nu_{0}\right)q_{01}\right)=u \left(q_{0}\right).\] 
We conclude
\begin{eqnarray*}
\Lambda &\geq &  -\frac12\Lambda+\frac12(1-\nu_0)d
\end{eqnarray*}

This implies $
1-\nu_{0}\leq\frac{3}{d}\Lambda.$
Since
\[
\int_{\Theta\times [0,1]} x\left(1-x\right)d\mu\left(x\right)=\prob(a^{(1)}=0,a^{(2)}= 1)= \nu\left(1-\nu_{0}\right),\]

we conclude that \[
\int_{\Theta\times [0,1]} x\left(1-x\right)d\mu\left(x\right)\leq\frac{3}{2d}\Lambda.\]

\section{No Social Learning: Proof of Theorem  \ref{th learning sym}}
\setcounter{equation}{0} \renewcommand{\theequation}{D.\arabic{equation}}

We proceed in two steps. We first prove that for any ESS and provided $\Lambda$ is small enough, the interim beliefs with a unanimous sample satisfy  $p_0\leq 1-\hat p$  and $p_n\geq \hat p$.  This is the content of Proposition  \ref{lem:beliefs not too high} below, which holds in full generality when $\hat p<1$. We next prove in Proposition \ref{lem:beliefs not too low} that, under the additional assumptions of Theorem \ref{th learning sym}, the reverse inequalities also hold. 

\subsection{Step 1}

\begin{proposition}
\label{lem:beliefs not too high} There exists
$\Lambda^{*}>0$ such that for every $(\lambda,\rho)$ s.t. $\Lambda\leq \Lambda_*$ and every ESS $(\mu,\sigma)$ of $G(\lambda,\rho)$, one has $p_0\leq 1-\hat p$
(and $p_n\geq  \hat p$).
\end{proposition}

We argue by contradiction: assume  that there is a convergent sequence $(\lambda_m,\sigma_m)$ such that $\lim_m\Lambda_m=0$, and a convergent sequence $(\mu_m,\sigma_m)$ of ESSs, such that $p_{0,m}>1-\hat p$ along the sequence.  
We denote by $g_{\theta,m}$ the dynamics of $x$ under the strategy $\sigma_m$, and we denote by  $\phi_{\theta,m}(k)$ the (steady-state) probability of playing $a=1$ in state $\theta$, given $(\mu_m,\sigma_m)$, when sampling $k$.
 We assume that the sequences $(\mu_m)$ and  $(\phi_{\theta,m}(k))$ are convergent for each $\theta$ and $k$, with limits $\mu$ and $\phi_\theta(k)$.  This implies that $(g_{\theta,m})$ converges uniformly, and we  denote by $g_{\theta}$ its limit.
 
 By Theorem \ref{th consensus}, the conditional distribution $\mu_\theta$ is concentrated on $\{0,1\}$ for each $\theta$, and we write $\mu_\theta(a)$ instead of $\mu_\theta(\{a\})$. By symmetry, $\mu_0(a)+\mu_1(a)=1$ for each $a$. The next result relates $\mu_\theta(a)$ to the  limit dynamics $g_\theta$ at the boundary points.

\begin{lemma}\label{lemm simple}
For each $\theta\in \Theta$, the following implication holds:
if $\mu_\theta(0)>0$ then either $ g_\theta(0)= 0$ or $ g_\theta(1)=0$. Similarly, if $\mu_\theta(1)>0$, then $g_\theta(1)= 1$ or $g_\theta(0)=1$.
\end{lemma}

\begin{proof}
By Theorem \ref{th consensus general}, 
$
\lim_m \mu_{\theta,m}\left(\left[x,1-x\right]\right)=0$ 
for each $x\in (0,1)$ and each  $\theta\in\Theta$.
Fix $x>0$. The steady-state equations imply that
\begin{eqnarray*}
\mu_{\theta,m}\left([0,x]\right) & =&\left(1-\lambda_{m}\right)\mu_{\theta,m}\left(g_{\theta,m}^{-1}\left([0,x]\right)\right)+\lambda_{m}\mu_{1-\theta,m}\left(g_{\theta,m}^{-1}\left([0,x]\right)\right) \mbox{ for each
 $m$,}
\end{eqnarray*}
hence $\lim_m \left(\mu_{\theta,m}\left([0,x]\right)  -\mu_{\theta,m}\left(g_{\theta,m}^{-1}\left([0,x]\right)\right)\right)=0$, which implies in turn that for each $x\in\left(0,1\right)$
\begin{eqnarray*}
\mu_\theta(0) & =& \lim_m \mu_{\theta,m}\left([0,x]\right) = \lim_m \mu_{\theta,m}\left(g_{\theta,m}^{-1}\left([0,x]\right)\right)\\
& =& \lim_m \left\{\mu_{\theta,m}\left(g_{\theta,m}^{-1}\left([0,x]\right)\cap [0,x]\right)+\mu_{\theta,m}\left(g_{\theta,m}^{-1}\left([0,x]\right)\cap [1-x,1]\right)\right\}
\end{eqnarray*}
If $g_\theta(0)> 0$, then $ g_{\theta,m}^{-1}\left([0,x] \right)\cap [0,x]=\emptyset$ for $x$ small enough and all $m$ large enough. Similarly $g_\theta(1)> 0$ implies $ g_{\theta,m}^{-1}\left([0,x]\right)\cap [1-x,1]=\emptyset$ for $x$ small enough and all $m$ large enough. Hence, if both $g_\theta(0)$ and $g_\theta(1)$ are non-zero, one has $\mu_\theta(0)=0$. 
\end{proof}
\bigskip

\begin{proof}[Proof of Proposition \ref{lem:beliefs not too high}] We continue the proof of the proposition. 

 Since $\mu=\lim_m \mu_m$, it is easy to check that 
\[
p_0=\lim_{m\to +\infty} p_{0,m}=\frac{\mu_{1}\left(0\right)}{\mu_{0}\left(0\right)+\mu_{1}\left(0\right)}=  \mu_1(0)\text{ and }p_n=\lim_{m\to +\infty} p_{n,m}=\frac{\mu_{1}\left(1\right)}{\mu_{0}\left(1\right)+\mu_{1}\left(1\right)}=  \mu_1(1).
\]
Since $p_{0,m}>1-\hat p$ for each $m$, the probability $\phi_{\theta,m}(0)$ of playing action 1 when  sampling 0 is  bounded away from zero as $m$ varies.\footnote{An agent either acquires information, or not.  In the former case,  since $p_{0,m}>1-\hat p$, and since $\hat p<1$, the probability of playing action 1 is at least $1-H_\theta\left(\hat p\right)$. If instead the agent chooses not to acquire information, then it must be that $p_{0,m}\geq \hat p$, and then she plays action 1  w.p. 1.} 
Therefore, $\inf_m g_{1,m}(0)>0$, which  implies $g_{1}\left(0\right)>0$.

On the other hand,  $p_0>0$ hence $\mu_{1}\left(0\right)>0$.  Using Lemma \ref{lemm simple},  we obtain $g_{1}\left(1\right)=0$. 
Since $g_{0}\leq g_{1}$, this yields $g_0(1)=0$ in turn. 

Since  $p_n\leq \hat p$ by symmetry, one has $\mu_0(1)>0$. Repeating the argument of the previous paragraph, we obtain $g_\theta(0)=1$ for each $\theta$. 

It follows that for each $x$ small enough, each $y>0$, and each $\theta$, $\lim_m \mu_{\theta,m}\left( g_m^{-1}([0,x])\cap [0,x]\right)=0$, which implies (see the proof of Lemma \ref{lemm simple}) that
\begin{eqnarray*}
\mu_\theta(0)& =&\lim_m \mu_{\theta,m}\left( g_m^{-1}([0,x])\cap [1-y,1]\right)= \lim_m \mu_{\theta,m}\left([1-y,1]\right)= \mu_\theta(1):
\end{eqnarray*}
 $\mu$ is the uniform distribution over $\Theta\times \{0,1\}$ -- a contradiction. 
\end{proof}

\bigskip 

\subsection{Step 2}

In the sequel, we simply write $H_\theta(1-\hat p)$ in place of $\frac12\left(H_\theta(1-\hat p)+H_\theta(1-\hat p)_-\right)$ for clarity. 

We now prove that, under \textbf{Assumption A} below, the reverse implications hold (for $\lambda$ small) for regular ESS, thereby completing the proof of  Theorem  \ref{th learning sym}.
\begin{center}
\textbf{Assumption A}: $\displaystyle \prod_{\theta\in \Theta} \left(1-\rho + n\rho H_\theta(1-\hat p)\right)<1$.
\end{center}
\begin{proposition}
\label{lem:beliefs not too low}
Under \textbf{Assumption A}, there exists $\Lambda^{*}$ with the following property. For every $(\lambda,\rho)$ such that $\Lambda \leq \Lambda_*$, and every regular ESS of $G(\lambda,\rho)$, one has 
$p_0\geq 1-\hat p$. 
\end{proposition}

Let $(\lambda,\rho)$, and a regular ESS $(\sigma,\mu)$ of $G(\lambda,\rho)$ be given. We argue by contradiction: we assume below that $p_0<1-\hat p$, and eventually derive a contradiction if $\Lambda $ is small enough.

We start  with an obvious, useful observation. Since $p_1<\hat p$, agents with a sample $k=1$ either acquire information, or play action 0 for sure. Either way, the probability in state $\theta$ of playing action 1 is at most $\phi_\theta(1)\leq 1-H_\theta(1-\hat p)$.\footnote{This is the only place where the regularity assumption is used.} On the other hand, since $p_0< 1-\hat p$, agents with a unanimous sample $k=0$ play 0 for sure hence $\phi_\theta(0)=0$. This proves Claim \ref{lem:bound signals} below.

\begin{claim}
\label{lem:bound signals}
One has $\phi_\theta(0)= 0$ and $\phi_\theta(1)\leq 1-H_\theta(1-\hat p)$. 
\end{claim}

We first prove Lemmas \ref{lemm mon}, \ref{lem:g linear},  and \ref{lem:path counting} (below), then proceed with the proof of the proposition.

Let $F_\theta$ be the cdf of $x$ in state $\theta \in \Theta$. Lemma \ref{lemm mon} bounds the change in the cdf at two consecutive population
states, by a factor of order $\lambda$.  
Lemma \ref{lem:g linear} provide linear bounds on $g_\theta$ around 0. This resembles the use of a random walk (with drift) for the case $n=2$  (see Lemma \ref{lemm cemetery}). Lemma \ref{lem:path counting} is a generalization
of the result for $n=2$  that the only invariant measures of $\sigma$ are non-informative (see Lemma \ref{lemm10}).

\begin{lemma}\label{lemm mon}
If $g_\theta(0)= 0$ and  $g_{\theta}:[0,x_0]\to [0,g_{\theta}(x_0)]$ is a strictly increasing    for some $x_0>0$, then for $x<x_0$
\begin{eqnarray}
F_{0}\left(x\right) & =&\left(1-\lambda\right)F_{0}\left(g_{0}^{-1}\left(x\right)\right)+\lambda F_{1}\left(g_{0}^{-1}\left(x\right)\right),\label{eq:steady state eq}\\
F_{1}\left(x\right) & =&\left(1-\lambda\right)F_{1}\left(g_{1}^{-1}\left(x\right)\right)+\lambda F_{0}\left(g_{1}^{-1}\left(x\right)\right).\nonumber 
\end{eqnarray}
and
\begin{eqnarray}
F_{0}\left(g_{0}^{-1}\left(x\right)\right)-F_{0}\left(x\right) & =&\lambda \left(F_{0}\left(g_{0}^{-1}\left(x\right)\right)-F_{1}\left(g_{0}^{-1}\left(x\right)\right)\right), \label{eq:F0 change}\\
0\leq F_{1}\left(x\right)-F_{1}\left(g_{1}^{-1}\left(x\right)\right) & =&\lambda \left(F_{0}\left(g_{1}^{-1}\left(x\right)\right)-F_{1}\left(g_{1}^{-1}\left(x\right)\right)\right),
\end{eqnarray}

\end{lemma}

\medskip

\begin{proof}
The first two equations describe stationary distributions. The last two are derived from the first two with simple algebra. \end{proof}
\medskip

\begin{lemma}
\label{lem:g linear}There are $0<\tilde{g}_{0}<1$, \textup{$\tilde{g}_{1}<1/\tilde{g}_{0}$,}
and $x_0>0$ such that
\[
g_\theta\left(x\right)\leq\tilde{g}_{\theta}x\mbox{ for each } x\in [0,x_0].
\]
In addition, functions $g_\theta$ are strictly increasing on $x\in [0,x_0]$. 
\end{lemma}

\noindent
\begin{proof}
Recall that 
\[
g_{\theta}\left(x\right)=\left(1-\rho\right)x+\rho\sum_{k=0}^n\left(\begin{array}{c}
n\\
k
\end{array}\right)x^{k}\left(1-x\right)^{n-k}\phi_\theta(k).
\]
Since $\phi_{\theta}\left(0\right)=0$ (cf. Claim \ref{lem:bound signals}), we have
\[
\frac{g_{\theta}\left(x\right)}{x}\leq1-\rho+\rho \left(n \phi_\theta\left(1\right)+M x \right),\mbox{ for some $M$ and all $x\in [0,1]$.}
\]
Since $\phi_{\theta}\left(1\right)\leq 1-H_\theta(1-\hat p)$ (cf. Claim \ref{lem:bound signals}) and by  \textbf{Assumption A},  there exists $x^{*}>0$ so that 
\[
\left(1-\rho+\rho\left(n \phi_0(1)+Mx^{*}\right)\right)\left(1-\rho+\rho\left(n\phi_1(1)+Mx^{*}\right)\right)<1.
\]
The result follows, with $\tilde{g}_{\theta}=(1-\rho) + \rho \left(n \left(1-H_\theta\left(1-\hat p\right)\right)+Mx^{*}\right)$.
\end{proof}

\begin{lemma}
\label{lem:path counting} Let  $x_0$ be as in Lemma \ref{lem:g linear} and let $x^{*}=g_0(x_0)<x_0$. One has
\[
\mu_{0}\left(\{0\}\right)=\mu_{0}\left([0,x^{*}]\right)=\mu_{1}\left(\{0\}\right)=\mu_{1}\left([0,x^{*}]\right).
\]
\end{lemma}

\noindent
\begin{proof}
Define 
\begin{equation}
\omega =\inf_{x\leq g^{-1}_0(x^{*})}\frac{F_{1}\left(x\right)}{F_{0}\left(x\right)}\text{ and }\eta=\sup_{x\leq x^{*}}\frac{F_0(g_{0}^{-1}(x))}{F_0(x)}-1.
\end{equation}
The inequality in (\ref{eq:F0 change}) implies that $\omega\leq1$ and, clearly, $\eta\geq0$.

Take $x\leq g_0^{-1}(x^{*})=x_0$. By dividing the second equation of (\ref{eq:steady state eq}) by $F_0(g_1^{-1}(x))$ and noticing that $F_0(g_1^{-1}(x))>F_0(g_0(x))$, we obtain
\begin{equation}
\frac{F_1(x)}{F_0(x)} (\eta+1) \geq \frac{F_1(x)}{F_0(x)}\frac{F_0(x)}{F_0(g_0(x))}  \geq \frac{F_1(x)}{F_0(x)} \frac{F_0(x)}{F_0(g_1^{-1}(x))}= \left(1-\lambda\right)\frac{F_1(g_1^{-1}(x))}{F_0(g_1^{-1}(x))}+\lambda \geq \left(1-\lambda\right)\omega+\lambda.\nonumber
\end{equation}
Because the above holds for each $x\leq g_0^{-1}(x^{*})$, we get $\omega (\eta+1) \geq \left(1-\lambda\right)\omega+\lambda$ and thus, $\omega\geq \frac{\lambda}{\lambda+\eta}$.

By dividing the first equation of (\ref{eq:steady state eq}) by $F_0(g_0^{-1}(x))$, we obtain
\begin{equation}
\frac{F_0(x)}{F_0(g_{0}^{-1}(x))} \geq \left(1-\lambda\right)+\lambda \omega.\nonumber
\end{equation}
Because the above holds for each $x\leq x^{*}$, it must be that
\begin{equation}
\eta +1 = \sup_{x\leq x^{*}}\frac{F_0(g_{0}^{-1}(x))}{F_0(x)} \leq \frac{1} {1-\lambda+\lambda \omega}. \nonumber
\end{equation}
Combining with $\omega\geq \frac{\lambda}{\lambda+\eta}$, we get
\begin{eqnarray}
\nonumber \eta  &\leq& \frac{1}{1-\lambda +\frac{\lambda^2}{\lambda+\eta}}-1=\frac{\lambda+\eta}{\lambda+\eta(1-\lambda)}-1=\frac{\lambda\eta}{\lambda+\eta(1-\lambda)}.  
\end{eqnarray}
After multiplying both sides of the inequality by the denominator and subtracting $\lambda\eta$, we obtain $\eta^2(1-\lambda)\leq 0$. 

Hence $\eta=0$, which further implies that $\omega=1$. Recalling the definitions of $\eta$ and $\omega$, we get
$F_{0}\left(x\right) =  F_{0}\left(\tilde{g}_{0}^{-1}(x)\right) =  F_{1}\left(x\right)$ for each $x<x^{*}$. By iterating over applications of function $g_0()$, and noticing that, for each $0<x<x^{*}$, there exists $m$ such that $g_0^m(x^*)\leq x\leq g_0^{m-1}(x^*)$, we obtain that  
\begin{equation}
F_0(x^{*})=F_0\left(g_0^m(x^*)\right) \leq F_0(x) \leq F_0\left(x^*\right), \nonumber
\end{equation}
which implies that $F_0(x)=F_0(x^*)=F_1(x)$. This ends the proof of the Lemma. 
\end{proof}

\medskip
We now proceed with the proof of Proposition \ref{lem:beliefs not too low}.

\medskip\noindent

\begin{proof}[Proof of Proposition \ref{lem:beliefs not too low}] We aim at a contradiction for  $\Lambda$ small enough.
By Lemma \ref{lem:path counting},
one has  \[
\mu_{0}\left(\{0\}\right)=\mu_{0}\left([0,x^{*}]\right)=\mu_{1}\left(\{0\}\right)=\mu_{1}\left([0,x^{*}]\right),
\]
where $x^*>0$ is independent of $\Lambda$.

This implies $p_0=\frac12> 1-\hat p$ -- a contradiction.\end{proof}

\section{The Continuous-Time Approximation: Proof of Theorem \ref{thm blurry}}
\setcounter{equation}{0} \renewcommand{\theequation}{E.\arabic{equation}}

The proof of Theorem \ref{thm blurry} is divided in two parts. We first focus on the analysis of the continuous-time game $\Gamma(\lambda,\rho)$, and establish the existence of an ESS  $(\sigma,\mu)$ with the desired properties. We next address the convergence issue \emph{per se}, and prove that for small $\ep>0$, the discrete-time $G(\lambda\ep,\rho\ep)$ has an equilibrium $(\sigma_\ep,\mu_\ep)$ close to $(\sigma,\mu)$.

\subsection{The analysis of the continuous-time game $\Gamma(\lambda,\rho)$}

\subsubsection{The continuous-time game $\Gamma(\lambda,\rho)$.}

We recall the (slightly informal) description of $\Gamma(\lambda,\rho)$ and provide a formal definition of ESS.

The timeline is $\dR_+$. The state follows a (continuous-time) Markov process $(\theta_t)_{t\geq 0}$ over $\Theta$, with constant switching rate $\lambda$ per unit of time. Agents get replaced at rate $\rho>0$. W.l.o.g., we normalize $\rho$ to $\rho=1$.

Incoming agents observe first a sample of size $n$ from the current population, next decide whether or not to acquire binary information with precision $\pi$, and choose an action.  We focus on strategies that acquire information w.p. 1 when $k=\frac{n}{2}$, and assume that agents choose the majority action in their sample when they don't acquire information, and follow their signal when they do. Consequently,  strategies are simply functions $\sigma:\{0,\ldots, n\}\to [0,1]$.

Given a strategy  $\sigma \in [0,1]^{n+1}$ and $\theta$, we define $h_\theta:[0,1]\to [0,1]$ by
$h_\theta(x):= \rho\left(g_\theta(\theta)-x\right)$, with 
\[g_\theta(x):= \sum_{k=0}^n {n\choose k} x^k(1-x)^{n-k} \left(\sigma_k\left(\theta\pi+(1-\theta)(1-\pi)\right) +(1-\sigma_k)1_{k>n/2}\right).\footnote{The function $g_\theta$ depends on the strategy $\sigma$. For simplicity, we make no reference to it in the notation $g_\theta$, unless needed.}\]

The strategy $\sigma$ defines a piecewise deterministic Markov process $(Z_t):=(\theta_t,x_t)$ over $\Theta\times [0,1]$, where the second component $x_t$ moves continuously over $[0,1]$, according to the vector field $h_{\theta_t}$ -- that is, $x(t)=x_t$
obeys the differential equation $x'(t)=h_{\theta_t}(x(t))$ between two consecutive state changes. 

\medskip

The process $(x_t)$ admits an explicit description, see \cite{benaim}. Denote by $0=\tau_0<\tau_1<\cdots$ the successive  jumps of $(\theta_t)$ and by $u_i:=\tau_{i+1}-\tau_i$ the 'inter-arrival' times. The r.v.'s $(u_i)$ are i.i.d., and follow an exponential distribution with parameter $\lambda$.  Since $|\Theta|=2$, the state at time $\tau_i$ is $\alpha_i:= \theta_0+i\mbox{ mod }2$. 

For $\theta\in \Theta$, let $\Phi^\theta:\dR_+\times [0,1]\to [0,1]$ be the flow associated with the vector field $h_\theta,$ that is, $\Phi^\theta(s,x)$ is the value at time $s$ of the solution to the Cauchy problem $x'(t)=h_\theta(x(t))$, $x(0)=x$. Since $x(t)$ follows the vector field $h_{\theta_0}$ from $\tau_0$ to $\tau_1$, one has $x_{\tau_1}= \Phi^{\theta_0}(\tau_1,x_0)$ and more generally, $x_{\tau_{k+1}}= \Phi^{\alpha_{k}}(u_k,x_{\tau_k})$ for $k\geq 0$.

Thus, when starting from  $(\theta_0,x_0)$, one has
\begin{equation}\label{explicit}x_t=\Phi^{\alpha_k}_{t-\tau_k}\circ \Phi^{\alpha_{k-1}}_{u_{k-1}}\circ \cdots \circ \Phi^{\alpha_0}_{u_0}(x_0) \mbox{ on the event $\tau_k\leq t< \tau_{k+1}$},\end{equation}
where $\Phi_s^\theta= \Phi^\theta(s,\cdot)$.  For  fixed $t$, $x_t$ is a  function of $(\theta_0,x_0)$ and of the inter-arrival times $(u_i)$.

For later reference, we note that if $x\in (0,1)$, then $\Phi^\theta(s,x)\in (0,1)$ for each $\theta$ and all $s\geq 0$. 
\medskip

The extension to $\Gamma(\lambda,\rho)$ of the ESS concept is straightforward: a pair $(\mu,\sigma)$ is an ESS if (i) $\mu$ is an invariant measure for the Markov process $(\theta_t,x_t)$ induced by $\sigma$, (ii) $\sigma$ is optimal given the interim beliefs deduced from $\mu$ and (iii) $\sigma$ and $\mu$ treat the two states and actions symmetrically.

\subsubsection{An ODE verification result for invariant distributions}

In the main body, we argued heuristically that the densities of an invariant measure, if they exist, satisfy a simple ODE. We here state and prove the corresponding verification statement. 

\begin{lemma}\label{verif}
Assume that $f_\theta:(0,1)\to \dR_+$ is $C^1$, that  $\displaystyle \int_{(0,1)} f_\theta dx= 1$  and  that $(f_0,f_1)$ solves 
\begin{equation}\label{ODE}(f_1h_1)'= \lambda (f_0-f_1)\mbox{ and }(f_0h_0)'= \lambda (f_1-f_0)\mbox{ over }(0,1).\end{equation} 
Let $\mu_\theta\in \Delta([0,1])$ be the probability measure with density $f_\theta$, and let
$\mu :=\frac12 \delta_{0}\otimes \mu_0 +\frac12 \delta_1 \otimes \mu_1\in \Delta(\Theta\times [0,1])$ the measure with conditionals $\mu_0$ and $\mu_1$. Then $\mu$ is  invariant for 
$\sigma $ in $\Gamma(\lambda,\rho)$. 
\end{lemma}

If $h_0$ and $h_1$ do not vanish on $(0,1)$, standard results imply the existence and uniqueness of a solution to (\ref{ODE}) for fixed values of, say, $f_0(\frac12)$ and $f_1(\frac12)$. It is given by
\begin{equation}\label{sol}
f_1(x)= \frac{c_1}{h_1(x)}\exp\left(-\lambda \int_{\frac12}^x \left\{\frac{1}{h_0(t)}+\frac{1}{h_1(t)}\right\}dt \right)\mbox{ and }
f_0(x)= -\frac{c_0}{h_0(x)}\exp\left(-\lambda \int_{\frac12}^x \left\{\frac{1}{h_0(t)}+\frac{1}{h_1(t)}\right\}dt \right),\end{equation}
with $c_1= h_1(\frac12)f_1(\frac12)$ and  $c_0=- h_0(\frac12)f_0(\frac12)$.

\medskip\noindent
\begin{comment}
\begin{proof}
The infinitesimal generator $L$ of the Markov process $(\theta_t,x_t)$ is given by 
\begin{equation}\label{generator}
L\phi(\theta,x)= h_\theta(x)\phi'(x) +\lambda \left(\phi_{1-\theta}(x)-\phi_\theta(x)\right),
\end{equation}
for $(\theta,x)\in \Omega$, and each pair $\phi= (\phi_\theta)$ of functions in $C^1((0,1))$, see Bena\"{\i}m et \emph{al.} (2014).

We need to prove that $\displaystyle \int_\Omega L\phi d\mu= 0$ for each smooth $\phi$, see \cite{varadhan}, Section 7.4.

Let  $\phi_\theta$ ($\theta\in \Theta$) be a $C^1$ function with compact support in $(0,1)$. One has 
\begin{eqnarray*}
2\int_\Omega L\phi d\mu &=& \int_0^1 \phi'_1 h_1f_1 dx +\lambda \int_0^1 \left(\phi_1-\phi_0\right)f_1 dx \\
&& +\int_0^1 \phi'_0h_0f_0dx +\lambda \int_0^1 \left(\phi_0-\phi_1\right)f_0 dx.
\end{eqnarray*}
Using first an integration by parts in the first and third integrals, next the equality $(f_\theta h_\theta)'=\lambda (f_{1-\theta}-f_\theta)$, one obtains $L\phi=0$.
\end{proof}
\end{comment}

\begin{proof}
The infinitesimal generator $L$ of the Markov process $(\theta_t,x_t)$ is given by 
\begin{equation}\label{generator}
L\xi(\theta,x)= h_\theta(x)\xi'(x) +\lambda \left(\xi_{1-\theta}(x)-\xi_\theta(x)\right),
\end{equation}
for $(\theta,x)\in \Omega$, and each pair $\xi= (\xi_\theta)$ of functions in $C^1((0,1))$, see Bena\"{\i}m et \emph{al.} (2014).

We need to prove that $\displaystyle \int_\Omega L\xi d\mu= 0$ for each smooth $\xi$, see \cite{varadhan}, Section 7.4.

Let  $\xi_\theta$ ($\theta\in \Theta$) be a $C^1$ function with compact support in $(0,1)$. One has 
\begin{eqnarray*}
2\int_\Omega L\xi d\mu &=& \int_0^1 \xi'_1 h_1f_1 dx +\lambda \int_0^1 \left(\xi_1-\xi_0\right)f_1 dx \\
&& +\int_0^1 \xi'_0h_0f_0dx +\lambda \int_0^1 \left(\xi_0-\xi_1\right)f_0 dx.
\end{eqnarray*}
Using first an integration by parts in the first and third integrals, next the equality $(f_\theta h_\theta)'=\lambda (f_{1-\theta}-f_\theta)$, one obtains $L\xi=0$.
\end{proof}

\subsubsection{The case of a sample size $n=3$}

We now turn to the analysis of $\Gamma(\lambda, \rho)$, in the specific case where the sample size is $n=3$.  We restrict ourselves to (symmetric) strategies $\sigma=(\beta,\alpha)$ such that $\beta(0)=\beta(3)=0$ and $\beta(1)=\beta(2)>0$. We identify such a strategy with the probability  $b:=\beta(1)=\beta(2)$ of acquiring information with a non-unanimous sample.

For such a strategy $b\in [0,1]$, the law of motion of $x$ in state $\theta$ is \footnote{Though $h_\theta$ depend on $b$, we write $h_\theta^b$ only when we think there is a risk of confusion.}
\[h_\theta(x)= x(1-x)\left(x(2-3b)+3b\pi_\theta-1\right).\]
The function $h_\theta$  vanishes at 0 and 1 and  has a 3rd root if $b\neq \frac23$, which is $\displaystyle \frac{3b\pi_\theta-1}{3b-2}$. 
This third root lies outside of  $(0,1)$ iff $b\in [\frac{1}{3\pi}, \frac{1}{3(1-\pi)}]$.

\subsubsection{A characterization of all invariant distributions}

We denote by $\calU$ the uniform distribution over $\Theta\times \{0,1\}$. 
\begin{proposition}\label{lemm invariant}
For each $b\in \displaystyle \left(\frac{2}{3},\frac{1}{3(1-\pi)}\right)$ and each $\lambda >0$, the Markov process  $(\theta_t,x_t)$ induced by the strategy $b$ has the following properties.
\begin{description}
\item[R1] There is a unique invariant distribution $\mu$ with the property that $\mu(\Theta\times (0,1))=1$. It is symmetric, and its conditionals $\mu_\theta$ have $C^1$ densities $f_\theta$, which solve (\ref{ODE}) on $(0,1)$.
\item[R2] The symmetric invariant distributions are exactly the convex combinations of $\mu$ and of $\calU$.
\end{description}
\end{proposition}

\noindent
\begin{proof}
We prove \textbf{R1} in two independent steps. Using the verification theorem, we first prove the existence of a distribution $\mu$ with the desired properties. Using the  general theory of Markov chains, we next show that $(\theta_t,x_t)$ has at most one invariant distribution supported by $\Theta\times (0,1)$. 
\bigskip

\noindent
\textbf{Step 1: Existence.} Fix $b$ throughout and set 
\begin{equation}\label{formula f}f_1(x):= \frac{1}{h_1(x)}\exp\left(-\lambda\int_{\frac12}^x \left\{\frac{1}{h_0(t)}+\frac{1}{h_1(t)}\right\}dt \right)\end{equation} and
\[f_0(x):= -\frac{1}{h_0(x)}\exp\left(-\lambda \int_{\frac12}^x \left\{\frac{1}{h_0(t)}+\frac{1}{h_1(t)}\right\}dt \right).\]
Since $b\in \displaystyle \left[\frac23,\frac{1}{3(1-\pi)}\right)$, one has $h_0<0<h_1$ over $(0,1)$. 
Obviously, $f_\theta>0$ on $(0,1)$,   $f_1(1-x)= f_0(x)$ for each $x\in (0,1)$, and   $(f_\theta)_{\theta \in \Theta}$ solves (\ref{ODE}). We prove that $f_\theta\in L^1(0,1)$. Since $\|f_\theta\|_1$ is independent of $\theta$ thanks to symmetry, this will imply that the pair $f_\theta:=f_\theta/\|f_\theta\|_1$, $\theta\in \Theta$, is a pair of $C^1$ densities  solving (\ref{ODE}) . The existence of $\mu$ will follow, using Lemma \ref{verif}.

To show the integrability of $f_\theta$, we note that  $f_1\geq f_0$ on $[\frac12,1)$ (resp. $f_1\leq f_0$ on $(0,\frac12]$) so that it suffices to show that $\displaystyle \int_{\frac12}^1 f_1(x) dx <+\infty$. 

Elementary computations show that  in the neighborhood of 1, the functions $1/h_1$ and $1/h_0+1/h_1$ are respectively equivalent to\footnote{The notation $f\sim_1 g$ means that $\lim_{t\to 1} f(t)/g(t)=1$.}
\[\frac{1}{h_1(t)}\sim_1 \frac{c^1_{\pi}}{1-t}\mbox{ and }\frac{1}{h_0(t)}+ \frac{1}{h_1(t)}\sim_1 \frac{c^2_{\pi}}{1-t}\]
in the neighborhood of 1, with 
\[c^1_{\pi}= \frac{1}{1-3b(1-\pi)}>0\mbox{ and } c^2_{\pi}= \frac{3b-2}{\left(1-3b(1-\pi)\right)\left(3b\pi-1\right)}>0.\]
In particular, there exist $c_{\pi}$ such that $0<c_{\pi}<c^2_{\pi}$ and $\displaystyle\frac{1}{h_0(t)}+ \frac{1}{h_1(t)}\geq \frac{c_{\pi}}{1-t}$ for every $t\in [\frac12,1)$.

Using (\ref{formula f}), this yields
\begin{equation}\label{integ}
f_1(x) \leq \frac{1}{h_1(x)} 2^{\lambda c_{\pi}} (1-x)^{\lambda c_{\pi}},\ x\in [\frac12,1).
\end{equation}
Since $\frac{1}{h_1(x)}\sim_1 \frac{c^1_{\pi}}{1-x}$, the RHS of (\ref{integ}) is integrable on $[\frac12,1)$, hence $f_1$ is integrable as well, as desired.

\bigskip\noindent
\textbf{Step 2: Uniqueness.}  Recall that $\tau_m$ is the $m$-th switching time of $(\theta_t)$. We set $Z_t=(\theta_t,x_t)$ and for $m\geq 1$, $\tilde Z_m:=Z_{\tau_m}$, the sampled process at the switching times.
We view $(\tilde Z_m)$ as a Markov chain with state space $\Theta\times (0,1)$, and we denote $\tilde\prob$ its kernel.\footnote{That is, $\tilde\prob(z,A)$ denotes the probability of $\tilde Z_1\in A$, given $\tilde Z_0=z$, and $\tilde\prob_z$ is the law of the sequence $(\tilde Z_m)$, when starting from $\tilde Z_0= z$.}

By Proposition 2.4 in Bena\"im et al. (2014), the sets $\calP_{inv}$ and $\tilde \calP_{inv}\subset \Delta(\Theta\times (0,1))$ of invariant distributions for $(Z_t)$ and $(\tilde Z_m)$ respectively are homeomorphic.\footnote{Bena\"im et al. (2014) assume a compact state space. However, as far as their Proposition 2.4 is concerned, compactness is only used to ensure that $\calP_{inv}\neq\emptyset$.} Thus, \textbf{R1} will follow if we prove  that $|\tilde\calP_{inv}|\leq 1$. 

Let $a\in (0,\frac12)$ be arbitrary,  set $C=\Theta\times [a,1-a]$, and denote by $\sigma_C:=\inf\{m\geq 1: \tilde Z_m\in C\}$ the first return to $C$.

The fact that $(\tilde Z_m)$ has at most one invariant measure follows from Claims \ref{claim1} and \ref{claim2} below, using  results and terminology from \cite{douc}.  Indeed,
  Claim \ref{claim1} states  that $C$ is \emph{accessible} (see Definition 3.5.1), while Claim  \ref{claim2} states that $C$ is a \emph{petite set} (see Definition 9.4.1). By Lemma 9.4.3, this implies that $\tilde\prob$ is an irreducible Markov kernel. By Corollary 9.2.16, it admits at most one invariant measure.\footnote{Definitions 3.5.1, 9.4.1, and Lemma 9.4.3 are from 
  \cite{douc}.}

\begin{claim}\label{claim1}
One has $\tilde\prob_z(\sigma_C<+\infty)>0$ for each $z\in \Theta\times (0,1)$. 
\end{claim}

\begin{claim}\label{claim2}
There exists a non-zero measure $\nu\in\Delta(\Theta\times (0,1))$ such that 
\[\tilde \prob(z, A)+\tilde\prob^2(z,A)\geq \nu(A)\]
for every event $A\subset \Theta\times (0,1)$ and each $z\in C$.
\end{claim}

\noindent
\begin{proof}[Proof of Claim \ref{claim1}] Let $z=(\theta_0,x_0)\in \Theta\times (0,1)$ be given, and assume for concreteness that $\theta_0=0$. 
We prove that given $\tilde Z_0=z$, there is a positive probability that 
$\tilde Z_2\in C$. Thus,  $\tilde \prob_z(\tilde Z_2\in C)=\tilde \prob^2(z,C)>0$, which implies the result.

Recall that 
\[\tilde Z_1=(\theta_{\tau_1},x_{\tau_1})= (1,\Phi^0(u_0,x_0))\mbox{ and } \tilde Z_2=(\theta_{\tau_2}, x_{\tau_2})= (0,\Phi^1_{u_1}\circ \Phi^0_{u_0}(x_0)).\]
Since $h_0<0<h_1$ on $(0,1)$, the map $t\mapsto \Phi^0(t,x_0)$ is a diffeomorphism from $(0,+\infty)$ to $(0,x_0)$ and the map $t\mapsto \Phi^1(t,x)$ a diffeomorphism from $(0,+\infty)$ to $(x,1)$, for each $x$.
Since $u_0,u_1\sim \calE(\lambda)$, the r.v. $\Phi^0(u_0,x_0)$ has a density, positive on $(0,x_0)$ and similarly, $\Phi^1(u_1,x_1)$ has a density, positive on $(x_1,1)$ (for each $x_1)$. 
We stress that $\Phi^0(u_0,x_0)$ is the solution to the ODE $x'(t)= h_0(x(t))$, $x(0)=x_0$, evaluated at the random time $u_0$. Since $u_0$ and $u_1$ are independent, this implies that $\Phi^1_{u_1}\circ\Phi^0_{u_0}(x_0)$ has a density, positive on $(0,1)$, hence $\tilde\prob_z(\tilde Z_2\in C)>0$. 
\end{proof}

\medskip\noindent
\begin{proof}[Proof of Claim \ref{claim2}]
Fix $\underline{a},\bar a\in (0,1)$ such that $\underline{a}<\bar a<a$, and let $\nu$ be the uniform distribution on 
$B_0\cup B_1$, where 
\[B_0:= \{0\}\times [1-\bar a,1-\underline{a}]\mbox{ and } B_1:=\{1\}\times [\underline{a},\bar a].\]
We prove the existence of $c>0$ s.t. 
\begin{equation}\label{goal}\tilde\prob(z,A)+\tilde\prob^2(z,A)\geq c\nu(A), \mbox{ for each $z\in C$ and each $A$.} \end{equation}

For given $x_0\in (0,1)$, the density of $\Phi^0(u_0,x_0)$ is given  by 
\[\psi_{x_0}(x):=\frac{\lambda e^{-\lambda t}}{|h_0\left(\Phi^0(t,x_0)\right)|},\ x\in (0,x_0)\]
where $t:=t_{x_0,x}$ solves $\Phi^0(t,x_0)=x$: $t$ is the amount of time to reach $x$ starting from $x_0$.

On the compact set $[a,1-a]\times [\underline{a},\bar a]$, the map $(x_0,x)\mapsto t_{x_0,x}$ is continuous and $\psi_{x_0}(x)$ positive, hence there exists $c_0>0$ s.t. $\phi_{x_0}(x)\geq c_0>0$ for every $(x_0,x)\in  [a,1-a]\times [\underline{a},\bar a]$. This implies 
\begin{equation}\label{eq21}\tilde \prob_{0,x_0}(\tilde Z_1\in A) \geq c_0 \nu(A\cap B_1),\mbox{ for each $x_0\in [a,1-a]$.}\end{equation}

A similar argument shows the existence of $c_1\in (0,1)$ such that 
\begin{equation}\label{eq22}\tilde \prob_{1,x_0} (\tilde Z_1\in A)\geq c_1 \nu(A\cap B_0)\mbox{ for each }x_0\in [\underline{a},\bar a].\end{equation}
Combining (\ref{eq21}) and (\ref{eq22}) shows that the inequality
\begin{equation}\label{eq23}
\tilde\prob(z, A)+\tilde\prob^2(z, A)\geq c_0c_1\nu(A)
\end{equation}
holds for each $z=(0,x_0)$ with $x_0\in [a,1-a]$ and by symmetry, also holds for $z=(1,x_0)$. 
\end{proof}
\bigskip

The proof of \textbf{R2} is easy.  The set $\Theta\times \{a\}$ is absorbing for each $a\in A$, and the unique invariant distribution is the uniform distribution $\calU_a$ over the two-point set $\Theta\times \{a\}$. Using \textbf{R1}, the set of invariant distributions is therefore the convex hull of $\{\mu,\calU_1,\calU_0\}$. The symmetric ones are the convex combinations of $\mu$ and $\calU$.
\end{proof}

\subsubsection{Interim beliefs and the existence of an equilibrium exhibiting social learning.}

For a strategy $b\in [\frac{2}{3},\frac{1}{3(1-\pi)})$ and to avoid confusion, we denote by $\mu^b\in \Delta\left(\Theta\times (0,1)\right) $ the unique invariant measure  that is supported by $\Theta\times (0,1)$, and we denote by $f_\theta^b$ the density of $\mu^b_\theta$, as given by \textbf{R1} in Proposition \ref{lemm invariant}. 

For $k\in \{0,1,2,3\}$, we define  $LR_k(b)$ to be the likelihood ratio of one's belief with a sample of $k$, assuming $(\theta,x)\sim \mu^b$. Formally, 

\[LR_k(b):= \frac{\int_0^1 x^k(1-x)^{3-k}  d\mu^b_1(x)}{\int_0^1 x^k(1-x)^{3-k} d\mu^b_0(x)} = \frac{\int_0^1 x^k(1-x)^{3-k}  f^b_1(x)dx}{\int_0^1 x^k(1-x)^{3-k} f^b_0(x)dx.}\]

The next lemma summarizes the key properties of $LR_k$.

\begin{lemma}\label{lemm LR}
The following properties hold: 
\begin{description}
\item[P1] For each $k$, the function $LR_k$ is continuous over $\left[\frac23,\frac{1}{3(1-\pi)}\right)$.
\item[P2] For each $b\in (\frac23,\frac{1}{3(1-\pi)})$  $LR_k(b)$ is (strictly) increasing in $k$.
\item[P3] $LR_2(\frac23)=1$ for each $\lambda$.
\item[P4] $\lim_{b\to \frac{1}{3(1-\pi)}}\lim_{\lambda\to 0}LR_2(b)=+\infty$.
\end{description}
\end{lemma}

\begin{proof}
We start with \textbf{P1}. Observe that $b\mapsto f_\theta^b(x)$ is continuous for each $x\in (0,1)$. Since $x^k(1-x)^{3-k}f^b_\theta(x)>0$ for each $x\in (0,1)$, the  (general) dominated convergence theorem implies that $b\mapsto\displaystyle \int_0^1x^k(1-x)^{n-k} f_\theta^b(x) dx$ is continuous. 
\medskip

We turn to \textbf{P2}. Let $k<3$ and denote by $\tilde \mu^b_\theta\in \Delta((0,1))$ the probability measure with density $\displaystyle \tilde f_\theta^b(x):= \frac{x^k(1-x)^{2-k} f_\theta^b(x)}{\int_0^1x^k (1-x) ^{2-k} f_\theta^b(x) dx}$. For $b\in \left(\frac23,\frac{1}{3(1-\pi)}  \right)$, 
\[\frac{f_1^b(x)}{f_0^b(x)}= -\frac{h_0^b(x)}{h_1^b(x)}= \frac{1-3b(1-\pi)+x(3b-2)}{3b\pi-1 -x(3b-2)}.\]
Thus, the ratio $\tilde f_1^b/\tilde f_0^b$ is increasing, which implies that $\tilde \mu^b_1$ (strictly) dominates $\tilde \mu_0^b$ in the first-order sense. 
Since $x\mapsto 1-x$ and $x\mapsto x$ are resp. decreasing and increasing, one obtains
\[\frac{\int_0^1 (1-x)d\tilde \mu_1^b(x)}{\int_0^1 (1-x)d\tilde \mu_0^b(x)} < 1 < \frac{\int_0^1 x d\tilde \mu_1^b(x)}{\int_0^1 x d\tilde \mu_0^b(x)},\]
which implies in turn $LR_k(b)<LR_{k+1}(b)$. 
\medskip

\textbf{P3} is straightforward. Indeed, for $b=\frac23$, $h_0^b(\cdot)+h_1^b(\cdot)=0$, hence $f_0^b=f_1^b$.
\bigskip\noindent

We conclude with \textbf{P4}. Define $h_\theta^{b,*}(x):=\displaystyle \frac{h_\theta^b(x)}{x(1-x)}$, and observe that $h_\theta^{b,*}(\cdot)>0$ on the \emph{closed} interval $[0,1]$. \textbf{P4} follows from Claim \ref{claim3} below, since $\displaystyle  -\frac{h_0^{b,*}(1)}{h_0^{b,*}(1)}= \frac{3b\pi-1}{1-3b(1-\pi)}$, which converges to $+\infty$ when $b\to \displaystyle \frac{1}{3(1-\pi)}$. 

\begin{claim}\label{claim3}
For $b\in (\frac23,\frac{1}{3(1-\pi)}) $, one has 
\[\lim_{\lambda\to 0} LR_2(b)= -\frac{h_0^{b,*}(1)}{h_1^{b,*}(1)}\]
\end{claim}

\begin{proof}[Proof of the claim]
Define  $\displaystyle g(t):=\int_{\frac12}^t \left\{\frac{1}{h_0(x)}+\frac{1}{h_1(x)}\right\}dx$. Observe that $g>0$ on $(0,1)$, $g(t)=g(1-t)$ for each $t\in (0,1)$, and that $\lim_{t\to 1} g(t)=+\infty$. These properties ensure that 
the measure $\nu_\lambda$ on $(0,1)$ with density $\displaystyle \frac{e^{-\lambda g(t)}}{\int_{[0,1]} e^{-\lambda g}}$, converges (weakly) to $\nu_*:=\displaystyle \frac12\left(\delta_0 +\delta_1\right)$ as $\lambda\to 0$. 

Since 
\[LR_2(b)= \frac{\int_0^1 \frac{x}{h_1^{b,*}(x)} \exp\left(-\lambda g(x)\right) dx}{\int_0^1 \frac{-x}{h_0^{b,*}(x)} \exp\left(-\lambda g(x)\right) dx}= 
\frac{\int_0^1\frac{x}{h_1^{b,*}(x)} d\nu_\lambda(x)}{\int_0^1\frac{-x}{h_0^{b,*}(x)} d\nu_\lambda(x)},\]

it follows that\footnote{Since the limit denominator is non-zero.}
\[\lim_{\lambda \to 0}LR_2(b)= \frac{\int_0^1\frac{x}{h_1^{b,*}(x)} d\nu_*(x)}{\int_0^1\frac{-x}{h_0^{b,*}(x)} d\nu_*(x)}= -\frac{h_0^{b,*}(1)}{h_1^{b,*}(1)}.\]
\end{proof} 

This concludes the proof of Lemma \ref{lemm LR}.
\end{proof}
\medskip

\subsubsection{Conclusion of the continuous-time analysis}

Proposition \ref{prop eq CT} below concludes, by showing the existence of an ESS of $\Gamma(\lambda,\rho)$ with a welfare that exceeds $\hat p$. 

\begin{proposition}\label{prop eq CT}
For $\lambda>0$ small enough, there  exists  $b_*\in \displaystyle \left(\frac23,\frac{1}{3(1-\pi)}\right)$, and  $\bar a\in (0,1)$ such that the following holds. 
\begin{itemize}
\item $(b_*,\mu)$ is an ESS of $\Gamma (\lambda,\rho)$ if and only if $\mu=a \calU +(1-a) \mu^{b_*}$ for some $a\leq \bar a$.
\end{itemize}
In addition, there exists $\underline{b}\in (\frac23,b_*)$ and $\bar b\in (b_*,\frac{1}{3(1-\pi)})$ such that 
\begin{equation}\label{eq 24}LR_2(\underline{b}) < \frac{\hat p}{1-\hat p} < LR_2(\bar b).\end{equation}
\end{proposition}

\begin{proof}
The result follows easily from Lemma \ref{lemm LR}. By  \textbf{P4}, there exist $\bar b$ and $\lambda_0$ such that $LR_2(\bar b)>\frac{\hat p}{1-\hat p}$ for all $\lambda <\lambda_0$. For such a $\lambda$, the existence of $\underline{b},b_*$, such that 
$\frac23<\underline{b}< b_*<\bar b$ and $LR_2(\underline{b})<\frac{\hat p}{1-\hat p}=LR_2(b_*)$ follows from \textbf{P3} and the intermediate value theorem. 
Since $LR_3(b_*)>LR_2(b_*)$, the pair  $(b_*,\mu^{b_*})$ is an ESS of $\Gamma(\lambda, \rho)$. 

Let now $\mu$ be an invariant distribution for $b_*$. By Lemma \ref{lemm invariant}, $\mu= a\calU +(1-a)\mu^{b_*}$ for some $a\in [0,1]$. Write $l_k(a)$ for the likelihood ratio of one's belief with a sample of $k$, assuming $(\theta, x) \sim a\calU +(1-a)\mu^{b_*}$. Since the probability of a 'confusing' sample $k=2$ is 0 under $\calU$, $l_2(a)$ is independent of $a$ and thus,  $l_2(a)=l_2(0)=\frac{\hat p}{1-\hat p}$ for each $a$. Hence $(b_*,\mu)$ is an  ESS if and only if $l_3(a)\geq \frac{\hat p}{1-\hat p}$. Since $l_3$ is continuous and decreasing, with $l_3(0)>\frac{\hat p}{1-\hat p}> 1= l_3(1)$, the existence of $\bar a$ follows. 
\end{proof}

\subsection{Convergence to continuous-time}

This section contains the (mostly technical) issues related to the convergence of the discrete-time analysis to the continuous-time one.  In Section \ref{sec invariant}, we show that if $(\sigma_\ep)$ is a family of strategies that converges to $\sigma$, then invariant measures $\mu^\ep$ for $\sigma_\ep$ in $G(\lambda\ep,\rho\ep)$ converge to invariant measures $\mu$ for $\sigma$ in $\Gamma(\lambda, \rho)$. In Section \ref{sec tight}, we prove that the limit $\mu$ of distributions that are concentrated on $\Theta\times (0,1)$ is also concentrated on $\Theta\times (0,1)$, thereby excluding cases where the limit distribution might put a positive weight on the extreme points. Armed with these technical preliminaries, we show in Section \ref{sec conclusion} that for $\ep$ small enough, the game $G(\lambda\ep,\rho\ep)$ has an equilibrium $(b^\ep,\mu^\ep)$ that is 'close' to the equilibrium $(b_*,\mu^{b_*})$ of Proposition \ref{prop eq CT}, and conclude.

\subsubsection{Limits of invariant measures are invariant measures}\label{sec invariant}

This section is devoted to  Proposition \ref{prop convergence} below. 
\begin{proposition}\label{prop convergence}
For $\ep>0$, let  a strategy $\sigma_\ep$, and a distribution $\mu_\ep\in \Delta(\Theta\times [0,1])$ be given. Assume that $(\sigma_\ep)_\ep$ and $(\mu_\ep)_\ep$ converge, with limits $\sigma=\lim_{\ep\to 0} \sigma_\ep$ and 
$\mu=\lim_{\ep\to 0} \mu_\ep$.

If $\mu_\ep$ is invariant for $\sigma_\ep$ in $G(\lambda\ep,\rho\ep)$, then $\mu$ is invariant for $\sigma$ in $\Gamma(\lambda,\rho)$.
\end{proposition}

\begin{proof}
We let $\sigma_\ep,\sigma,\mu_\ep$ and $\mu$ be given as stated. As above, we denote by $(z_t)_{t\geq 0}=(\theta_t,x_t)_{t\geq 0}$ the continuous-time Markov process induced by $\sigma$ in $\Gamma(\lambda,\rho)$, and by $\prob$ its kernel. In particular, $\prob^t_{z_0}$ is the law of $z_t$ given $z_0$.

For $\ep>0$, we denote by  $(\tilde z_m^\ep)_{m\in \dN}:= (\tilde \theta_m^\ep,\tilde x_m^\ep)_{m\in \dN}$ the discrete-time, Markov chain induced by $\sigma_\ep$ in $G(\lambda\ep,\rho \ep)$. Following the interpretation of $G(\lambda\ep,\rho\ep)$ as a discretized version of $\Gamma(\lambda,\rho)$, we identify $(\tilde z^\ep_m)$  with the c\`adl\`ag, continuous-time process $(z_t^\ep)=  (\theta_t^\ep,x_t^\ep)_{t\geq 0}$ defined by 
$z^\ep_t:= \tilde z^\ep_{\lfloor \frac{t}{\ep}\rfloor}$.\footnote{To avoid confusion between time and stages, we restrict the use of $t$ to denote the underlying, 'physical' continuous-time, and we label the stages of $G(\lambda \ep, \rho\ep)$ as $m=0,1,2,\ldots$. Thus, the $m$-th stage of $G(\lambda\ep,\rho \ep)$ takes place at time $m\ep$.} We also denote by $\prob_{z_0}^{t,\ep}$ the law of $z^\ep_t$ when starting from  $z_0$. Although $(z^\ep_t)$ is not a continuous-time Markov process, the distribution $\mu_\ep$ can nevertheless be thought of as an invariant measure for $(z_t)$, in the sense that $\mu_\ep \prob^{t,\ep}_{\cdot}= \mu_\ep$ for each $t>0$.\footnote{This follows since $\mu_\ep$ is an invariant distribution for the Markov chain $(\tilde z_m^\ep)$.}

We need to prove that $\mu \prob^tf = \mu f $ for each continuous function $f\in C(\Theta\times [0,1])$. Since $(\mu_\ep)$ weakly converges to $\mu $, one has $\lim_{\ep\to 0} \mu_\ep f = \mu f$. Since $\mu_\ep \prob^{t,\ep}_{\cdot}f = \mu_\ep f$, we actually need to prove that 
\begin{equation}\label{goal conv}
\lim_{\ep\to 0} \int_{\Theta\times [0,1]} \E_{(\theta, x)\sim \prob^{t,\ep}_{z_0}}[f(\theta,x)] d\mu_\ep(z_0)=  \int_{\Theta\times [0,1]} \E_{(\theta, x)\sim \prob^{t}_{z_0}}[f(\theta,x)] d\mu(z_0).
\end{equation}

Recall from (\ref{explicit}) that $z_t$ is a function of $z_0$ and of the sojourn times $(u_i)$, which are independent and exponentially distributed. 
A similar representation holds for $z_t^\ep$ as well. We denote the successive jump and sojourn times of $(\theta_t^\ep)$ by $(\tau_i^\ep)$ and $(u_i^\ep)$. Note that the random variables $(u^\ep_i/\ep)$ are idd geometric variables with parameter $\lambda\ep$. 

If $t<\tau_1^\ep$, $x^\ep_t=\tilde x^\ep_{\lfloor\frac{t}{\ep}\rfloor}$ is obtained by applying $id + \ep h_{\theta_0}^{\sigma_\ep}$ iteratively: 
\[x_t^\ep:=\Phi^{\theta_0,\ep}(t,x_0):=\left(id + \ep h_{\theta_0}^{\sigma_\ep}\right)\circ \cdots \circ \left(id + \ep h_{\theta_0}^{\sigma_\ep}\right)(x_0),\]
where $id + \ep h_{\theta_0}^{\sigma_\ep}$ is applied $\lfloor\frac{t}{\ep}\rfloor$ times. 

More generally, $x_t^\ep$ is given  by 
\[x_t^\ep= \Phi^{\alpha_k,\ep}_{t-\tau^\ep_k}\circ \cdots \circ \Phi^{\alpha_0,\ep}_{u_0^\ep}(x_0) \mbox{ on the event $\tau_k^\ep \leq t <\tau^\ep_{k+1}$}.\]

We next introduce a coupling between $(z_t)$ and $(z_t^\ep)$, that allows us to view both $z_t$ and $z_t^\ep$ as functions of the exponential variables $(u_i)$.
For fixed $\ep>0$, set $v_i^\ep:=j\ep$ whenever $a_{j-1}\leq u_i<a_j$, with 
$a_j:= \displaystyle -\frac{j}{\lambda}\ln(1-\lambda \ep)$. 
It is straightforward to check that $\frac{v_i^\ep}{\ep}=j$ follows a geometric distribution with parameter $\lambda\ep$. In addition, the r.v.'s $(v_i^\ep)$ inherit independence from the independence of $(u_i)$. Therefore, the distributions of the sequences $(u_i^\ep)$ and $(v_i^\ep)$ coincide. We may thus assume that $(u_i)$ and $(u_i^\ep)$ are defined on the same probability space, with $u_i^\ep$ be given by $v_i^\ep$. In addition, standard calculations show that $|u_i^\ep-u_i|\leq \lambda \ep u_i^\ep$ for $\ep$ small enough, which implies 
\[|u_i^\ep -u_i|\leq 2\lambda \ep u_i^\ep,\mbox{ for every }\ep \mbox{ small enough.}\]

Denote by $\bf{Q}$ the law of the sequence $\vec{u}=(u_i)_{i\in \dN}$, and let   $N:=\max\{p:\tau_p\leq t\}$ and $N_\ep:=\max\{p: \tau_p^\ep\leq t\}$ stand for  the number of state transitions prior to $t$ for $(\theta_t)$ and $(\theta^\ep_t)$. Then   $\mu \prob^t f$ and $\mu_\ep \prob^{t,\ep}f$ are given by

\[\mu \prob^t f = \E_{\vec{u}\sim\bf{Q}}\E_{(\theta_0,x_0)\sim \mu}\left[f\left(\Phi^{\alpha_N}_{t-\tau_N}\circ \Phi^{\alpha_{N-1}}_{u_{N-1}}\circ \cdots \circ \Phi^{\alpha_0}_{u_0}(x_0)\right)\right]\]
and 

\[\mu_\ep \prob^{t,\ep}f=\E_{\vec{u}\sim\bf{Q}}  \E_{(\theta_0,x_0)\sim \mu_\ep}\left[f\left(\Phi^{\alpha_{N_\ep},\ep}_{t-\tau^\ep_{N_\ep}}\circ \cdots \circ \Phi^{\alpha_0,\ep}_{u_0^\ep}(x_0)\right)\right].\]
For fixed $\vec{u}$, it is easily checked, using standard results on ODE, that $f\left(\Phi^{\alpha_{N_\ep},\ep}_{t-\tau^\ep_{N_\ep}}\circ \cdots \circ \Phi^{\alpha_0,\ep}_{u_0^\ep}(x_0)\right)$ converges to $f\left(\Phi^{\alpha_N}_{t-\tau_N}\circ \Phi^{\alpha_{N-1}}_{u_{N-1}}\circ \cdots \circ \Phi^{\alpha_0}_{u_0}(x_0)\right)$, uniformly over $(\theta_0,x_0)$, which implies 
\[\lim_{\ep \to 0}  \E_{(\theta_0,x_0)\sim \mu_\ep}\left[f\left(\Phi^{\alpha_{N_\ep},\ep}_{t-\tau^\ep_{N_\ep}}\circ \cdots \circ \Phi^{\alpha_0,\ep}_{u_0^\ep}(x_0)\right)\right]= \E_{(\theta_0,x_0)\sim \mu}\left[f\left(\Phi^{\alpha_N}_{t-\tau_N}\circ \Phi^{\alpha_{N-1}}_{u_{N-1}}\circ \cdots \circ \Phi^{\alpha_0}_{u_0}(x_0)\right)\right].\]
The result follows by dominated convergence.
\end{proof}

\subsubsection{A tightness result}\label{sec tight}

Fix $\lambda_0>0$ be such that the conclusion of Proposition \ref{prop eq CT} holds for each $\lambda<\lambda_0$. We let   $\lambda<\lambda_0$ be given, and we then fix $\frac23<\underline{b}<\bar b<\frac{1}{3(1-\pi)}$ such that (\ref{eq 24}) holds. 

In this section, we represent a distribution $\mu\in \Delta(\Theta\times [0,1])$ by its conditional cdf $F_\theta$ in each state $\theta\in  \Theta $.

Given $b$, we define an operator 
$T_{b,\ep}:\Delta(\Theta\times [0,1])\rightarrow \Delta(\Theta\times [0,1])$ by the equation
\[
\left(T_{b,\ep})F\right)_{\theta}\left(x\right)=\left(1-\ep\lambda\right)F_{\theta}\left(g_{\theta}^{-1}\left(x\right)\right)+\ep \lambda F_{1-\theta}\left(g_{\theta}^{-1}\left(x\right)\right).\footnote{For each $b\in (\frac23,\frac{1}{3(1-\pi)})$, the function $g_\theta$ is an increasing bijection on $[0,1]$, hence $g_\theta^{-1}$ is well-defined.}
\]
The operator $F$ is weakly continuous and any fixed-point of $T_{b,\ep}$ is an invariant distribution for the strategy $b$. 

For each $A_{0},A_{1},\gamma>0$, we define  $\Delta_{A_{0},A_{1},\gamma}\subset \Delta(\Theta\times [0,1])$
to be those symmetric distributions  $\left(F_{0},F_{1}\right)\in \Delta(\Theta\times [0,1])$ such that 
\[
F_{\theta}\left(x\right)\leq A_{\theta}x^{\gamma},\mbox{ for each $x$ and $\theta\in \Theta$}.
\]
Notice that such distributions do not put any weight on  $x=0$
 and $x=1$.

\begin{proposition}\label{prop Marcin}
There exist $A_0,A_1,\gamma$ and $\ep_0$ such that the following holds. For every $b\in [\underline{b},\bar b]$ and every $\ep<\ep_0$, the process $(\tilde \theta_m^b,\tilde x_m^b)$ induced by the strategy $b$ in 
$G(\lambda\ep,\rho \ep)$ has an invariant distribution in $\Delta_{A_0,A_1,\gamma}$.
\end{proposition}

Proposition \ref{prop Marcin} shows the existence of invariant distributions with some regularity, uniformly over $b$ and $\ep$. It ensures in particular that weak limits of such distributions don't put any weight on $x=0$ and $x=1$.

The proof will make use of the following simple observation.
\begin{lemma}
\label{lem:properties of h} There exists $x>0$ and $c_0<0<c_1$ with $c_0+c_1>0$ such that 
\[h_\theta(x) \geq \ep\rho c_\theta x,\mbox{ for each  $b\in\left[\underline{b},\bar b\right]$ and  $x<x_0$}.\]
\end{lemma}
\begin{proof}
Choose $\ep_1>0$ such that $2\ep_1 < 3\underline{b}-2$, and set $c_\theta := 3 \underline{b}\pi_\theta -1 -\ep$. By the choice of $\ep_1>0$, $c_0+c_1= 3\underline{b}-2 -2\ep_1 >0$. 
On the other hand, 
\[h^*_\theta(x):= \frac{h_\theta(x)}{\ep \rho x(1-x)}:= 3 b\pi_\theta +3 x (1-b) -(1+x)= c_\theta +\ep -x(3b +1).\]
Hence $h^*(x) \geq c_\theta$ as soon as $x < x_0:=\ep_1/(3\bar b+1)$. 
\end{proof}

We fix $x_0$, $c_0$ and $c_1$ as given by Lemma \ref{lem:properties of h}. 
The central step in the proof is Lemma \ref{lem:regularization} below.
\begin{lemma}
\label{lem:regularization}There exist $A_0,A_1,\gamma$ and $\ep_0$ such that the following holds. For
each $b\in\left[\underline{b},\bar b\right]$ and $\ep <\ep_0$, one has
\[
T_{b,\ep}\left(\Delta_{A_{0},A_{1},\gamma}\right)\subseteq\Delta_{A_{0},A_{1},\gamma}.
\]
\end{lemma}

\begin{proof}
We first show the existence of  $A_{0}^{\prime},A_{1}^{\prime}$ and $\gamma>0$ such that 
\begin{align}
A_{0}^{\prime} & \geq\left(1-\ep\lambda\right)A_{0}^{\prime}\left(1+\ep\rho c_{0}\right)^{-\gamma}+\ep\lambda A_{1}^{\prime}\left(1+\ep\rho c_{0}\right)^{-\gamma},\label{eq:reg inequalities}\\
A_{1}^{\prime} & \geq\left(1-\ep\lambda\right)A_{1}^{\prime}\left(1+\ep\rho c_{1}\right)^{-\gamma}+\ep\lambda A_{0}^{\prime}\left(1+\ep\rho c_{1}\right)^{-\gamma}\nonumber 
\end{align}
for all $\ep>0$ small enough.
Dividing  both sides of the two inequalities by $A_{1}^{\prime}$, denoting $A=A_0^\prime/A_1^\prime$, after some algebra, we obtain
\begin{align}
A \left(\left(1+\ep\rho c_{0}\right)^{\gamma} - \left(1-\ep\lambda\right)  \right) & \geq \ep\lambda,\nonumber \\
\left(1+\ep\rho c_{1}\right)^{\gamma} - \left(1-\ep\lambda\right) & \geq \ep\lambda A \nonumber 
\end{align}
Combining the two inequalities and some further algebra shows that the existence of required  $A_{0}^{\prime},A_{1}^{\prime}$ and $\gamma>0$  is equivalent to
\[ \left(\left(1+\ep\rho c_{1}\right)^{\gamma}-(1-\ep\lambda)\right) \left(\left(1+\ep\rho c_{0}\right)^{\gamma}-(1-\ep\lambda)\right) \geq \ep^2\lambda^2.\]
for all $\ep>0$ small enough.

A first-order Taylor expansion of the LHS as $\ep\to 0$ shows that this holds as soon as 
\[\left(\gamma \rho c_0+\lambda\right)\left(\gamma \rho c_1+\lambda\right) > \lambda^2,\]
or equivalently, $\gamma\rho c_0c_1 +\lambda\left(c_0+c_1\right)>0$, which holds for $\gamma$ small enough.

Fix constants $A_{0}^{\prime},A_{1}^{\prime},\gamma$ that satisfy inequalities (\ref{eq:reg inequalities}),
and  $a>\frac{1}{\min A_{\theta}x_{0}^{\gamma}}$, then set $A_{\theta}=aA_{\theta}^{\prime}$
for each state $\theta$. 

Let now  $F=\left(F_{0},F_{1}\right)\in\Delta_{A_{0}^{\prime},A_{1}^{\prime},\gamma}$,  $b\in [\underline{b},\bar b]$ and $\ep<\ep_0$ be arbitrary. We check that $T_{b,\ep}F\in \Delta_{A_0,A_1,\gamma}$.

For $x\leq\frac{1}{2}x_{0}\leq\frac{1}{1+\rho c_{0}}x_{0}$, one has 
\begin{align*}
T_{b,\ep}F_{\theta}\left(x\right) & =\left(1-\ep\lambda\right)F_{\theta}\left(g_{\theta}^{-1}\left(x\right)\right)+\ep\lambda F_{\theta}\left(g_{1-\theta}^{-1}\left(x\right)\right)\\
 & \leq\left(1-\ep\lambda\right)aA_{\theta}^{\prime}\left(\frac{1}{1+\ep\rho c_{\theta}}\right)^{\gamma}x^{\gamma}+\ep\lambda aA_{1-\theta}^{\prime}\left(\frac{1}{1+\ep\rho c_{1-\theta}}\right)^{\gamma}x^{\gamma}\\
 & \leq aA_{\theta}^{\prime}x^{\gamma}=A_{\theta}x^{\gamma}.
\end{align*}
For $x\geq x_0$, one has $T_{b,\ep}F_{\theta}(x)\leq 1\leq A_{\theta}x_{0}^{\gamma}\leq A_{\theta}x^{\gamma}$, by the choice of $a$.
This conclude the proof of the lemma.
\end{proof}
\bigskip

We conclude the proof of Proposition \ref{prop Marcin}.  For $\ep<\ep_0$ and $b\in [\underline{b},\bar b]$, $T_{b,\ep}$ maps $\Delta_{A_0,A_1,\gamma}$  into itself. Since $T_{b,\ep}$ is weakly continuous, and $\Delta_{A_0,A_1,\gamma}$ is weakly compact, $T_{b,\ep}$ has a fixed point.

\subsubsection{Conclusion}\label{sec conclusion}

We conclude the proof of Theorem \ref{thm blurry}, using the previous section. We first argue that 
for each $\ep$ small enough, there exists $b_\ep\in (\underline{b},\bar b)$ and an invariant distribution $\mu_\ep$ for $b_\ep$ in $G(\lambda\ep,\rho\ep)$ such that the interim belief with a sample $k=2$ is $\hat p$. We next prove that, if $\ep$ is small enough, the interim belief with a belief $k=3$ is higher than $\hat p$, and deduce that $(\mu_\ep,b_\ep)$ has the desired properties.

To start with, fix $\ep<\ep_0$. 
Given $b$, and an invariant distribution $\mu$ for $b$ in $G(\lambda\ep,\rho\ep)$  we denote by $LR_2^\ep(b,\mu)$ the likelihood ratio of one's interim  belief  with a sample $k$, assuming $(\theta,x)\sim \mu$. 
Define the set-valued map $\Psi^\ep:[\underline{b},\bar b]\to \dR_+$ as the set of all possible interim beliefs for $k=2$, when considering all invariant measures in $\Delta_{A_0,A_1,\gamma}$. 
\[\Psi^\ep(b):=\{LR_2^\ep(b,\mu):\ \mu\in \Delta_{A_0,A_1,\gamma},\ \mu\mbox{ is  invariant for }b\}.\]

By Proposition \ref{prop Marcin}, $\Psi^\ep$ is non-empty valued. In addition, since $\Delta_{A_0,A_1,\gamma}$ is convex and weakly compact, the map $\Psi^\ep(\cdot)$ is upper hemi-continuous, with convex values. This implies that the range $I:=\cup_b \Psi^\ep(b)$ of $\Psi^\ep$ is a closed interval. 

We claim that for $\ep$ small enough, the interval $I$ contains $\displaystyle \frac{\hat p}{1-\hat p}$. 
Indeed,  if $\mu_\ep\in \Delta_{A_0,A_1,\gamma}$ is invariant for $\underline{b}$ in $G(\lambda\ep,\rho\ep)$, then any limit point $\mu$ of $(\mu_\ep)$ is  invariant for $\underline{b}$ in $\Gamma(\lambda,\rho)$ by Proposition \ref{prop convergence}, and still belongs to $\Delta_{A_0,A_1,\gamma}$ by compactness. By Proposition \ref{lemm invariant}, this implies $\mu= \mu^{\underline{b}}$ and $\lim_{\ep\to 0}LR_2^\ep(\underline{b},\mu_\ep)=LR_2(\underline{b})<\frac{\hat p}{1-\hat p}$. 
A similar argument shows that if $\mu_\ep\in \Delta_{A_0,A_1,\gamma}$ is invariant for $\bar b$ in $G(\lambda \ep,\rho\ep)$, then $
\lim_{\ep\to 0}LR_2^\ep(\bar{b},\mu_\ep)=LR_2(\bar{b})>\frac{\hat p}{1-\hat p}$. 
Hence there is $\ep_1<\ep_0$ such that $\frac{\hat p}{1-\hat p}\in I$. 

For $\ep<\ep_1$, let $b_\ep \in [\underline{b},\bar b]$ and an invariant distribution $\mu_\ep\in \Delta_{A_0,A_1,\gamma}$ for $b_\ep$ in $G(\lambda\ep,\rho\ep)$ be such that $LR_2^\ep(b_\ep,\mu_\ep)=\frac{\hat p}{1-\hat p}$, 
and let $(b_*,\mu_*)$ be an arbitrary limit point of $(b_\ep,\mu_\ep)$. As above, one necessarily has $\mu_*= \mu^{b_*}$, hence \footnote{Where the limit is taken along a sequence $(b_\ep,\mu_\ep)$ that converges to $(b_*,\mu_*)$.} $\lim_{\ep\to 0}LR_k^\ep(b_\ep,\mu_\ep)= LR_k(b_*) $ for $k=2,3$. 

By the choice of $(b_\ep,\mu_\ep)$, it follows that $LR_2(b_*)= \frac{\hat p}{1-p}$. By Lemma \ref{lemm LR}, \textbf{P2}, this implies $LR_3(b_*)>\frac{\hat p}{1-p}$ and therefore, $LR_3^\ep(b_\ep,\mu_\ep)>\frac{\hat p}{1-\hat p}$ for all $\ep$ small enough.
The pair  $(b_\ep,\mu_\ep)$ has all the desired properties.

The same argument as in Proposition \ref{prop eq CT} next implies that $(b_\ep,a\calU +(1-a)\mu_\ep)$ is an ESS as well, provided $a>0$ is small enough.

\section{Proof of Theorem \ref{thm:social planner}}\label{sec proof SP}
\setcounter{equation}{0} \renewcommand{\theequation}{F.\arabic{equation}}

We prove first the result in the case $\rho=1$, and discuss  the case $\rho<1$ in Section \ref{sec rho<1}.

When $\rho=1$, one has $x_t= \chi_t$ for each $t$, and the welfare associated with a pair $(\sigma,\mu)$ reads
\[W(\sigma,\mu)= \int_{\Theta\times [0,1]} \left(1- \left|\theta-x\right|\right) d\mu(\theta,x)- c \int_{\Theta\times [0,1]} \sum_{k=0}^n {n\choose k}x^k (1-x)^{n-k} \beta(k) d\mu(\theta,x).\]
The first integral is the fraction of agents whose action match the state. Indeed, for given $(\theta_t,x_t)$, the fraction of agents whose action \emph{mismatch} the current state $\theta_t$ is 
$\left| \theta_t-x_t\right|$.
The second integral is the cost incurred in acquiring information. 
\medskip

The proof is organized as follows. In Section \ref{sec bounds} , we provide a simple bound on Markov chains. Section \ref{sec def} contains the definition of $\sigma_\lambda$
and state some properties of the dynamics $g^\lambda_\theta$ induced by $\sigma_\lambda$. Section \ref{sec rho=1} concludes the proof for the case $\rho=1$.

\subsection{Preliminary bounds on Markov process}\label{sec bounds}

Let a (symmetric) strategy $\sigma$ be given, with dynamics $g_\theta$. For $0\leq x<y \leq 1$, define 
\[m(x,y):=\inf\left\{l\geq 0: g_1^l(x)\geq y\right\},\]
where $g_1^l$ is the $l$-th iterate of $g_1$ and $\inf\emptyset =+\infty$.

\begin{lemma}\label{lemm bounds}
Assume that $g_\theta$ is increasing over $[0,1]$, with $g_0<g_1$ on $[0,1]$.  Then for each $K<+\infty$ and each invariant measure $\mu\in \Delta(\Theta\times [0,1])$ for $\sigma$, one has 
\[\int_{\Theta\times [0,1]} \left| \theta -x \right| d\mu(\theta,x) \leq \lambda \left( K+m(0,1-K\lambda)\right).\]
\end{lemma}

\begin{proof}
Fix $K$ and $\mu$. The conclusion is trivial if $m(0,1-K\lambda)=+\infty$, so let's assume that $m:=m(0,1-K\lambda)<+\infty$. We introduce an auxiliary process $(y_t)$ defined by 
$y_0=x_0$ and 
\[
y_{t}=\begin{cases}
g_{\theta_{t}}\left(y_{t-1}\right), & \text{if }\theta_{t}=\theta_{t-1}\\
1-\theta_{t} & \text{if }\theta_{t}\neq\theta_{t-1}.
\end{cases}
\]
That is, $(y_t)$ follows the same dynamics as $(x_t)$, except that $(y_t)$ swings to the 'wrong' consensus whenever the state changes. 

The sequence $(\theta_t,y_t)$ has a countable support in $\Theta\times [0,1]$, which consists 
of the pairs $z_{\theta,l}= (\theta,g^l(1-\theta))$ for $\theta\in \Theta$ and $l\geq 0$. It also has  a unique invariant measure $\nu$, given by
\[\nu(z_{\theta,l})= (1-\lambda) \nu(z_{\theta,l-1})= (1-\lambda)^l \nu(z_{\theta,0})\mbox{ for each } l\geq 0\mbox{ and } \theta\in \Theta.\]
Using the symmetry of $\mu$, it follows that $\nu(z_{\theta,0})= \frac12 \lambda$. 

Since $g_\theta$ is increasing, one has $\left| y_t-\theta_t\right |\geq \left| x_t-\theta_t\right |$ for each $t$ and therefore, 
\[\int_{\Theta\times [0,1]} |\theta-x| d\mu(\theta,x) \leq  \int_{\Theta\times [0,1]} |\theta-y| d\nu(\theta,y).\]

On the other hand,
\begin{align*}
\int_{\Theta \times [0,1]} \left|\theta-y\right|d\nu\left(\theta,y\right) & =2\sum_{k\geq 0}\left|1-g^{k}\left(0\right)\right|\nu\left(y_{1,k}\right)\\
 & =\left(\sum_{k\geq 0}\lambda\left(1-\lambda\right)^{k}\left|1-g_1^{k}\left(0\right)\right|\right)\\
 & \leq\sum_{k<m}\lambda\left(1-\lambda\right)^{k}+K\lambda\sum_{k\geq m}\lambda\left(1-\lambda\right)^{k}\\
 & \leq \left(1-\left(1-\lambda\right)^{m}\right)+K\lambda\leq\lambda\left(K+m\right).
\end{align*}
\end{proof}
\bigskip

The bound of Lemma \ref{lemm bounds} is also useful to provide an upper bound of the  fraction of agents who buy information, as we next show.

\begin{lemma}\label{lemm bounds2}
For every invariant measure $\mu$ for $\sigma$, one has 
\[\int_{\Theta\times [0,1]} \sum_{k=0}^n {n\choose k} x^k (1-x)^{n-k} \beta(k) d\mu(\theta,x)\leq 2\beta(0) +(n+1)! \int_{\Theta\times [0,1]}|\theta-x| d\mu(\theta,x).\]
\end{lemma}

\begin{proof}
The result follows from the series of inequalities
\begin{eqnarray*}
\int_{\Theta\times [0,1]} \sum_{k=0}^n {n\choose k} x^k (1-x)^{n-k} \beta(k) d\mu(\theta,x) &\leq &2\beta(0) +\int_{\Theta\times [0,1]} \sum_{k=1}^{n-1} {n\choose k} x^k (1-x)^{n-k} \beta(k)d\mu(\theta,x)  \\
&\leq & 2\beta(0) +(n+1)! \int_{\Theta\times [0,1]} x(1-x) d\mu(\theta,x)  \\
&\leq & 2\beta(0) +(n+1)! \int_{\Theta\times [0,1]} \left|\theta-x\right| d\mu(\theta,x) .
\end{eqnarray*}
\end{proof}

\subsection{The strategy $\sigma_\lambda$}\label{sec def}
We introduce a family $(\sigma_\lambda)$ parametrized by $\lambda$. We will denote $g^\lambda_\theta$ the dynamics induced by $\sigma_\lambda$ and, later, use $m_\lambda(x,y)$ to denote the function $m(x,y)$ introduced in the previous section, whenever we refer to $\sigma_\lambda$. 

We assume for simplicity that $H_\theta$ is continuous at $\frac12$. If it is not, $H_\theta(\frac12)$ should be replaced by $\frac12\left(H_\theta(\frac12)+H_\theta(\frac12)_-\right)$ in what follows. Denote $h:=1-H_1(\frac12)= H_0(\frac12)$ and observe that $h>\frac12$. 

The probability $\beta^\lambda(k)$ of acquiring information increases linearly between $k=0$ and $k=\lfloor \frac{n}{2}\rfloor$, with $\beta^\lambda(0)= \frac{\lambda}{h}$ and $\beta^\lambda(\frac{n}{2})= 1$ (if $n$ is even). Formally, 
\[\beta^\lambda(k):= \frac{\lambda}{h} +\frac{2}{n} \left(1-\frac{\lambda}{h}\right)k\mbox{ for }k\leq \lfloor\frac{n}{2}\rfloor,\]
and $\beta^\lambda(n-k)= \beta^\lambda(k)$ for $k>\frac{n}{2}$. 

The decision $\alpha^\lambda(k,q)$ depends on $q$, but not on $k$, with $\alpha^\lambda(k,q)=1$ if $q>\frac12$, $\alpha^\lambda(k,q)=0$ if $q<\frac12$ and $\alpha^\lambda(k,\frac12)= \frac12$.

When agents do not acquire information (which implies $k\neq \frac{n}{2}$), they follow the majority action in their sample. Note that $\sigma_\lambda$ is symmetric.

Denote by $\phi_{\theta}^\lambda(k)$ the probability of playing action 1 in state $\theta$ with a sample $k$:

\begin{equation}\label{eq p}\phi_{\theta}^\lambda(k):=  \begin{cases} \beta^\lambda(k) \left(1-H_\theta\left(\frac12\right)\right) & \mbox{ if } k\leq \frac{n}{2} \\ \left(1-\beta^\lambda(k)\right)  +\beta^\lambda(k) \left(1-H_\theta\left(\frac12\right)\right)
\mbox{ if } k> \frac{n}{2} \end{cases}\end{equation}
With such notation, one has 
\[g_\theta^\lambda(x)= \sum_{k=0}^n{n\choose k}x^k (1-x)^{n-k}\phi_{\theta}^\lambda(k).\]
Standard computations show that 
\begin{equation}\label{derive}(g^\lambda_\theta)'(x)= n\sum_{k=0}^{n-1}{n-1 \choose k} x^k (1-x)^{n-1-k}\left(\phi^\lambda_{\theta}(k+1)-\phi^\lambda_{\theta}(k)\right),\end{equation}
and 
\[(g^\lambda_\theta)^{''}(x)= n(n-1)\sum_{k=0}^{n-2}{n-2 \choose k} x^k (1-x)^{n-2-k}\left(\left(\phi^\lambda_{\theta}(k+2)-\phi^\lambda_{\theta}(k+1)\right)-\left(\phi^\lambda_{\theta}(k+1)-\phi^\lambda_{\theta}(k)\right)\right).\]
From (\ref{eq p}), one has 
\[\phi^\lambda_{\theta}(k+1)-\phi^\lambda_{\theta}(k)=\begin{cases}
\left(1-H_\theta\left(\frac12\right)\right) \frac{2}{n} \left(1-\frac{\lambda}{h}\right) & \mbox { if } k+1\leq \frac{n}{2} \\
H_\theta\left(\frac12\right) \frac{2}{n} \left(1-\frac{\lambda}{h}\right) & \mbox { if } k\geq \frac{n}{2} \\
\frac{1}{n}\left(1-\frac{\lambda}{h}\right)& \mbox { if } k< \frac{n}{2} < k+1
\end{cases}\]

In particular, $g^\lambda_\theta$ is increasing (if $\lambda <h$). In addition, $\phi^\lambda_{\theta}(k+1)-\phi^\lambda_{\theta}(k)$ is piecewise constant as $k$ increases, and is easily seen to be non-increasing if $\theta=1$ and non-decreasing if $\theta=0$. Thus, $g^\lambda_1$ is concave and $g^\lambda_0$ is convex. 
\bigskip

As $\lambda$ decreases, $\phi^\lambda_{\theta}(k)$ decreases for $k\leq \frac{n}{2}$ and increases for $k\geq \frac{n}{2}$, with limit
\[\phi_{\theta}(k):= \left(1-H_\theta\left(\frac12\right)\right) \times \frac{2k}{n} \mbox{ for } k\leq \frac{n}{2}.\]
It follows that $(g^\lambda_\theta)$ and $((g^\lambda_\theta)')$ converge uniformly as $\lambda\to 0$, with limits $g_\theta$ and $g'_\theta$.

From (\ref{derive}) one has $g'_1(0)= 2h>1$ and $g'_0(0)= 2(1-h)<1$ therefore there is $a>0$  and $\ep>0$ such that 
\[g'_1(x) > 1+2a \mbox{ and } g'_0(x) < 1-2a\mbox{ for each } x\leq \ep.\]
By uniform convergence, there is  $\lambda_0< a h \ep$ such that 
\begin{equation}\label{der2}
(g^\lambda_1)'(x) > 1+2a \mbox{ and } (g^\lambda_0)'(x) < 1-2a\mbox{ for each } x\leq \ep.
\end{equation}
Let $\lambda <\lambda_0$. By (\ref{der2}), one has 
\[g^\lambda_1(x) \geq g^\lambda_1(0) + \ep (1+2a)\geq \ep +a\ep,\]
\[g_0^\lambda(\ep) \leq g_0^\lambda (0) +(1-2a)\ep = \lambda\frac{1-h}{h} +(1-2a)\ep \leq \ep -a \ep.\]
Finally, with $K:=\frac{1}{2ah}$, one has 
\begin{eqnarray*}
g_0^\lambda(K\lambda) &\leq & g^\lambda_0(0) +(1-a) K\lambda \\
&\leq & K\lambda +\lambda \left(\frac{1-h}{h} -a K\right) \\
&\leq & K\lambda -\lambda.
\end{eqnarray*}

Lemma \ref{lemm summ} collects these results.
\begin{lemma} \label{lemm summ}
There is $a,\ep,\lambda_0,K>0$ such that the following holds for each $\lambda <\lambda_0$:
\begin{enumerate}
\item $g_1^\lambda(x) \geq x + a\ep$ for all $x\in [\ep,1-\ep]$.
\item $g^\lambda_1(1-K\lambda) \geq 1-K\lambda +\lambda$.
\end{enumerate}
\end{lemma}

\begin{proof}
The second item holds since $g^\lambda_0(K\lambda)+g^\lambda_1(1-K\lambda )=1$ by symmetry of $g_\theta^\lambda$.

The first item holds since $g^\lambda_1$ is concave and since the inequality $g^\lambda_1(x)\geq x +ax$ holds for $x=\ep$ and $x=1-\ep$ (again using the symmetry of $g^\lambda_\theta$).
\end{proof}
\medskip

We henceforth that $a, \ep, \lambda$ and $K$  are as given in Lemma \ref{lemm summ}. 

\subsection{Conclusion: the case $\rho=1$}\label{sec rho=1} 

The conclusion of the theorem follows from the lemma below,  from the estimates in Section \ref{sec bounds}, and from the formula for $\beta^\lambda(0)$.

\begin{lemma}
There exists $A$ such that 
\[m_\lambda(0,1-K\lambda)\leq -A\ln \lambda \mbox{ for all } \lambda <\lambda_0.\]
\end{lemma}

\begin{proof}
Fix $\lambda$. One has 
\[m_\lambda(0,1-K\lambda)\leq m_\lambda(0,\ep)+m_\lambda(\ep,1-\ep) + m_\lambda(1-\ep, 1-K\lambda).\]
For each $x\leq \ep$ and $l$ such that $(g^\lambda_1)^l(x) \leq \ep$, one has 
\begin{eqnarray*}
(g^\lambda_1)^l(x) -(g^\lambda_1)^{l-1}(x) &\geq & (1+a) \left((g^\lambda_1)^{l-1}(x) -(g^\lambda_1)^{l-2}(x)\right)  \\
&\geq &(1+a)^{l-1} \left(g^\lambda_1(x)-x\right),
\end{eqnarray*}
 and therefore, $(g^\lambda_1)^l(x) \geq (1+a)^{l-1} \lambda$, since $g^\lambda_1(0)= \lambda$. This  implies $\displaystyle m_\lambda(0,\ep) \leq -\frac{\ln \lambda}{\ln (1+a)}$.

The same argument shows that $\displaystyle m(1-\ep,1-K\lambda) \leq \frac{\ln \lambda}{\ln (1-a)}$. 

Finally, thanks to the first proof of Lemma \ref{lemm summ} again, one has $m(\ep,1-\ep) \leq \frac{1-2\ep}{a\ep}$. 

\end{proof}
\subsection{Conclusion: the case $\rho<1$}\label{sec rho<1}

For $\rho=1$, the steady-state fraction of agents who play the correct action is given by $\displaystyle I_0:= \int_{\Theta\times [0,1]} (1-|\theta-x|)d\mu(\theta,x)$. For $\rho<1$, this is no longer true since $\mu$ is the steady-state distribution of $(\theta_t,x_t)$, where the definition of $x_t$ accounts for the fact that agents sample from the (possibly) distant past. The correct formula involves the fraction 
\[\chi_{t+1}= \frac{1}{\rho}\left(x_{t+1} -(1-\rho) x_t\right)\]
of agents who play action 1 in period $t+1$.

Denote by $\mu^*$ the (steady-state) distribution of $(\theta_t,x_t,\theta_{t+1},x_{t+1})$ induced by $\mu$ over two consecutive periods. The correct formula is then given by 
\[I_1 :=\int \left(1-\left|\theta_{t+1}-\frac{1}{\rho} \left(x_{t+1}-(1-\rho)x_t\right)\right|\right) d\mu^*(\theta_t,x_t,\theta_{t+1},x_{t+1}),\]
where the integral is taken over $\Theta\times [0,1]\times \Theta\times [0,1]$.

Lemma \ref{lemm fin} compares $I_0$ and $I_1$.

\begin{lemma}\label{lemm fin}
One has $\displaystyle I_1\geq I_0 -2\frac{1-\rho}{\rho} \lambda$.
\end{lemma}

This shows that the arguments in the earlier sections also deliver the result in the case $\rho<1$.
\medskip\noindent

\begin{proof}
One has 
\begin{eqnarray*}
1-I_1 & = & \int \left\{ 1_{\theta_{t+1}=1} \left(1-\frac{1}{\rho}\left(x_{t+1}-(1-\rho)x_t\right)\right) +1_{\theta_{t+1}=0}\times \frac{1}{\rho} \left(x_{t+1}-(1-\rho)x_t\right)\right\}d \mu^*(\theta_t,x_t,\theta_{t+1},x_{t+1}) \\
&=& \frac{1}{\rho} \int |\theta_{t+1}-x_{t+1}| d\mu^*  - \frac{1-\rho}{\rho}\int \left\{ 1_{\theta_{t+1}=1}(1-x_t) +1_{\theta_{t+1}=0}x_t \right\}d \mu^*.
\end{eqnarray*}

On the other hand, and because the state changes with probability $\lambda$ in each period, one has 
\[\int  1_{\theta_{t+1}=1}(1-x_t) d \mu^*  \leq  \int 1_{\theta_{t}=1}(1-x_t) d \mu^* +\lambda\]
and \[\int 1_{\theta_{t+1}=0}x_t d \mu^* \leq \int 1_{\theta_{t}=0}x_t d \mu^* +\lambda.\]
Substituting in the previous equality, we obtain
\[1-I_1 \leq  \frac{1}{\rho} \int |\theta_{t+1}-x_{t+1}| d\mu^* -\frac{1-\rho}{\rho}\int |\theta_t-x_t| d\mu^* +2\lambda\frac{1-\rho}{\rho}.\]
Since the marginals of $\mu^*$ on $(\theta_t,x_t)$ and $(\theta_{t+1}, x_{t+1})$ are equal to $\mu$, it follows that 
\[1-I_1 \leq \int_{\Theta\times [0,1]} |\theta-x| d\mu(\theta,x) +2\lambda  = 1-I_0 +2\lambda\frac{1-\rho}{\rho}.\]
\end{proof}

\end{appendices}

 \end{document}